\documentclass[11pt]{article}
\pdfoutput=1
\usepackage{etex}
\usepackage[margin=1in]{geometry}
\usepackage[latin2]{inputenc}
\usepackage[english]{babel}
\usepackage[all=normal,paragraphs=tight,bibliography=tight]{savetrees}
\usepackage{tikz}

\usepackage{graphicx}

\usepackage{multicol}

\usepackage{bbm}
\usepackage{mathrsfs}
\usepackage{latexsym}
\usepackage{lmodern}

\usepackage{amsmath,amsfonts,amssymb,amsthm} % AMS packages
\usepackage{mathtools}
\usepackage{thmtools}
\usepackage{thm-restate}
\usepackage{enumerate}

\newtheorem{theorem}{Theorem}[section]
\newtheorem{lemma}[theorem]{Lemma}
\newtheorem{claim}[theorem]{Claim}
\newtheorem{corollary}[theorem]{Corollary}
\newtheorem{proposition}[theorem]{Proposition}

\newtheorem{conjecture}{Conjecture}
\theoremstyle{definition}
\newtheorem{definition}[theorem]{Definition}

\newcommand{\retheorem}{theorem}

\usepackage{graphicx}

\usepackage{wrapfig}
\usepackage{booktabs} % Better tables
\usepackage{longtable}
\usepackage{float}

\usepackage{enumerate}
\usepackage{xspace}
\usepackage{comment}
\usepackage{ifthen}
\usepackage[noadjust,nospace]{cite} % Sort citations.
\usepackage{ellipsis}
\usepackage{fixltx2e}
\usepackage{url}
\usepackage{scrtime}

\usepackage{microtype} % Prettifies pdf output. Uncomment if you have trouble with this. Will also reduce page count.

\usepackage[draft,author=]{fixme}
\fxsetup{theme=color}

\newcommand{\FPT}{\ensuremath{\text{FPT}}\xspace}

\newcommand{\W}[1]{\ensuremath{\text{W}[#1]}}

\newcommand{\N}{\ensuremath{\mathbb{N}}\xspace} % Natural numbers
\newcommand{\R}{\ensuremath{\mathbb{R}}\xspace}	% Real numbers

{\bfseries \upshape}{\itshape}
{\bfseries \upshape}{\itshape}

\newcommand{\mc}{\mathcal}

\newcommand{\pr}[3]{M_{#1}(#2,#3)}
\newcommand{\prg}[4]{M^{#4}_{#1}(#2,#3)}
\newcommand{\cl}[2]{{\rm{cl}}_{#1}(#2)}

\newcommand{\Problem}[1]{\textsc{#1}\xspace}

\newcommand{\DS}{\Problem{Dominating Set}}

\DeclareMathOperator{\wcol}{wcol}
\DeclareMathOperator{\col}{col}
\DeclareMathOperator{\coim}{coim}

\DeclareMathOperator{\rad}{radius} % radius
\DeclareMathOperator{\den}{density} % density
\newcommand{\lexprod}{\ensuremath{\mathbin{{}\bullet{}}}}

\def\grad_#1{\nabla\!_{#1}}
\def\topgrad_#1{\widetilde \nabla\!_#1}

\renewcommand{\leq}{\leqslant}

\renewcommand{\geq}{\geqslant}

\newcommand{\widthm}[1]{\ensuremath{\mathop\mathbf{#1}}\xspace}

\newcommand{\width}{\widthm{width}}

\newcommand{\nab}{\mathop{\triangledown}}
\newcommand{\topnab}{\mathop{\tilde \triangledown}}
\newcommand{\dist}{\ensuremath{\text{dist}}}

\newcommand{\numcliques}{\ensuremath{\mathop{\#\omega}}}
\newcommand*{\ie}{i.e.\@\xspace}

\makeatletter
\newcommand*{\etc}{%
    \@ifnextchar{.}%
        {etc}%
        {etc.\@\xspace}%
}
\makeatother
\newif\ifshort

\newif\ifappend

\newcommand{\omitted}{%
    \ifappend%
    \else%
        \textup{[$\star$]}
    \fi
}

\usepackage{framed}
\usepackage{ctable}

\usepackage{xargs}
\usepackage{framed}
\renewenvironmentx{leftbar}[2][1=0.5pt, 2=5pt]{% 
  \def\FrameCommand{\vrule width #1 \hspace{#2}}\MakeFramed {\advance\hsize-\width \FrameRestore}}%
{\endMakeFramed}
\newenvironment{draft}[1]{%
                \color{gray}%
                \begin{leftbar}%                
                \begin{center}\emph{Draft:} {#1}\end{center}
            }{
                \end{leftbar}%
             }

\def\compxy(#1){\mathcal{U}_{#1}}

\def\Nesetril{Ne\v{s}et\v{r}il\xspace}
\def\Dvorak{Dvo\v{r}\'{a}k\xspace}

\newcommand{\ourtitlethanks}
{The research of P.\ Drange, F.\ Fomin, M.\ Pilipczuk, and M.\ Pilipczuk
leading to these results has received funding from the European
Research Council under the European Union's Seventh Framework
Programme (FP/2007-2013) / ERC Grant Agreement n.~267959.
Mi. Pilipczuk currently holds a post-doc position at Warsaw Centre of Mathematics and Computer Science, and is supported by the Foundation for Polish Science.
F.\ Reidl, F.\ Villaamil, and S.\ Sikdar are supported by DFG-Project RO
927/13-1 ``Pragmatic Parameterized Algorithms''.
M.\ Dregi and D.\ Lokshtanov are supported by the BeHard grant under the
recruitment programme of the of Bergen Research Foundation.
S.\ Saurabh is supported by PARAPPROX, ERC starting grant no.~306992.}

\newcommand{\eps}{{\varepsilon}}

\newcommand{\Oh}{{\mathcal{O}}}

\newcommand{\ds}{\mathbf{ds}}

\def\domcoresize{f_{\text{coresize}}}

\newcommand{\iftenpages}[2]{\ifthenelse{\equal{\tenpages}{yes}}{#1}{#2}}

\newcommand{\onlyfullversion}[1]{\iftenpages{}{#1}}

\newcommand{\tenpages}{no}

\shortfalse

\title{Kernelization and Sparseness: the case of {\sc{Dominating Set}}\thanks{\ourtitlethanks}}
\author{
  P{\r a}l Gr{\o}n{\r a}s Drange
  \thanks{
    Department of Informatics, University of Bergen, Norway, \texttt{\{pal.drange,markus.dregi,fomin,daniello\}@ii.uib.no}.
  }
  \and
  Markus S. Dregi$^\dagger$
  \and
  Fedor V. Fomin$^\dagger$
  \and
  Stephan Kreutzer%
  \thanks{
    Institute of Software Technology and Theoretical Computer Science, Technische Universit\"at Berlin, Germany, \texttt{\{stephan.kreutzer,sebastian.siebertz\}@tu-berlin.de}.
  }
  \and
  Daniel Lokshtanov$^\dagger$
  \and
  Marcin Pilipczuk
  \thanks{
    Institute of Informatics, University of Warsaw, Poland, \texttt{\{marcin,michal\}.pilipczuk@mimuw.edu.pl}.
  }
  \and
  Micha\l{} Pilipczuk$^\S$
  \and
  Felix Reidl
  \thanks{
    Theoretical Computer Science, Department of Computer Science, RWTH Aachen University, Germany, \texttt{\{reidl,fernando.sanchez,sikdar\}@cs.rwth-aachen.de}
  }
  \and
  Fernando S\'{a}nchez Villaamil$^{\P}$
  \and 
  Saket Saurabh\thanks{
    Institute of Mathematical Sciences, India \texttt{saket@imsc.res.in} and
    Department of Informatics, University of Bergen, Norway, \texttt{saket.saurabh@ii.uib.no}.
  }
  \and
  Sebastian Siebertz$^\ddagger$
  \and
  Somnath Sikdar$^{\P}$
}
\date{\today}

\def\cqedsymbol{\ifmmode$\lrcorner$\else{\unskip\nobreak\hfil
\penalty50\hskip1em\null\nobreak\hfil$\lrcorner$
\parfillskip=0pt\finalhyphendemerits=0\endgraf}\fi} 

\newcommand{\cqed}{\renewcommand{\qed}{\cqedsymbol}}

\newcommand{\probCDS}{\textsc{Connected Dominating Set}}

\newcommand{\probMotif}{\textsc{Graph Motif}}
\newcommand{\compassunless}{$\mathrm{NP} \subseteq \mathrm{coNP} / \mathrm{poly}$}
\newcommand{\CDSG}[1]{G^{\textrm{cds}}_{#1}}
\newcommand{\ColG}[1]{H^{\textrm{col}}_{#1}}

\newcommand{\defparproblem}[4]{
  \vspace{1mm}
\noindent\fbox{
  \begin{minipage}{0.97\textwidth}
  \begin{tabular*}{0.97\textwidth}{@{\extracolsep{\fill}}lr} #1 & {\bf{Parameter:}} #3 \\ \end{tabular*}
  {\bf{Input:}} #2  \\
  {\bf{Question:}} #4
  \end{minipage}
  }
  \vspace{1mm}
}

\begin{document}
\maketitle

\begin{abstract}
We prove that for every positive integer $r$ and for every graph class $\mc G$ of bounded expansion, the {\sc{$r$-Dominating Set}} problem admits a linear kernel on graphs from $\mc G$. Moreover, when $\mc G$ is only assumed to be nowhere dense, then we give an almost linear kernel on $\mc G$ for the classic {\sc{Dominating Set}} problem, i.e., for the case $r=1$. These results generalize a line of previous research on finding linear kernels for {\sc{Dominating Set}} and {\sc{$r$-Dominating Set}}~\cite{AlberFN04,BodlaenderFLPST09,FominLST10,FominLST12,FominLST13}. However, the approach taken in this work, which is based on the theory of sparse graphs, is radically different and conceptually much simpler than the previous approaches.

We complement our findings by showing that for the closely related {\sc{Connected
Dominating Set}} problem, the existence of such kernelization
algorithms is unlikely, even though the problem is known to admit a
linear kernel on $H$-topological-minor-free graphs~\cite{FominLST13}. Also, we prove that for any somewhere dense class $\mc G$, there is some $r$ for which {\sc{$r$-Dominating Set}} is W[$2$]-hard on $\mc G$. Thus, our results fall short of proving a sharp dichotomy for the parameterized complexity of {\sc{$r$-Dominating Set}} on subgraph-monotone graph classes: we conjecture that the border of tractability lies exactly between nowhere dense and somewhere dense graph classes.

\end{abstract}
\newpage

\newcommand{\Wt}

\section{Introduction}
\label{sec:intro}

\paragraph*{Domination and kernelization.} In the classic {\sc{Dominating Set}} problem, 
given a graph $G$ and an integer $k$, we are asked to determine the existence of a subset 
$D\subseteq V(G)$ of size at most $k$ such that every vertex $u\in V(G)$ is {\em{dominated}} 
by $D$: either $u$ belongs to $D$ itself, or it has a neighbor that belongs to $D$. The {\sc{$r$-Dominating Set}} problem, for a positive integer $r$, is a generalization where each vertex of $D$ dominates all the vertices at distance at most $r$ from it; by setting $r=1$ we obtain the original problem.
{\sc{Dominating Set}} is NP-hard and remains so even in very restricted settings, e.g. on planar 
graphs of maximum degree~$3$ (cf.\ [GT2] in Garey and Johnson~\cite{GareyJ79}). The complexity 
of {\sc{Dominating Set}} and {\sc{$r$-Dominating Set}} was studied intensively under different algorithmic frameworks, most 
importantly from the points of view of approximation and of parameterized complexity. 
In this work we are interested in the latter paradigm.

{\sc{Dominating Set}} parameterized by the target size $k$ plays a central role in 
parameterized complexity as it is a predominant example of a $\W{2}$-complete problem. 
Recall that the main focus in parameterized complexity is on designing {\em{fixed-parameter}} 
algorithms, or shortly {\em{FPT}} algorithms, whose running time on an instance of 
size $n$ and parameter $k$ has to be bounded by 
$f(k)\cdot n^c$ for some computable function $f$ and constant $c$. Downey and Fellows 
introduced a hierarchy of parameterized complexity classes 
$\FPT\subseteq \W{1}\subseteq \W{2}\subseteq \ldots$ that is believed to be strict, 
see~\cite{DowneyF99,FlumGrohebook}. As {\sc{Dominating Set}} is $\W{2}$-complete in 
general, we do not expect it to be solvable in FPT time.

However, it turns out that various restrictions on the input graph lead to 
robust tractability of {\sc{Dominating Set}}. Out of these, a particularly fruitful 
line of research concerned investigation of the complexity of the problem in sparse 
graph classes, like planar graphs, graphs of bounded genus, or graphs excluding some 
fixed graph $H$ as a minor. In these classes we can even go one step further than 
just showing fixed-parameter tractability: It is possible to design a linear 
kernel for the problem. Formally, a {\em{kernelization algorithm}} (or a {\em{kernel}}) 
is a polynomial-time preprocessing procedure that given an instance $(I,k)$ of a 
parameterized problem outputs another instance $(I',k')$ of the same problem which 
is equivalent to $(I,k)$, but whose total size $|I'|+k'$ is bounded by $f(k)$ for 
some computable function $f$, called the {\em{size}} of the kernel. If~$f$ is polynomial (resp.\ linear),
then such an algorithm is called a {\em{polynomial}}
(resp.\ {\em{linear}}) {\em{kernel}}. Note that the existence of such a kernelization 
algorithm immediately implies that the problem can be solved by a very efficient 
fixed-parameter algorithm: after applying kernelization, any brute-force search or 
more clever algorithm runs in time bounded by a function of~$k$ only.

The quest for small kernels for {\sc{Dominating Set}} on sparse graph classes began 
with the groundbreaking work of Alber et al.~\cite{AlberFN04}, who showed the first 
linear kernel for the problem on planar graphs. This work also introduced the 
concept of a {\em{region decomposition}}, which proved to be a crucial tool for 
constructing linear kernels for other problems on planar graphs later on. Another 
important step was the work of Alon and Gutner~\cite{AlonG08,Gutner09}, who gave 
an $\Oh(k^c)$ kernel for the problem on $H$-topological-minor free graphs, where~$c$
depends on~$H$ only. Moreover, if $H=K_{3,h}$ for some~$h$, then the 
size of the kernel is actually linear. This led Alon and Gutner to pose the 
following excellent question: Can one characterize the families of graphs where 
{\sc{Dominating Set}} admits a linear kernel? 

The research program sketched by the works of Alber et al.~\cite{AlberFN04} and Alon 
and Gutner~\cite{AlonG08,Gutner09} turned out to be one of particularly fruitful 
directions in parameterized complexity in recent years, and eventually led to the 
discovery of new and deep techniques. In particular, linear kernels for 
{\sc{Dominating Set}} have been given for bounded genus graphs~\cite{BodlaenderFLPST09}, 
apex-minor-free graphs~\cite{FominLST10}, $H$-minor-free graphs~\cite{FominLST12}, and 
most recently $H$-topological-minor-free graphs~\cite{FominLST13}. In all these results, 
the notion of {\em{bidimensionality}} plays the central role. Using variants of the 
Grid Minor Theorem, it is possible to understand well the connections between the minimum 
possible size of a dominating set in a graph and its treewidth. The considered graph 
classes also admit powerful decomposition theorems that follow from the Graph Minors 
project of Robertson and Seymour~\cite{gm16}, or the recent work of Grohe and 
Marx~\cite{GroheM12} on excluding~$H$ as a topological minor. The combination of these 
tools provides a robust base for a structural analysis of the input instance, which 
leads to identifying {\em{protrusions}}: large portions of the graph that have constant 
treewidth and small interaction with other vertices, and hence can be efficiently 
replaced by smaller gadgets. The protrusion approach, while originating essentially 
in the work on the {\sc{Dominating Set}} problem, turned out to be a versatile tool 
for finding efficient preprocessing routines for a much wider class of problems. 
In particular, the {\em{meta-kernelization}} framework of Bodlaender et al.~\cite{BodlaenderFLPST09}, 
further refined by Fomin et al.~\cite{FominLST10}, describes how a combination of 
bidimensional and finite-state properties of a generic problem leads to the
construction of linear kernels on bounded genus and $H$-minor-free graphs. 

Beyond the current frontier of $H$-topological-minor-free graphs~\cite{FominLST13}, 
kernelization of {\sc{Dominating Set}} was studied in graphs of bounded degeneracy. 
Recall that a graph is called {\em{$d$-degenerate}} if every subgraph contains 
a vertex of degree at most~$d$. Philip et al.~\cite{PhilipRS12} obtained a kernel 
of size $\Oh(k^{(d+1)^2})$ on $d$-degenerate graphs for constant~$d$, and more 
generally a kernel of size $\Oh(k^{\max(i^2,j^2)})$ on graphs excluding the complete bipartite graph 
$K_{i,j}$ as a subgraph. However, as proved by Cygan et al.~\cite{CyganGH13}, 
the exponent of the size of the kernel needs to increase with~$d$ at least quadratically: 
the existence of an $\Oh(k^{(d-1)(d-3)-\eps})$ kernel for any $\eps>0$ would imply 
that $\textrm{NP}\subseteq \textrm{coNP}/\textrm{poly}$. Thus, in these classes the 
existence of a linear kernel is unlikely.

As far as {\sc{$r$-Dominating Set}} is concerned, the current most general result gives a linear kernel for any apex-minor-free class~\cite{FominLST10}, and follows from a general protrusion machinery. The techniques used by Fomin et al.~\cite{FominLST12,FominLST13} for $H$-(topological)-minor-free classes are tailored to the classic {\sc{Dominating Set}} problem, and do not carry over to an arbitrary radius $r$. Therefore, up to this point the existence of linear kernels for {\sc{$r$-Dominating Set}} on $H$-(topological)-minor-free classes was open.

\paragraph*{Sparsity.} The concept of sparsity has been recently the subject of intensive 
study both from the point of view of pure graph theory and of computer science. 
In particular, the notions of graph classes of {\em{bounded expansion}} and {\em{nowhere dense}} 
graph classes have been introduced by \Nesetril and Ossona de Mendez. The main idea 
behind these models is to establish an abstract notion of sparsity based on known 
properties of well-studied sparse graph classes, e.g.\ $H$-minor-free graphs, and 
to develop tools for combinatorial analysis of sparse graphs based only on this 
abstract notion. We refer to the book of \Nesetril and Ossona de Mendez~\cite{Sparsity} 
for an introduction to the topic. 

Intuitively, a graph class~$\mathcal{G}$ has {\em{bounded expansion}} if any minor 
obtained by contracting disjoint subgraphs of radius at most~$r$ is 
$d_r$-degenerate, for some constant $d_r$. 
%even after performing contractions of any selection of disjoint subgraphs 
%of radius at most $r$, for every $r\geq 0$. 
Thus, this property can be thought of as strengthened degeneracy that persists 
after ``local'' minor operations. The notion of a {\em{nowhere dense}} graph class 
is a further relaxation of this concept\onlyfullversion{; we refer to Definition~\ref{def:nd} for a 
formal definition}. In particular, every graph class $\mathcal{G}$ that has bounded 
expansion is also nowhere dense, and all the aforementioned classes on which the 
existence of a linear kernel for {\sc{Dominating Set}} has been established (planar, 
bounded genus, $H$-minor-free, $H$-topological-minor-free) have bounded expansion.

From the point of view of theoretical computer science, of particular importance is the 
program of establishing fixed-parameter tractability of model checking First Order logic 
on sparse graphs. A long line of work resulted in FPT algorithms for model checking 
First Order formulae on more and more general classes of sparse 
graphs~\cite{DawarGK07, DvorakKT13, FlumG01, FrickG01, GroheKS14, Seese96}, similarly 
to the story of kernelization of {\sc{Dominating Set}}. Finally, FPT algorithms 
for the problem have been given for graph classes of bounded expansion by 
\Dvorak et al.~\cite{DvorakKT13}, and very recently for nowhere dense graph classes 
by Grohe et al.~\cite{GroheKS14}. This is the ultimate limit of this program: as proven 
in~\cite{DvorakKT13}, for any class $\mathcal{G}$ that is not nowhere dense 
(is {\em{somewhere dense}}) and is closed under taking subgraphs, model checking 
First Order formulae on $\mathcal{G}$ is not fixed-parameter tractable (unless $\FPT=\W{1}$).

Fixed-parameter tractability of {\sc{$r$-Dominating Set}} on nowhere dense graph classes 
follows immediately from the result of Grohe et al.~\cite{GroheKS14}, since the 
problem is definable in First Order logic (for each constant $r$). However, an explicit algorithm was given 
earlier by Dawar and Kreutzer~\cite{DawarK09}.

\medskip

To summarize, we would like to stress that \textsc{Dominating Set} has repeatedly served
as a trigger for developing new techniques in parameterized complexity:
the subexponential algorithm on planar graphs~\cite{AlberBFKN02} lead to the theory of bidimensionality;
the kernelization algorithm on planar graphs~\cite{AlberFN04} initiated meta-theorems and protrusion-based
techniques on planar graphs and beyond, which were further refined by techniques developed
for graphs with excluded topological minor~\cite{FominLST13}; and, last but not least, 
the work on \textsc{$r$-Dominating Set} in nowhere-dense graphs~\cite{DawarK09} led to generic First Order logic
results on sparse classes of graphs.
Therefore, we believe that understanding the kernelization status of \textsc{Dominating Set} and \textsc{$r$-Dominating Set}
in sparse graph classes may again lead to very fruitful developments.

%In this work we make an abrupt turn in the approach to kernelization algorithms 
%for {\sc{Dominating Set}} on sparse graphs. We namely explore the topic in 
%the context of the theory of {\em{bounded expansion}} and {\em{nowhere dense}} graph classes. 

\paragraph*{Kernelization results.} In this work we prove that having bounded expansion or 
being nowhere dense is sufficient for a graph class to admit an (almost) linear 
kernel for {\sc{Dominating Set}}. Henceforth, for a graph~$G$, we let~$\ds(G)$ denote the 
minimum size of a dominating set of~$G$.

\begin{restatable}{\retheorem}{restatemainbe}
  \label{thm:main-be} Let $\mathcal{G}$ be a graph class of bounded
  expansion.  There exists a polynomial-time algorithm that given a
  graph $G \in \mathcal{G}$ and an integer $k$, either correctly
  concludes that $\ds(G) > k$ or finds a subset of vertices $Y
  \subseteq V(G)$ of size $\Oh(k)$ with the property that $\ds(G) \leq
  k$ if and only if $\ds(G[Y]) \leq k$.
\end{restatable}
\begin{restatable}{\retheorem}{restatemainnd}
  \label{thm:main-nd}
  Let $\mathcal{G}$ be a nowhere dense graph class and let $\eps > 0$
  be a real number.  There exists a polynomial-time algorithm that
  given a graph $G \in \mathcal{G}$ and an integer $k$, either
  correctly concludes that $\ds(G) > k$ or finds a subset of vertices
  $Y \subseteq V(G)$ of size $\Oh(k^{1+\eps})$ with the property that
  $\ds(G) \leq k$ if and only if $\ds(G[Y]) \leq k$.
\end{restatable}

In both cases, to obtain a kernel we apply the 
algorithm given by Theorem~\ref{thm:main-be} or~\ref{thm:main-nd},
and then either provide a trivial no-instance (in case the algorithm 
concluded that $\ds(G)>k$), or we output $(G[Y],k)$. This immediately yields the following:

\begin{corollary}
For every hereditary graph class~$\mathcal{G}$ with bounded expansion, 
{\sc{Dominating Set}} admits a kernel of size $\Oh(k)$ on graphs from~$\mathcal{G}$. 
For every hereditary and nowhere dense graph class~$\mathcal{G}$ and every $\eps>0$, 
{\sc{Dominating Set}} admits a kernel of size $\Oh(k^{1+\eps})$ on graphs from~$\mathcal{G}$.
\end{corollary}

Note that we formally need to assume that the graph class $\mathcal{G}$ is hereditary 
(closed under taking induced subgraphs), in order to ensure that the 
output instance $(G[Y],k)$ is of the same problem as the input one. However, this is a purely formal problem: for any class $\mathcal{G}$ that either has bounded expansion or is nowhere dense, its closure under taking induced subgraphs also has this property, with exactly the same expansion parameters. So for the sake of kernelization we can always remain in the closure of $\mathcal{G}$ under taking induced subgraphs.

For {\sc{$r$-Dominating Set}}, for $r>1$, we can give a linear kernel for any graph class of bounded expansion.
Unfortunately, there is a technical subtlety that does not allow us to state the kernel in a nice form as in Theorems~\ref{thm:main-be} and~\ref{thm:main-nd}.
Instead, we can kernelize an annotated version of the problem, where only a given subset of vertices of $G$ needs to be dominated.
The annotated version can be reduced to the classic one by simple gadgeteering that, unfortunately, may lead to a slight increase in the bounded expansion
guarantees.

In the following, by $\ds_r(G)$ we denote the minimum size of an {\sc{$r$-Dominating Set}} in a graph $G$, while for $Z\subseteq V(G)$, $\ds_r(G,Z)$ denotes the minimum size of a {\em{$(Z,r)$-dominator}} in $G$, that is, a set $D\subseteq V(G)$ that $r$-dominates (i.e., is at distance at most $r$) every vertex of $Z$. 
For an integer $i$, by $\grad_i(G)$ we denote the maximum density of a minor of $G$ that is constructed by creating disjoint connected subgraphs
of radius at most $i$; see Section~\ref{ss:grads} for a formal definition.

\begin{restatable}{\retheorem}{restatemainbedominator}\label{thm:mainrbe-dominator}
Let $\mathcal{G}$ be a graph class of bounded
  expansion, and let $r$ be a positive integer.  There exists a polynomial-time algorithm that given a
  graph $G \in \mathcal{G}$ and an integer $k$, either correctly
  concludes that $\ds_r(G) > k$ or finds subsets of vertices $Z\subseteq W\subseteq V(G)$, where $|W|=\Oh(k)$, with the property that $\ds_r(G) \leq
  k$ if and only if $\ds_r(G[W],Z) \leq k$.
\end{restatable}

%By adding a simple gadget we can reduce the problem back to {\sc{$r$-Dominating Set}}.

\begin{restatable}{\retheorem}{restatemainbereduced}\label{thm:mainrbe-reduced}
Let $\mathcal{G}$ be a graph class of bounded
  expansion, and let $r$ be a positive integer.  There exists a polynomial-time algorithm that given a
  graph $G \in \mathcal{G}$ and an integer $k$, either correctly
  concludes that $\ds_r(G) > k$ or finds a graph $G'$ such that $|V(G')|=\Oh(k)$, $\grad_i(G')\leq \max(\grad_i(G)+1,2)$ for each nonnegative integer $i$, and $\ds_r(G)\leq k$ if and only if $\ds_r(G')\leq k+1$.
\end{restatable}

Although, formally speaking, Theorems~\ref{thm:mainrbe-dominator} and~\ref{thm:mainrbe-reduced} do not give linear kernels for {\sc{$r$-Dominating Set}} on $\mc G$, because in Theorems~\ref{thm:mainrbe-dominator} we reduce to a different problem, whereas the output graph of Theorem~\ref{thm:mainrbe-reduced} may not belong to $\mc G$ because of some simple gadget attached to it (that may slightly worsen the sparsity guarantees), arguably both these results are as good as a linear kernel compliant to the most restrictive definition. In fact, both of them give a polynomial compression algorithm into bitsize $\Oh(k\log k)$, which is indistinguishable from a linear kernel using current lower bounds techniques. Also, from the proof it is imminent that the problem used in Theorem~\ref{thm:mainrbe-dominator}, where only a subset of vertices needs to be dominated, is much more natural in our context. 

The obtained results strongly generalize the previous results on linear kernels 
for {\sc{Dominating Set}} on sparse graph classes~\cite{AlberFN04,BodlaenderFLPST09,FominLST10,FominLST12,FominLST13}, 
since all the graph classes considered in these results have bounded expansion. 
Moreover, by giving a linear kernel for {\sc{$r$-Dominating Set}} on any class $\mc G$ of bounded expansion, we obtain the same result for any $H$-minor-free or $H$-topological-minor-free class  as well. The existence of such kernels was not known before.

We see the main strength of our results in that they constitute an 
abrupt turn in the current approach to kernelization of {\sc{Dominating Set}} and {\sc{$r$-Dominating Set}}
on sparse graphs: the tools used to develop the new algorithms are radically 
different from all the previously applied techniques. Instead of investigating 
bidimensionality and treewidth, and relying on intricate decomposition theorems 
originating in the work on graph minors, our algorithms exploit only basic 
properties of bounded expansion and nowhere dense graph classes. As a result, 
this paper presents essentially self-contained proofs of all the stated kernelization results.
The only external facts that we use are basic properties of 
weak and centered colorings, and the constant-factor approximation algorithm for {\sc{$r$-Dominating Set}} 
of \Dvorak~\cite{Dvorak13}. All in all, the results show that only the 
combinatorial sparsity of a graph class is essential for designing (almost) 
linear kernels for {\sc{Dominating Set}}, and further topological constraints 
like excluding some (topological) minor are unnecessary.

\paragraph*{Lower bounds.} We complement our study by proving that for the closely related {\sc{Connected Dominating Set}} 
problem, where the sought dominating set $D$ is additionally required to induce 
a connected subgraph, the existence of even polynomial kernels for bounded expansion 
and nowhere dense graph classes is unlikely. More precisely, we prove the following result:

\begin{restatable}{\retheorem}{restatecondomlb}
  \label{thm:condom-lb}
  There exists a class of graphs~$\mathcal{G}$ of bounded expansion
  such that \probCDS{} does not admit a polynomial kernel when
  restricted to~$\mathcal{G}$, unless \compassunless{}. 
  Furthermore,~$\mathcal{G}$ is closed under taking subgraphs.
\end{restatable}

Up to this point, linear kernels for {\sc{Connected Dominating Set}} were given 
for the same family of sparse graph classes as for {\sc{Dominating Set}}: a linear 
kernel for the problem on $H$-topological-minor-free graphs was obtained by Fomin 
et al.~\cite{FominLST13}. Hence, classes of bounded expansion constitute the point 
where the kernelization complexity of both problems diverge: while {\sc{Dominating Set}} 
admits a linear kernel by Theorem~\ref{thm:main-be}, for {\sc{Connected Dominating Set}} 
even a polynomial kernel is unlikely by Theorem~\ref{thm:condom-lb}.
Our intuition about the phenomenon is as follows: 
the connectivity constraint has a completely different nature, and topological 
properties of the graph class become necessary to handle it efficiently. Indeed, a deeper examination of the proof of 
Theorem~\ref{thm:condom-lb} shows that we essentially exploit only the connectivity constraint to establish the lower bound.

Next, we show also that nowhere dense classes form the ultimate limit of parameterized tractability of {\sc{$r$-Dominating Set}}, similarly as it was the case for model checking First Order formulae.
\begin{restatable}{theorem}{wsomewheredense}\label{thm:w2-hardness}
For every somewhere dense graph class $\mc G$ that is closed under taking subgraphs, there exists an integer $r$ such that {\sc{$r$-Dominating Set}} is $\W{2}$-hard on graphs from $\mc G$.
\end{restatable}
Theorem~\ref{thm:w2-hardness} together with our kernelization results give rise to an interesting dichotomy conjecture about the parameterized complexity of {\sc{$r$-Dominating Set}} on graph classes closed under taking subgraphs. We expand this topic in Section~\ref{sec:conclusion}.

\begin{comment}Indeed, 
a deeper examination of our proofs shows that in the case of {\sc{Connected Dominating Set}} 
we are able to ``partially'' kernelize the domination constraint, whereas the proof of 
Theorem~\ref{thm:condom-lb} exploits only the connectivity constraint to establish the lower bound.
\end{comment}

\paragraph*{Our techniques.} As explained before, the techniques applied to prove
our kernelization results differ radically from tools used 
in the previous works~\cite{AlberFN04,BodlaenderFLPST09,FominLST10,FominLST12,FominLST13}. 
The main reason is that so far all the approaches were based on bidimensionality and 
decomposition theorems for graph classes with topological constraints, like 
$H$-(topological)-minor-free graphs. For bounded expansion and nowhere dense classes, 
there are no known global decomposition theorems. Bidimensional arguments also cease to work, 
since they are inextricably linked to surface embeddings of graphs, meaningless in the 
world of nowhere dense and bounded expansion graph classes.

The failure of known techniques, seemingly a large obstacle for our project,
actually came as a blessing as it forced us to search for the ``real'' reasons why 
{\sc{$r$-Dominating Set}} admits linear kernels on sparse graph classes.
Identifying the right tools enabled us to streamline the reasoning so that it 
is significantly simpler than the previous works. We now briefly describe the 
main approach for proving Theorem~\ref{thm:mainrbe-dominator}; Theorem~\ref{thm:main-be} actually follows as a by-product of this proof. The proof of 
Theorem~\ref{thm:main-nd} is technically more complicated. This is due to the fact that certain tools for bounded expansion graph classes, which simplify the reasoning significantly, cease to work in the more general setting of nowhere dense classes.

The first general idea is to kernelize the instance in two phases: Intuitively, in the first phase we 
reduce the number of {\em{dominatees}}, vertices whose domination is 
essential, and in the second phase we reduce the number of {\em{dominators}}, vertices 
that are sensible to use to dominate other vertices. In order to formalize this 
approach, we introduce the following notion: a subset $Z\subseteq V(G)$ 
is an {\em{$r$-domination core}} if every minimum-size subset $D\subseteq V(G)$ that $r$-dominates~$Z$ 
is guaranteed to $r$-dominate the whole graph. Hence, every vertex whose domination is 
identified as irrelevant can safely be removed from the domination core. In the 
first phase of the algorithm we find an $r$-domination core in the graph of size 
linear in the parameter~$k$, and in the second phase we reduce the number of 
vertices outside it. The first phase is the most difficult one, while the 
second is much simpler. However, it is the second phase where technical problems arise, due to which for $r>1$ we are not able to obtain a kernel that is an induced subgraph of the input graph. Intuitively, we can locate $\Oh(k)$ relevant dominatees and $\Oh(k)$ relevant dominators, but the remaining vertices may also be essential to preserve connections between dominators and dominatees. It is precisely this role that is problematic to reduce; note that the problem arises only for $r>1$.

The small domination core is found iteratively, by first taking $Z=V(G)$ and 
then removing vertices from~$Z$ one by one. Hence, the main difficulty is to find a 
vertex that can be safely removed from~$Z$; for the sake of this overview, we focus on the first iteration 
when $Z=V(G)$. The first step is to apply the approximation algorithm of \Dvorak for 
{\sc{$r$-Dominating Set}} on graphs of bounded expansion~\cite{Dvorak13}. This algorithm has the 
following very important feature: given a parameter~$k$, it either provides an $r$-dominating set 
of size $\Oh(k)$, or it outputs a proof that $\ds_r(G)>k$ in the form of a 
{\em{$2r$-scattered set}}~$S$ of size larger than~$k$; recall that~$S$ is $2r$-scattered if every 
two vertices of~$S$ are at distance more than~$2r$ from each other. The idea is to 
apply the algorithm of \Dvorak repeatedly: % not once, but several times:
In each iteration we either 
identify another approximate $r$-dominating set and remove it from the graph, or we find a 
large $2r$-scattered set in the remaining instance and terminate the iteration. However, for reasons that will become clear later, after each iteration we perform a ``closure'' step: we iteratively remove from the graph all the vertices that ``see'' a super-constant number of removed vertices at distance at most $3r$. As we 
work in a graph class of bounded expansion, it can be shown that this closure blows up the number of removed vertices only by a constant factor. Moreover, the whole process can be shown to terminate 
after a constant number of steps; the closure step is an important part of this reasoning.

Hence, we end up with the following structure in 
the graph: an $r$-dominating set $X\subseteq V(G)$ of size $\Oh(k)$, and
a set $S\subseteq V(G)\setminus X$ that is $2r$-scattered in $G-X$. Moreover, every vertex of $V(G)\setminus X$ sees only a constant number of vertices of $X$ within radius $3r$ of it (more precisely, we count only vertices reachable by paths of length at most $3r$ that do not intersect $X$ before the endpoint). By carefully selecting the parameters
of the approximation, we can ensure that $|S|>c|X|$ for as large a constant~$c$ as we like.

Having identified such a pair $(X,S)$, we partition $S$ into equivalence 
classes such that two vertices are equivalent when they have exactly the 
same $1$st, $2$nd, $3$rd, etc. neighborhoods in~$X$, up to radius $3r$. Using the sparseness of the graph class we work with, it can be shown that the number of such classes is bounded linearly in $|X|$; in fact, this argument is the crux of our approach. 

Recall that we assumed that $|S|>c\cdot |X|$ for a constant $c$ as large as we like. Hence, we can identify a class $\kappa$ that has a large number of vertices; more precisely, at least a constant that is as large as we like. Then we argue that any vertex of $\kappa$ is an irrelevant dominatee that can be removed from the $r$-domination core. The rationale is as follows: Vertices of $\kappa$ are equivalent from the point of view of $r$-dominating them ``via'' $X$, whereas $r$-dominating them not via $X$ is suboptimal due to their large number. In the latter argument it is crucial that every vertex of $\kappa$ sees only a constant number of vertices of $X$ at distance at most $3r$; this property was achieved by dint of the closure step when constructing $(X,S)$.

This reasoning can be applied as long as $|Z|>Ck$ for some constant~$C$, so we 
eventually compute an $r$-domination core of size linear in~$k$. To reduce the number of 
dominators, we again apply the neighborhood diversity argument. We partition 
the vertices of $V(G)\setminus Z$ into classes with respect to their $r$-neighborhoods 
in~$Z$, the number of these classes is linear in~$|Z|$, and it is safe to pick only one arbitrary vertex from each class as a relevant dominator. For Theorem~\ref{thm:main-be}, the graph induced by relevant dominatees and dominators is the required kernel. For Theorem~\ref{thm:mainrbe-dominator}, we moreover need to preserve connections between dominators and dominatees, which introduces technical problems. Essentially, we are able to reduce the number of vertices needed for connections to $\Oh(k)$, but we cannot require their domination in the kernel. As mentioned, Theorem~\ref{thm:mainrbe-reduced} follows from Theorem~\ref{thm:mainrbe-dominator} by simple gadgeteering.

When trying to generalize the result to the nowhere dense setting (Theorem~\ref{thm:main-nd}), the main difference is that the closure step does not carry over. The construction of $(X,S)$ can be performed similarly only for $r=1$, by using a different argument for why the procedure finishes after a constant number of iterations. Moreover, this also creates complications in the second phase, where an irrelevant dominatee is identified. We can still partition $V(G)\setminus X$ into an almost linear number of classes with respect to the neighborhoods in $X$, and hence find a class that has a large number of vertices of $S$, but it is no longer obvious that any vertex of such a class is an irrelevant dominatee. To overcome this issue, we create an auxiliary structure of a graph on the set of classes. We prove that this class graph is sparse by itself. Hence, using a potential argument we can find a class $\kappa$ where the number of vertices from $S$ is large compared to the number of classes with which $\kappa$ neighbors in the class graph. For such $\kappa$ it can be argued that any its vertex is an irrelevant dominatee. The analysis of the class graph is, however, very delicate; we are able to perform it only for $r=1$.

The proof of Theorem~\ref{thm:condom-lb} uses the technique of 
compositionality to refute the existence of a polynomial kernel, and is based 
on the kernelization hardness result for {\sc{Connected Dominating Set}} 
on $2$-degenerate graphs presented by Cygan et al.~\cite{bimber}. The output 
instances of the original construction of Cygan et al.~\cite{bimber} do not have 
bounded expansion, but after adding a number of new technical ideas the 
construction can be modified to ensure this property. 

The proof of Theorem~\ref{thm:w2-hardness} uses the same technical characterization of somewhere dense classes as \Dvorak et al.~\cite{DvorakKT13} in their proof of intractability of model checking First Order logic formulae: if a graph class $\mc G$ is somewhere dense and closed under taking subgraphs, then for some $r_0$ it contains $r_0$-subdivisions of all graphs. Using this, we give a parameterized reduction from {\sc{Set Cover}} that shows $\W{2}$-hardness of {\sc{$3r_0$-Dominating Set}} on $\mc G$.

\paragraph*{Relation with the previous version.} The previous version of this paper~\cite{previous} did not contain the result for {\sc{$r$-Dominating Set}} on bounded expansion graph classes. Originally, our methodology for {\sc{Dominating Set}} on bounded expansion and nowhere dense classes was very similar, but the new techniques (mostly the Closure Lemma, i.e. Lemma~\ref{lem:closure}) developed for {\sc{$r$-Dominating Set}} greatly simplified our original proof for {\sc{Dominating Set}} on bounded expansion classes. Consequently, the presentation of results for bounded expansion classes follows the new, simplified approach, whereas for nowhere dense classes we resort to our old methodology. Throughout the text we sometimes remark about the differences.

\iftenpages{%
\medskip
\noindent In this extended abstract we present a sketch of the proof of Theorem~\ref{thm:main-be}.
The proof of Theorem~\ref{thm:main-nd} follows the same approach, but in a few places is more technically
involved.
}{%
\paragraph*{Organization of the paper.} In Section~\ref{sec:preliminaries} we 
recall the most important definitions and facts about sparse graph classes. We also prove some auxiliary results that will be used 
later on. Section~\ref{sec:bnd-exp} contains the proof of Theorems~\ref{thm:main-be},~\ref{thm:mainrbe-dominator}, and~\ref{thm:mainrbe-reduced} --- 
the main results for bounded expansion classes. Section~\ref{sec:nowhere-dense} 
contains the proof of Theorem~\ref{thm:main-nd} --- the main result for nowhere dense 
classes. This section is also equipped with additional preliminaries needed in the nowhere dense setting. In Section~\ref{sec:condom} we present the lower bound for 
{\sc{Connected Dominating Set}}, i.e., Theorem~\ref{thm:condom-lb}, whereas Section~\ref{sec:w-hardness} contains the proof of Theorem~\ref{thm:w2-hardness}.
Section~\ref{sec:conclusion} contains concluding remarks and prospects for future 
work. Proofs of auxiliary facts (marked with $\star$) that are very easy and/or 
follow directly from known results have been deferred to Appendix~\ref{sec:omitted} 
in order not to distract the attention of the reader. Appendix~\ref{sec:fewneighbourhoods} contains the proof of Lemma~\ref{lem:fewneighbourhoods}. This lemma was essentially already used in the literature, but we give our own proof for completeness.}

\section{Preliminaries}
\label{sec:preliminaries}

\subsection{Notation}

\paragraph{Basic graph notation}
All graphs we consider are finite, simple, and undirected.  For a
graph~$G$, we denote by $|G| = |V(G)|$ the number of vertices and by
$||G|| = |E(G)|$ the number of edges in~$G$.  The \emph{density} of a
graph~$G$, denoted $\den(G)$ is defined as $\den(G) = {||G||} /
{|G|}$.
% Observe that the average degree in $G$, denoted $\avg(G)$, is equal
% to $2 \cdot \den(G)$.\todo{Do we use average degree?}
For an integer~$k \in \N$ we denote by $[k] = \{1, \dots, k\}$ the first~$k$
positive integers.

For a vertex~$v$ in a graph~$G$, we denote by $N_G(v)=\{u\colon uv\in
E(G)\}$ the \emph{open neighborhood} of~$v$ and by $N_G[v]=N_G(v)\cup
\{v\}$ the \emph{closed neighborhood} of~$v$ in~$G$.  These notions
can be naturally extended to sets of vertices $X \subseteq V(G)$ as
follows: $N_G[X] = \bigcup_{v \in X} N_G[v]$ and $N_G(X) = N_G[X]
\setminus X$.  If~$G$ is clear from the context, we omit the
subscripts.  Furthermore, we write~$N_X(v)$ to denote the neighborhood
of~$v$ restricted to~$X$, i.e., $N_X(v) = N_G(v)\cap X$, and refer to
it as the $X$-neighborhood of~$v$.  The \emph{degree} of a vertex $v \in
V(G)$ is the number of neighbors it has, i.e., $\deg(v) = |N(v)|$.
For a positive integer $r$ and $v\in V(G)$, by $N_G^r[v]$ we denote the ball of radius $r$ around $v$, i.e., the set of vertices of $G$ that are at distance at most $r$ from $v$. We say that $v$ {\em{$r$-dominates}} every vertex at distance at most $r$ from it, or just {\em{dominates}} if $r=1$.

The induced subgraph~$G[X]$ for $X \subseteq V(G)$ is the graph with vertex
set~$X$ and for $x_1,x_2 \in X$ we have that $x_1x_2\in E(G[X])$ if
and only if $x_1x_2\in E(G)$.  A graph $H = (V_H,E_H)$ is a subgraph
of $G = (V_G, E_G)$ if $V_H\subseteq V_G$ and $E_H \subseteq E(G[V_H])$.  We will
% sometimes abuse the notation and
say that~$H$ is a subgraph of~$G$ if~$H$ is isomorphic to a subgraph
of~$G$.  For a set of vertices $X \subseteq V(G)$, we write $G - X$ to
denote the induced subgraph $G[V(G) \setminus X]$.
\begin{comment}
  \begin{draft}{Put this in prelims!}
    \begin{itemize}
    \item By $\coim f$ we denote the co-image of the function $f$.
    \item Denote by $\numcliques(\cdot)$ the number of complete
      subgraphs in a graph.
    \item $\grad_r(G \nab s) \leq \grad_{8sr}(G)$
    \item Theorem 7.10 in Sparsity: $\chi_p(G) \leq \wcol_{2^{p-2}}(G)$.
    \item Need notation $\mc G_{=\ell} := \{ G \in \mc G : |G| = \ell
      \}$
    \end{itemize}
  \end{draft}
\end{comment}

Given a graph~$G$, positive integer $r$, and two vertex subsets $D,Z \subseteq V(G)$, we say that~$D$
is a \emph{$(Z,r)$-dominator} if~$D$ $r$-dominates~$Z$ in~$G$, that is, every
vertex $z \in Z \setminus D$ is at distance at most $r$ from some vertex of~$D$. In case $r=1$ we simply say that $D$ is a {\em{$Z$-dominator}}.
We denote by $\ds_r(G,Z)$ the size of a smallest $(Z,r)$-dominator of~$G$.  By $\ds_r(G)$
we mean $\ds_r(G,V(G))$, i.e., the size of a smallest $r$-dominating set
in~$G$. The subscript $r$ is omitted when $r=1$.

A set $S\subseteq V(G)$ is $\ell$-scattered in~$G$ if for
every pair of distinct vertices $s_1,s_2\in S$, the distance
between~$s_1$ and~$s_2$ is at least~$\ell+1$, i.e., any path
from~$v_1$ to~$v_2$ has at least~$\ell$ internal vertices.  Note that
if there is a~$2r$-scattered set~$S$ of size~$k$, then any $r$-dominating
set of~$G$ must have size at least~$k$, since every vertex of~$G$ can
$r$-dominate at most one vertex of~$S$.  Hence, we call a $2r$-scattered
set of size~$k+1$ an obstruction for an $r$-dominating set of size~$k$.

A \emph{clique} in a graph is a subset of pairwise adjacent vertices.
We write~$\omega(G)$ to denote the \emph{clique number} of a graph
$G$, i.e., the size of a maximum clique in~$G$.  We write
$\numcliques(G)$ to be the total number of cliques in~$G$.  By~$K_c$
we denote the complete graph on~$c$ vertices,
% (called sometimes also a clique)
and by $K_{c_1,c_2}$ we denote the complete bipartite graph
% (called also a biclique)
with the sides of the bipartition of sizes~$c_1$ and~$c_2$,
respectively.

The \emph{radius} of a graph~$G$, denoted $\rad(G)$ is the minimum
integer~$r$ for which there exists a vertex~$v \in V(G)$ (a
\emph{center}) such that every vertex in~$V(G)$ is within distance at
most~$r$ from~$v$.

\paragraph{Minors and minor operations}
For an edge $e=uv$ in a graph $G$, the graph $G/e$ is the graph
obtained from \emph{contracting}~$e$, i.e., we replace the
vertices~$u$ and~$v$ with a vertex $w_{uv}$ that is adjacent to every
vertex of $N_G(\{u,v\})$ in $G/e$.  If $S \subseteq V(G)$ is a set of
vertices such that $G[S]$ is connected, we let $G/S$ denote the graph
obtained from $G$ by contracting $S$ to a single vertex.  That is,
$G/S$ is the graph obtained from deleting~$S$ from~$G$ and adding a
vertex~$v_S$ which is adjacent to every vertex of~$N_G(S)$; note that
this is equivalent to contracting all the edges of any spanning tree
of $G[S]$.

A somehow reverse operation of contraction is the operation of
\emph{subdivision}.  Given a graph~$G$ and an edge $uv = e \in E(G)$,
the graph obtained from subdividing~$e$ in~$G$ is the graph with
vertex set $V(G) \cup \{w_e\}$ and edge set $E(G) \setminus \{e\} \cup
\{uw_e, vw_e\}$.

A graph~$H$ which is obtained from a graph~$G$ after a sequence of
contractions is called a \emph{contraction of~$G$}.  If~$H$ is
subgraph of a contraction of~$G$, then we say that~$H$ is a
\emph{minor} of~$G$.  A graph~$G$ is said to be $H$-minor-free if~$H$
is not a minor of~$G$, and a graph class $\mathcal{G}$ is
$H$-minor-free if every graph of $\mathcal{G}$ is $H$-minor-free.

% For a graph class~$\mc G$ and integer~$\ell$, we define the
% \emph{$\ell$-truncation of~$\mc G$} as $\mc G_{\leq \ell} = \{ G \in
% \mc G : |G| \leq \ell \}$.  For any integer~$c \geq 1$ let~$K_c$ be
% the complete graph on~$c$ vertices.

\subsection{Shallow minors, grad and expansion}\label{ss:grads}
%
% 
% TODO : this definition could be improved if it used minor operations
% instead of "manually" defining a bijection/isomorphism etc.
%
%
\begin{definition}[Shallow minor]
  A graph~$M$ is an \emph{$r$-shallow minor} of~$G$, where~$r$ is an
  integer, if there exists a set of disjoint subsets $V_1, \ldots,
  V_{|M|}$ of~$V(G)$ such that
  \begin{enumerate}
  \item each graph $G[V_i]$ is connected and has radius at most~$r$,
    and
  \item there is a bijection $\psi \colon V(M) \to \{V_1, \ldots,
    V_{|M|}\}$ such that for every edge $uv \in E(M)$ there is an edge
    in~$G$ with one endpoint in $\psi(u)$ and second in $\psi(v)$.
  \end{enumerate}
  The set of all $r$-shallow minors of a graph~$G$ is denoted by~$G
  \nab r$.  Similarly, the set of all $r$-shallow minors of all the
  members of a graph class $\mc G$ is denoted by $\mc G \nab r =
  \bigcup_{G \in \mc G} (G \nab
  r)$.
\end{definition}

\begin{comment} % We probably don't need this definition
  \begin{definition}[Shallow topological minor]
    A graph~$M$ is an \emph{$r$-shallow topological minor} of~$G$ if
    there exists a set of internally vertex-disjoint paths $P_1,
    \ldots, P_{||M||}$ in~$G$ such that
    \begin{enumerate}
    \item each path $P_i$ has length at most~$2r+1$
    \item there is a bijection $\psi \colon E(M) \rightarrow \{P_1, \ldots,
      P_{||M||}\}$ such that $uv \in E(M)$ if and only if the path
      $\psi(uv)$ has endpoints $u$ and $v$.
    \end{enumerate}
    The set of all $r$-shallow topological minors of a graph~$G$ is
    denoted by~$G \topnab r$.
  \end{definition}
\end{comment}

\begin{definition}[Grad and bounded expansion]
  For a graph~$G$ and an integer~$r \geq 0$, we define the \emph{greatest
    reduced average density (grad) at depth~$r$} as
  \[
  \grad_r(G) = \max_{M \in G \nab r} \den(M) = \max_{M \in G \nab r}
  {||M||}/{|M|} .
  \]
  We extend this notation to graph classes as $\grad_r(\mc G) =
  \sup_{G \in \mc G} \grad_r(G)$.  A graph class $\mc G$ then has
  \emph{bounded expansion} if there exists a function $f \colon \N \to
  \R$ such that for all~$r$ we have that $\grad_r(\mc G) \leq f(r)$.
\end{definition}
We use shorthands $\nabla(G)$ and $\nabla(\mc G )$ to denote infinite sequences $(\grad_i(G))_{i\geq 0}$ and $(\grad_i(\mc G))_{i\geq 0}$.

Graph classes excluding a topological minor, such as planar and
bounded-degree graphs, have bounded expansion~\cite{Sparsity}.
Observe that bounded expansion implies bounded degeneracy, since
the degeneracy of $G$ lies between $\grad_0(G)$ and $2\grad_0(G)$.  However, the converse
does not hold: For an example, consider the class of cliques with each
edge subdivided once.

Observe that for every graph $G$ and integers $r\leq r'$, it holds that
$\grad_{r}(G)\leq \grad_{r'}(G)$, and the same inequality holds for
classes of graphs.  Let us revisit some basic properties of grads that
will be used later on.

\begin{restatable}{lemma}{pendant}\omitted
  \label{lem:grad-pendant}
  Let~$G$ be a graph and~let $G'$ be obtained from~$G$ by adding a
  pendant, i.e., a new vertex $v'$ with only one neighbor $v$. Then $\grad_r(G')\leq \max(\grad_r(G),1)$.
\end{restatable}

\begin{restatable}{lemma}{universalvertex}\omitted
  \label{lem:grad-universal-vertex}
  Let~$G$ be a graph and~let $G'$ be obtained from~$G$ by adding a
  universal vertex to~$G$, i.e., a vertex that is adjacent to every
  vertex of $V(G)$.  Then $\grad_r(G') \leq \grad_r(G) + 1$.
\end{restatable}

The following proposition follows directly from the definition of grads.

\begin{proposition}\label{prop:comp-of-grads}
  For every graph class $\mc G$ and every pair of nonnegative integers
  $r,s$, the following holds: $(\mc G \nab s)\nab r\subseteq \mc G\nab
  (2rs+r+s)$.  Consequently, $\grad_s(G')\leq \grad_{2rs+r+s}(\mc G)$
  for every $G'\in \mc G\nab r$.  In particular, $\grad_r(G')\leq
  \grad_{3r+1}(\mc G)$ and $\grad_1(G')\leq \grad_4(\mc G)$ for each
  $G'\in \mc G\nab 1$.
\end{proposition}

An important property of graphs of bounded expansion that we will use later on, is their
stability under taking lexicographic products.
\begin{definition}[Lexicographic product]
  Given two graphs~$G$ and~$H$, the \emph{lexicographic product} $G
  \lexprod H$ is defined as the graph on the vertex set $V(G) \times
  V(H)$ where the vertices $(u,a)$ and $(v,b)$ are adjacent if $uv \in
  E(G)$ or if $u = v$ and $ab \in E(H)$.
\end{definition}

Figure~\ref{fig:lexprod} exemplifies this procedure.  The following
lemma % demonstrates how the lexicographic product with a
shows that the grad of the lexicographic product of a graph and a
% small graph affects the grad of the resulting graph; for
% completeness, we include the proof in the appendix.
complete graph is bounded.

\begin{figure}[ht]
  \centering
  \includegraphics[width=.7\textwidth]{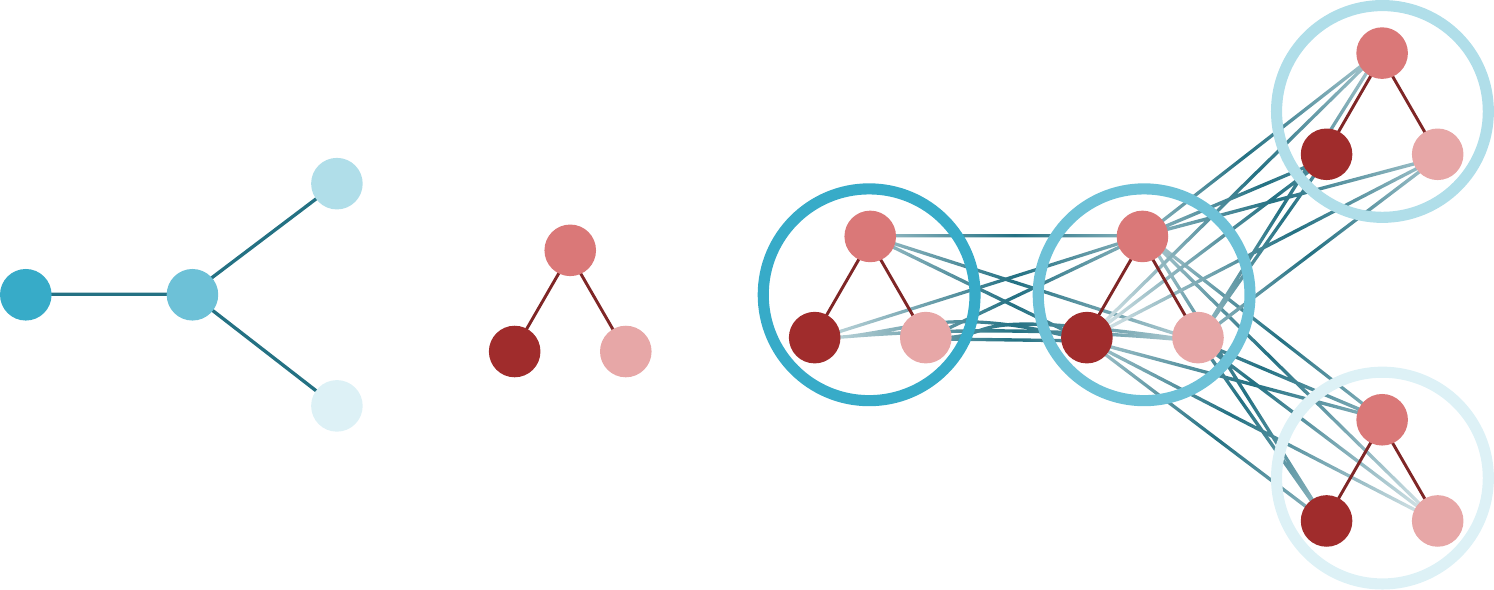}
  \caption{\label{fig:lexprod}The lexicographic product of a claw and a
    $P_3$.}
\end{figure}

\begin{restatable}{lemma}{lexprodprop}\omitted
  \label{lem:lexprod}
  For any graph~$G$ and non-negative integers~$c\geq 1$ and~$r$ we have
  that
  \[
  \grad_r(G \lexprod K_c)
  \leq 4(8c(r+c)\cdot \grad_r(G)+4c)^{(r+1)^2} .
  \]
\end{restatable}

\subsection{High degree vertices and the number of neighborhoods}

Gajarsk\'y et al.~\cite{BndExpKernels13} observed that in a graph $G$ from a class of bounded expansion, the number of possible neighborhoods in a given subset of vertices $X$ is bounded linearly in $|X|$. Moreover, the number of vertices that have many neighbors in $X$ is also small. This point of view, also used very recently by Bonamy et al.~\cite{BonamyKPS15}, is the cornerstone of our approach. More precisely, Gajarsk\'y et al. proved the following.

\begin{lemma}[\cite{BndExpKernels13}]\label{lem:twin-classes}
  Let $G$ be a graph, $X \subseteq V(G)$ be a vertex subset, and $R=V(G)\setminus X$.
  Then for every integer $p \geq \grad_1(G)$ it holds that
  \begin{enumerate}[(i)]
  \item\label{p:highdeg} $| \{ v \in R \colon |N_X(v)| \geq 2p \}| \leq 2p \cdot
    |X|$, and
  \item\label{p:neigh} $| \{ A \subseteq X \colon |A| < 2p \textrm{ and } \exists_{v\in R}\ A =
      N_X(v) \}| \leq (4^p + 2p) |X|$.
  \end{enumerate}
  Consequently, the following bound holds:
  \[
   | \{ A\subseteq X \colon \exists_{v\in R}\ A = N_X(v) \}| \leq 
  % | \{ N_X(v) \colon v \in R \}|
  \left( 4^{\grad_1(G)} + 4\grad_1(G) \right) \cdot |X| .
  \]
\end{lemma}

This statement is best suited for the standard {\sc{Dominating Set}} problem on graphs of bounded expansion, but to extend our result to {\sc{$r$-Dominating Set}}, we need proper generalizations. Suppose $G$ is a graph and $X$ is a subset of its vertices. For $u\in V(G)\setminus X$ and positive integer $r$, we define the {\em{$r$-projection}} of $u$ onto $X$ as follows: $\prg{r}{u}{X}{G}$ is the set of all those vertices $w\in X$, for which there exists a path $P$ in $G$ that starts in $u$, ends in $w$, has length at most $r$, and whose all internal vertices do not belong to $X$. Whenever the graph is clear from the context, we omit the superscript. In the following we will use the following strengthening of Lemma~\ref{lem:twin-classes}\eqref{p:highdeg}.

\newcommand{\thresh}{\xi}

\begin{lemma}[Closure lemma]\label{lem:closure}
Let $\mc G$ be a class of bounded expansion. There exists an algorithm that, given a graph $G\in \mc G$, positive integer $r$, and $X\subseteq V(G)$, computes the {\em{$r$-closure}} of $X$, denoted $\cl{r}{X}$, with the following properties (in the following, we denote $\thresh=\lceil 2\grad_{r-1}(\mc G)\rceil$):
\begin{enumerate}[(a)]
\item\label{p:superset} $X\subseteq \cl{r}{X}\subseteq V(G)$;
\item\label{p:size} $|\cl{r}{X}|\leq ((r-1)\thresh+2)\cdot |X|$; and
\item\label{p:nomore} $|\prg{r}{u}{\cl{r}{X}}{G}|\leq \thresh(1+(r-1)\thresh)$ for each $u\in V(G)\setminus \cl{r}{X}$.
\end{enumerate}
\end{lemma}
\begin{proof}
Consider the following iterative procedure.
\begin{enumerate}
\item Start with $H=G$ and $Y=X$. We will maintain the invariant that $Y\subseteq V(H)$.
\item If there exists a vertex $u\in V(H)\setminus Y$ with $|\prg{r}{u}{Y}{H}|\geq \thresh$, then do the following:
\begin{itemize}
\item Select an arbitrary subset $Z_u\subseteq \prg{r}{u}{Y}{H}$ of size $\thresh$.
\item For each $w\in Z_u$, select a path $P_w$ that starts at $u$, ends at $w$, has length at most $r$, and all its internal vertices are in $V(H)\setminus Y$.
\item Modify $H$ by contracting $\bigcup_{w\in Z_u} (V(P_w)\setminus \{w\})$ onto $u$, and add the obtained vertex to $Y$.
\end{itemize}
\item Otherwise, finish the procedure.
\end{enumerate}
Observe that in a round of the procedure above we always make a contraction of a connected subgraph of $H-Y$ of radius at most $r-1$. Also, the resulting vertex falls into $Y$ and hence does not participate in future contractions. Thus, at each point $H$ is an $(r-1)$-shallow minor of $G$. For any moment of the procedure and any $u\in V(H)$, by $\tau(u)$ we denote the subset of original vertices of $G$ that were contracted onto $u$ during earlier rounds. Note that either $\tau(u)=\{u\}$ when $u$ is an original vertex of $G$, or $\tau(u)$ is a set of cardinality at most $1+(r-1)\thresh$.

We claim that the presented procedure stops after at most $|X|$ rounds. Suppose otherwise, that we successfully constructed the graph $H$ and subset $Y$ after $|X|+1$ rounds. Examine graph $H[Y]$. This graph has $2|X|+1$ vertices: $|X|$ original vertices of $X$ and $|X|+1$ vertices that were added during the procedure. Whenever a vertex $u$ is added to $Y$ after contraction, then it introduces at least $\thresh$ new edges to $H[Y]$: these are edges that connect the contracted vertex with the vertices of $Z_u$. Hence, $H[Y]$ has at least $\thresh(|X|+1)$ edges, which means that
$$\den(H[Y])\geq \frac{\thresh(|X|+1)}{2|X|+1}>\grad_{r-1}(\mc G).$$
This is a contradiction with the fact that $H$ is an $(r-1)$-shallow minor of $G$.

Therefore, the procedure stops after at most $|X|$ rounds producing $(H,Y)$ where $|\prg{r}{u}{Y}{H}|<\thresh$ for each $u\in V(H)\setminus Y$. Define $\cl{r}{X}=\tau(Y)=\bigcup_{u\in Y} \tau(u)$. Property~\eqref{p:superset} is obvious. Since $|\tau(u)|=1$ for each original vertex of $X$ and $|\tau(u)|\leq 1+(r-1)\thresh$ for each $u$ that was added during the procedure, property~\eqref{p:size} follows. We are left with property~\eqref{p:nomore}.

By the construction $V(H)\setminus Y=V(G)\setminus \cl{r}{X}$. Take any $u\in V(H)\setminus Y$ and observe that $\prg{r}{u}{\cl{r}{X}}{G}\subseteq \tau(\prg{r}{u}{Y}{H})$. Since $|\prg{r}{u}{Y}{H}|< \thresh$ for each $u\in V(H)\setminus Y$ and $|\tau(u)|\leq 1+(r-1)\thresh$ for each $u\in V(H)$, property~\eqref{p:nomore} follows.
\end{proof}

Whenever $\grad_{r-1}(\mc G)\geq 1$, which will be the case in our main proof, we will use simplified, weaker bounds: $|\cl{r}{X}|\leq 3r\grad_{r-1}(\mc G)\cdot |X|$ and $|\prg{r}{u}{\cl{r}{X}}{G}|\leq 9r\grad_{r-1}(\mc G)^2$. Observe that Lemma~\ref{lem:closure} is not merely a generalization of Lemma~\ref{lem:twin-classes}\eqref{p:highdeg} to $r$-neighborhoods. It shows that a certain maximality property can be achieved; this property may be not true if, even for $r=1$, we would construct $\cl{r}{X}$ from $X$ by just adding all vertices with many neighbors in $X$.

The generalization of Lemma~\ref{lem:twin-classes}\eqref{p:neigh} which we will use later is the following.

\begin{lemma}\omitted\label{lem:fewneighbourhoods}
Let $\mc G$ be a class of bounded expansion and let $r$ be a positive integer. Let $G\in \mc G$ be a graph and $X\subseteq V(G)$. Then 
$$|\{Y \colon Y=\pr{r}{u}{X}\textrm{ for some }u\in V(G)\setminus X\}|\leq c\cdot |X|,$$
for some constant $c$ depending only on $r$ and the grads of $\mc G$.
\end{lemma}

Lemma~\ref{lem:fewneighbourhoods}, in a slightly different form, can be found in the PhD thesis of the eight author~\cite[Theorem 18]{FelixNew}. For the sake of completeness, in Appendix~\ref{sec:fewneighbourhoods} we give a self-contained proof of this result using centered colorings, which is similar to the proof given in~\cite{FelixNew}.

Finally, for the proof of Theorem~\ref{thm:mainrbe-dominator} for $r>1$ we will need the following lemma.

\begin{lemma}[Short paths closure lemma]\label{lem:shpclo}
Let $\mc G$ be a class of bounded expansion and let $r$ be a positive integer. Let $G\in \mc G$ be a graph and $X\subseteq V(G)$. Then there is a superset of vertices $X'\supseteq X$ with the following properties:
\begin{enumerate}[(a)]
\item\label{pr:sh-corr} Whenever $\dist_G(u,v)\leq r$ for some distinct $u,v\in X$, then $\dist_{G[X']}(u,v)=\dist_G(u,v)$.
\item\label{pr:sh-size} $|X'|\leq Q_r(\grad_{r-1}(\mc G))\cdot |X|$ for some polynomial $Q_r$.
\end{enumerate}
Moreover, $X'$ can be computed in polynomial time.
\end{lemma}
\begin{proof}
First, using Lemma~\ref{lem:closure} we compute $X_0=\cl{r}{X}$. Then $|X_0|\leq ((r-1)\thresh+2)|X|$, where $\thresh=\lceil 2\grad_{r-1}(\mc G)\rceil$, and for each vertex $u\notin X_0$ we have $|\prg{r}{u}{X_0}{G}|\leq \thresh(1+(r-1)\thresh)$. Now, for each pair of distinct vertices $u,v\in X_0$, select an arbitrary path $P_{u,v}$ that connects $u$ and $v$, and whose internal vertices do not belong to $X_0$, and which is the shortest among the paths satisfying these properties; in case there is no such path, put $P_{u,v}=\emptyset$. Note that $P_{u,v}$ can be computed in polynomial time. Then define $X'$ to be $X_0$ plus the vertex sets of all paths $P_{u,v}$ that have length at most $r$.

\begin{claim}\label{cl:blowup}
$|X'|\leq \widetilde{Q}_r(\grad_{r-1}(\mc G))\cdot |X_0|$, for some polynomial $\widetilde{Q}_r$.
\end{claim}
\begin{proof}
Let $H$ be a graph on vertex set $X_0$, where $uv\in E(X_0)$ if and only if $P_{u,v}$ exists and has length at most $r$, and hence its vertex set was added to $X$. Clearly $|X'|\leq |X_0|+(r-1)|E(H)|$, so it suffices to prove an upper bound on $|E(H)|$. Take any $w\in X'\setminus X_0$, and consider for how many pairs $\{u,v\}$ it can hold that $w\in P_{u,v}$. If $\{u,v\}$ is such a pair, then in particular $u,v\in \prg{r}{w}{X_0}{G}$. But we know that $|\prg{r}{w}{X_0}{G}|\leq \thresh(1+(r-1)\thresh)$, so the number of such pairs is at most $\tau=\binom{\thresh(1+(r-1)\thresh)}{2}$. Consequently, we observe that graph $H$ is an $(r-1)$-shallow minor (actually even an $\lceil (r-1)/2\rceil$-shallow topological minor) of $G\lexprod K_\tau$: when each vertex $w\in X\setminus X_0$ is replaced with $\tau$ copies, then we can realize all the paths $P_{u,v}$ in $G\lexprod K_\tau$ so that they are internally vertex-disjoint. From Lemma~\ref{lem:lexprod} we know that the $\grad_{r-1}(G\lexprod K_\tau)$ is bounded polynomially in $\grad_{r-1}(G)$ and $\tau$, which in turn is also bounded polynomially in $\grad_{r-1}(\mc G)$. Hence $\grad_{r-1}(G\lexprod K_\tau)$ is bounded polynomially in $\grad_{r-1}(\mc G)$. As $|E(H)|\leq \grad_{r-1}(G\lexprod K_\tau)\cdot |X_0|$ and $|X'|\leq |X_0|+(r-1)|E(H)|$, we are done.
\cqed\end{proof}

\begin{claim}\label{cl:correctness}
If $u,v\in X_0$ are distinct and $\dist_G(u,v)\leq r$, then $\dist_{G[X']}(u,v)=\dist_G(u,v)$.
\end{claim}
\begin{proof}
Let $R$ be a shortest path between $u$ and $v$ in $G$, and let $a_1,a_2,\ldots,a_q$ be consecutive vertices of $X_0$ visited on $R$, where $u=a_1$ and $v=a_q$. For each $i=1,2,\ldots,q-1$, let $R_i$ be the segment of $R$ between $a_i$ and $a_{i+1}$. Then the existence of $R_i$ certifies that some path of length at most $|R_i|$ between $a_i$ and $a_{i+1}$ was added when constructing $X'$ from $X_0$, and hence $\dist_{G[X']}(a_i,a_{i+1})\leq |R_i|$. Consequently, by the triangle inequality we infer that
$$\dist_{G[X']}(u,v)\leq \sum_{i=1}^{q-1}\dist_{G[X']}(a_i,a_{i+1})\leq \sum_{i=1}^{q-1}|R_i|=|R|=\dist_G(u,v).$$
However, the opposite inequality $\dist_{G[X']}(u,v)\geq \dist_G(u,v)$ follows directly from the fact that $G[X']$ is an induced subgraph of $G$. Hence we are done.
\cqed\end{proof}

Claim~\ref{cl:blowup} and the fact that $|X_0|\leq ((r-1)\thresh+2)|X|$ prove property~\eqref{pr:sh-size}. Claim~\ref{cl:correctness} and the fact that $X\subseteq X_0$ prove property~\eqref{pr:sh-corr}.
\end{proof}

\subsection{Domination and scattered sets}

% REMOVED AS DEFINED NOT FAR ABOVE HERE
%
% A \emph{$2$-scattered} set is a set of vertices pairwise in distance
% at least three from each other; Hence, their closed neighborhoods
% are pairwise disjoint.  Note that a $2$-scattered set of size~$> k$
% is an obstruction to a dominating set of size~$\leq k$.

We now state the constant-factor approximation for $r$-\DS proved by
\Dvorak~\cite{Dvorak13}.  The statement is slightly different from the
results there, and we therefore explain how this exact statement can
be derived from the work of \Dvorak{} in Appendix~\ref{sec:omitted}.

% In the following we will use a constant-factor approximation of the
% dominating set for graphs of bounded expansion which can be achieved
% with the help of the \emph{weak $k$-coloring number}, see
% Definition~\ref{def:wcol}.

\begin{restatable}{theorem}{dvorakfull}\omitted
  \label{thm:dvorak}
  Let $r$ be a positive integer. There is a polynomial $P_r$ and a polynomial-time algorithm that given a graph $G$ and an
  integer $k$, either finds an $r$-dominating set of size at most
  $P_r(\grad_r(G))\cdot k$ or a $2r$-scattered set of size at
  least~$k+1$ in~$G$.
\end{restatable}

We remark that the proof of Theorem~\ref{thm:dvorak} does not assume
that the graph belongs to some class of bounded expansion.  If this is
the case, then algorithm can be implemented with a slightly better
approximation ratio and in linear time.  However, in the nowhere
dense case it will be important for us that we can apply
Theorem~\ref{thm:dvorak} without this assumption, and in particular
that the running time does not depend exponentially on the grads
of~$G$.

% assumes that the graph under consideration belongs to a fixed class
% of bounded expansion $\mc G$, and thus the constant in the linear
% running time of the algorithm depends exponentially on the grads of
% $%\mc G$.  Observe that the approximation factor of the algorithm
% can be bounded by $(8\grad_1(\mc G)^3+1)^2\leq 2^8\grad_1(\mc G)^6$,
% provided that $\grad_1(\mc G)\geq 1$.  For simplicity, from now on we
% shall assume that $\grad_1(\mc G)\geq 1$ for all the considered graph
% classes $\mc G$ of bounded expansion.

\begin{comment}
  The important takeaway here is that the approximation factor of
  Proposition~\ref{prop:dvorak} can be expressed as depending only on
  $\grad_1(G)$ or $\grad_1(\mc G)$.  An explicit---probably not
  optimal---bound is
\[
  \wcol_r(G) \leq \left( (2\grad_{(r-1)/2}(G))^{(2r)^{2r}} + 1 \right)^r,
\]
see the appendix for a proof.
\end{comment}

We need the following strengthened version of \Dvorak's algorithm
that approximates domination of only some subset of vertices.

\begin{lemma}\label{lem:better-dvorak}
  Let $r$ be a positive integer. There is a polynomial $\widetilde{P}_r$ and a polynomial-time algorithm that, given a graph $G$, a
  vertex subset $Z \subseteq V(G)$ and an integer $k$, finds either
  \begin{itemize}
  \item a $(Z,r)$-dominator in $G$ of size at most
    $\widetilde{P}_r(\grad_r(G))\cdot k$, or
  \item a subset of~$Z$ of size at least~$k+1$ that is $2r$-scattered
    in $G$.
  \end{itemize}
\end{lemma}
\begin{proof}
  Obtain $G'$ from $G$ by the following construction: add two new vertices $v$ and $v'$, and for each vertex $w\in (V(G)\setminus Z)\cup \{v'\}$ create a new path $P_w$ of length $r$ with endpoints $v$ and $w$.  Apply Theorem~\ref{thm:dvorak} to graph $G'$ with parameter $k+1$.
  
  Suppose first that the algorithm outputs an $r$-dominating set $D$ in
  $G'$. Observe that a supergraph of $G'$ can be obtained from $G$ by iteratively adding pendants and once adding a universal vertex. Hence, by Lemmas~\ref{lem:grad-pendant},~\ref{lem:grad-universal-vertex} and~Theorem~\ref{thm:dvorak}, $D$ has size at most
  \[
  P_r(\grad_r(G'))\cdot(k+1) \leq
  P_r(\grad_r(G)+2)\cdot(k+1) \leq 
  2P_r(\grad_r(G)+2)\cdot k.
  \]

Construct $D'$ from $D$ as follows: First, remove all the vertices of $V(P_{v'})\cap D$ from $D$ (note that $V(P_{v'})\cap D\neq\emptyset$ since $D$ $r$-dominates $v'$) and replace them by $\{v\}$. Then, for every $w\in V(G)\setminus Z$, if $(V(P_w)\setminus \{v\})\cap D\neq \emptyset$, then remove $(V(P_w)\setminus \{v\})\cap D$ from $D$ and replace it with $\{w\}$. From the construction it follows that $|D'|\leq |D|$ and $D'$ is an $r$-dominating set. Since $D'\cap (V(G')\setminus V(G))=\{v\}$ and $v$ does not $r$-dominate any vertex of $Z$, we infer the $D'\cap V(G)$ is a $(Z,r)$-dominator of size at most $2P_r(\grad_r(G)+2)\cdot k-1$. Hence, we can take $\widetilde{P}_r(x)=2P_r(x+2)$.
  
  Suppose now that the algorithm provided a $2r$-scattered set $S$ in
  $G'$ of size at least $k+2$.  Observe that the graph $G'-Z$ has
  diameter $2r$ since $v$ is at distance at most $r$ from each vertex of this graph.  Hence
  any $2r$-scattered set in $G'$ contains at most one vertex from $V(G)
  \setminus Z$.  Therefore, $S$~can contain at most one vertex outside
  of $Z$ in $G'$, hence $|S \cap Z| \geq k + 1$ and $S\cap Z$ is the
  sought $2r$-scattered subset of $Z$.
\end{proof}

The same construction as in the proof of Lemma~\ref{lem:better-dvorak} shows the following simple reduction, which we will need for the proof of Theorem~\ref{thm:mainrbe-reduced}.

\begin{lemma}\label{lem:dominator-reduction}
There exists a polynomial-time algorithm that, given a graph $G$, set $Z\subseteq V(G)$, and positive integer $r$, outputs a supergraph $G'$ of $G$ such that the following holds:
\begin{itemize}
\item $\ds_r(G')=\ds_r(G,Z)+1$;
\item $|V(G')|\leq (r+1)\cdot (|V(G)|+1)$;
\item $\grad_i(G')\leq \max(\grad_i(G)+1,2)$, for each nonnegative integer $i$.
\end{itemize}
\end{lemma}
\begin{proof}
Construct $G'$ from $G$ as follows: Add two new vertices $v$ and $v'$ and connect $v$ with every vertex of $(V(G)\setminus Z)\cup \{v'\}$ by a path of length $r$. The reasoning contained in the proof of Lemma~\ref{lem:better-dvorak} shows that $\ds_r(G')\geq \ds_r(G,Z)+1$, while the opposite inequality follows from the observation that any $(Z,r)$-dominator in $G$ becomes an $r$-dominating set in $G'$ after adding vertex $v$. The bound on $|V(G')|$ follows directly from the construction, whereas the bound on the grads of $G'$ follows from Lemmas~\ref{lem:grad-pendant},~\ref{lem:grad-universal-vertex} and the fact that a supergraph of $G'$ can be constructed from $G$ by iteratively adding pendant vertices, and once adding a universal vertex.
\end{proof}

%%% Local Variables:
%%% TeX-command-default: "Make"
%%% mode: latex
%%% TeX-master: "main"
%%% End: 

\section{A kernel for graphs of bounded expansion}
\label{sec:bnd-exp}

In this section we give a linear kernels for \textsc{Dominating Set} and \textsc{$r$-Dominating Set} on
graphs of bounded expansion; that is, we prove
Theorems~\ref{thm:main-be},~\ref{thm:mainrbe-dominator}, and~\ref{thm:mainrbe-reduced}.  Let us fix a graph class $\mathcal{G}$
that has bounded expansion, and let $(G,k)$ be the input instance of
\textsc{$r$-Dominating Set}, where $G\in \mathcal{G}$.  We assume that
% $\grad_r(\mc G)\geq 1$ for every nonnegative integer~$r$, since
%
%\note{changed $\forall r \grad_r$ to $\grad_0$}
$\grad_0(\mc G) \geq 1$, otherwise~$G$ is a forest and the
\textsc{$r$-Dominating Set} problem can be solved in linear time.

% In the following, we assume that $\mc G$ is fixed and thus the
We assume that $\mc G$ is fixed, and hence so are also
% the algorithm knows
the values of $\grad_i(\mc G)$ for all nonnegative integers~$i$.
% (they are hard-coded in the implementation).  Equivalently, the
% algorithm may be given some upper bounds on these values as
% $\grad_r(\mc G)$, since value $\grad_r(\mc G)$ will be used only to
% argue that $\grad_r(G)\leq \grad_r(\mc G)$ for each $G\in \mc G$.
We discuss in Section~\ref{sec:conclusion} % the assumption
that the values of $\grad_i(\mc G)$ need not be known to
the algorithm in advance, but this assumption will significantly simplify the analysis.

As explained in Section~\ref{sec:intro}, the first goal is to reduce the
number of dominatees.  More precisely, we find a subset of vertices~$Z$ of
size linear in~$k$, called an \emph{$r$-domination core}, such that any minimum-size
$(Z,r)$-dominator is guaranteed to $r$-dominate the whole graph.  In this manner,
domination of vertices outside the $r$-domination core is not relevant to the
problem, and they can only serve the role of $r$-dominators.  Reducing their
number is performed in the second step of the algorithm.

\subsection{Reducing dominatees}
\def\NX2{\hat N^2}
% ಠ_ಠ
We begin with introducing formally the notion of an $r$-domination core:
\begin{definition}[$r$-domination core]\label{def:domcore}
  Let~$G$ be a graph and~$Z$ be a subset of vertices.  We say that~$Z$
  is an \emph{$r$-domination core} in~$G$ if % the following holds:
  every minimum-size $(Z,r)$-dominator in~$G$ is also an $r$-dominating set
  in~$G$.
\end{definition}
Clearly, the whole $V(G)$ is an $r$-domination core, but we look for an
$r$-domination core that is small in terms of~$k$.  Note that if~$Z$ is an
$r$-domination core, then $\ds_r(G)=\ds_r(G,Z)$.  Let us remark that in this
definition we do not require that every~$(Z,r)$-dominator is an $r$-dominating
set in~$G$; there can exist $(Z,r)$-dominators that are not of minimum
size and that do not dominate the whole graph.

%  #####                                  #####                       
% #     # #    #   ##   #      #         #     #  ####  #####  ###### 
% #       ##  ##  #  #  #      #         #       #    # #    # #      
%  #####  # ## # #    # #      #         #       #    # #    # #####  
%       # #    # ###### #      #         #       #    # #####  #      
% #     # #    # #    # #      #         #     # #    # #   #  #      
%  #####  #    # #    # ###### ######     #####   ####  #    # ###### 

The rest of this subsection is devoted to the proof of the following
theorem.

\begin{theorem}
  \label{thm:dcore}
  There exists a function $\domcoresize(\cdot)$ of the grads of $\mc G$ and a polynomial-time
  algorithm that, given an instance $(G,k)$ where $G\in \mc G $,
  either correctly concludes that $\ds_r(G) > k$, or finds an $r$-domination
  core $Z \subseteq V(G)$ with $|Z| \leq \domcoresize(\nabla(\mc
  G))\cdot k$.
\end{theorem}

We fix~$G$ and~$k$ in the following to improve readability.  For the
proof of Theorem~\ref{thm:dcore} we start with $Z = V(G)$ and
gradually reduce~$|Z|$ by removing one vertex at a time, while
maintaining the invariant that~$Z$ is an $r$-domination core.  To this end,
we need to prove the following lemma, from which
Theorem~\ref{thm:dcore} follows trivially as explained:

\begin{lemma}
  \label{lem:reduce-corce}
  There exists a function $\domcoresize(\cdot)$ of the grads of $\mc G$ and a polynomial-time
  algorithm that, given an $r$-domination core $Z \subseteq V(G)$ with $|Z|
  > \domcoresize(\nabla(\mc G)) \cdot k$,
  % and a promise that $Z$ is
  either correctly concludes that $\ds_r(G) > k$, or finds a vertex $z
  \in Z$ such that $Z \setminus \{z\}$ is still an $r$-domination core.
\end{lemma}

Thus, from now on we focus on proving Lemma~\ref{lem:reduce-corce}. To remove possible confusion, let us remark that function $\domcoresize(\cdot)$ that is yielded by our proof will depend only on the first $p$ grads of $\mc G$, for some constant $p$ depending on $r$. Thus, the algorithm does not need to have access to an infinite sequence of grads to compute the constants used in its code.

\subsubsection{Iterative extraction of $Z$-dominators}
\label{sec:extraction}

The first phase of the algorithm of Lemma~\ref{lem:reduce-corce} is to build a structural
decomposition of the graph~$G$.  More precisely, we try to ``pull
out'' a small set~$X$ of vertices that $r$-dominates~$Z$, so that after removing them,~$Z$ contains a large
subset~$S$, which is $2r$-scattered in the remaining graph.  Given such
a structure, intuitively we can argue that in any optimal $(Z,r)$-dominator, vertices of $X$ serve as ``hubs'' that route almost all the domination paths leading to vertices of $S$. This is because any vertex of $V(G)\setminus X$ can $r$-dominate only at most one vertex
from~$S$ via a path that avoids $X$.  Since~$S$ will be large compared to~$X$, some vertices
of~$S$ will be indistinguishable from the point of view of $r$-domination
routed through~$X$, and these will be precisely the vertices that can be removed
from the domination core.  The identification of the irrelevant
dominatee will be the goal of the second phase,
whereas the goal of this phase is to construct the pair $(X,S)$.

\newcommand{\Cdv}{C_{\rm{dv}}}
\newcommand{\Cclsz}{\Gamma_{\rm{cl}}}
\newcommand{\Cclnei}{\Delta_{\rm{cl}}}

Let $\Cdv=\widetilde{P}_r(\grad_r(\mc G))$ be the approximation ratio of the algorithm of Lemma~\ref{lem:better-dvorak}.  Given~$Z$, we first apply the algorithm of
Lemma~\ref{lem:better-dvorak} to~$G$,~$Z$, and the parameters~$r$ and~$k$.
Thus, we either find a $(Z,r)$-dominator~$Y_1$ such that $|Y_1|\leq
\Cdv\cdot k$, or
we find a subset $S \subseteq Z$ of size at least~$k+1$ that is
$2r$-scattered in~$G$.  In the latter case,
% we infer that $\ds(G,Z)>k$,
since~$S$ is an obstruction to an $r$-dominating set of size at most~$k$,
% any vertex of~$G$ can dominate only at most one vertex of~$S$.
% Since~$Z$ is a domination core, we have that $\ds(G) = \ds(G,Z)>k$,
% and
we may terminate the algorithm and provide a negative answer.  Hence,
from now on we assume that~$Y_1$ has been successfully constructed.

Let $C_0$ be a constant depending on $\nabla(\mc G)$, to be defined later. Now, in search for the pair $(X,S)$, we inductively construct sets $X_1,Y_2,X_2,Y_3,X_3,\ldots$ such that
$Y_1\subseteq X_1\subseteq Y_2\subseteq X_2\subseteq \ldots$ using the following definitions:
\begin{itemize}
\item If $Y_i$ is already defined, then set $X_i=\cl{3r}{Y_i}$.
\item If $X_i$ is already defined, then apply the algorithm of
  Lemma~\ref{lem:better-dvorak} to $G - X_i$, $Z \setminus X_i$, and
  the parameters $r$ and $C_0 \cdot |X_i|$. 
\begin{enumerate}
\item\label{big-scattered} Suppose the algorithm finds a set $S\subseteq Z \setminus X_i$ that is $2r$-scattered in $G - X_i$ and
  has cardinality greater than $C_0 \cdot |X_i|$. Then
  we let $X = X_i$, terminate the procedure and
  proceed to the second phase with the pair $(X,S)$.
\item\label{small-domset} Otherwise, the algorithm has found a $(Z
  \setminus X_i,r)$-dominator $D_{i+1}$ in $G - X_i$ of size at most $\Cdv\cdot C_0\cdot |X_i|$. Then set $Y_{i+1}=X_i\cup D_{i+1}$ and proceed.
\end{enumerate}
\end{itemize}

Let $\Cclsz=9r\grad_{3r-1}(\mc G)$ be the bound on the size blow-up in Lemma~\ref{lem:closure} applied to radius $3r$, and let $\Cclnei=27r\grad_{3r-1}(\mc G)^2$ be the upper bound on the sizes of $3r$-projections given by Lemma~\ref{lem:closure}. From Lemmas~\ref{lem:better-dvorak},~\ref{lem:closure}, and a trivial induction we infer that the following bounds hold for all $i$ for which $(Y_i,X_i)$ were constructed:
\begin{eqnarray*}
|Y_i|& \leq & \Cdv\Cclsz^{i-1}(1+\Cdv C_0)^{i-1}\cdot k,\\
|X_i|& \leq & \Cdv\Cclsz^i(1+\Cdv C_0)^{i-1}\cdot k.
\end{eqnarray*}
For a nonnegative integer $i$, let $K_i=\Cdv\Cclsz^i(1+\Cdv C_0)^{i-1}$.

In this manner, the algorithm consecutively extracts dominators
$D_2,D_3,D_4,\ldots$ and performs $3r$-closure, constructing sets $X_2,X_3,X_4,\ldots$ up to the
point when case (\ref{big-scattered}) is encountered.  Then the
computation is terminated and the sought pair $(X,S)$ is constructed.
We now claim that case (\ref{big-scattered}) always happens within a
constant number of iterations.

\newcommand{\numiter}{\Lambda}

\begin{lemma}\label{lem:extraction-terminates}
  Let $\numiter=\sum_{i=0}^r \Cclnei^i\leq (r+1)\Cclnei^r$. Assuming that $|Z| > K_{\numiter} \cdot k$,
  the construction terminates yielding some pair~$(X,S)$ before performing $\numiter$ iterations, that is, before constructing~$Y_{\numiter}$.
\end{lemma}
\begin{proof}
For the sake of contradiction, suppose~$Y_{\numiter}$ and~$X_{\numiter}$ were successfully constructed. Since $|Z|>K_{\numiter} \cdot k$ and $|X_{\numiter}|\leq K_{\numiter} \cdot k$, there is some vertex $u\in Z\setminus X_{\numiter}$. For an index $1\leq i\leq \numiter$, we shall say that a vertex $w\in X_i\setminus X_{i-1}$ is {\em{$i$-good}} if there is a path $P$ that starts at $u$, ends at $w$, has length at most $r$, and all its internal vertices do not belong to $X_i$ (we denote $X_0=\emptyset$). Vertex $w$ is {\em{good}} if it is good for some index $i$.

\begin{claim}\label{cl:upper}
The number of good vertices is at most $\numiter-1$.
\end{claim}
\begin{proof}
Let $w$ be any good vertex, and let $P$ be a path certifying this. Let $q\leq r$ be the length of $P$, and denote the vertices of $P$ by $u_i$ for $0\leq i\leq q$, where $u_0=u$ and $u_q=w$. Observe that internal vertices of $P$ can belong only to sets $X_j\setminus X_{j-1}$ for $j>i$, or to $V(G)\setminus X_{\numiter}$. We say that a vertex $u_\ell$ of $P$ is {\em{important}} if there is an index $j$, with $i\leq j\leq \numiter$, such that $u_\ell\in X_j\setminus X_{j-1}$ but $u_{\ell'}\notin X_j$ for all $\ell'<\ell$. Clearly, $w=u_q$ is important. Let $\ell_1<\ell_2<\ldots<\ell_p=q$ be the indices of important vertices on $P$, and let $j_1>j_2>\ldots>j_p$ be such that $u_{\ell_i}\in X_{j_i}\setminus X_{j_i-1}$, for all $1\leq i\leq p$. We will denote $\ell_0=0$, so $u_{\ell_0}=u$, and $j_0=\numiter+1$ (denoting $X_{\numiter+1}=V(G)$).
\begin{figure}[htb]
\centering
  \includegraphics[scale=.85]{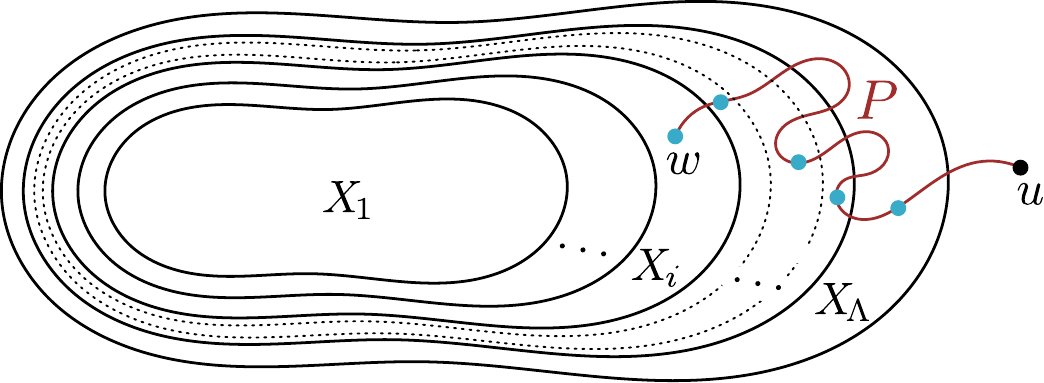}
  \caption{\label{fig:goodv} Situation in the proof of Claim~\ref{cl:upper}. Grey dots denote important vertices.}
\end{figure}
\noindent Consider any index $i$ with $1\leq i\leq p$. Observe that on the part between $u_{\ell_{i-1}}$ and $u_{\ell_{i}}$, path $P$ never entered $X_{j_{i-1}-1}$, because first such entrance would constitute an important vertex that was not recorded. Since $u_{\ell_{i}}\in X_{j_{i-1}-1}$, we infer that $u_{\ell_{i}}\in \prg{r}{u_{\ell_{i-1}}}{X_{j_{i-1}-1}}{G}$. Since $X_{j_{i-1}-1}=\cl{3r}{Y_{j_{i-1}-1}}$ by the construction, we infer that $|\prg{r}{u_{\ell_{i-1}}}{X_{j_{i-1}-1}}{G}|\leq \Cclnei$. Therefore, once vertex $u_{\ell_{i-1}}$ is selected, there are at most $\Cclnei$ choices for the next important vertex $u_{\ell_i}$. We infer that the choice of the sequence of important vertices on $P$ can be modeled by taking at most $r$ decisions, each from a selection of at most $\Cclnei$ options. Since $w$ is the last important vertex, there are at most $\sum_{i=1}^r \Cclnei^i=\numiter-1$ ways to select $w$.
\cqed\end{proof}

\begin{claim}\label{cl:lower}
For every $1\leq i\leq \numiter$, there is an $i$-good vertex.
\end{claim}
\begin{proof}
Recall that $D_i\subseteq X_i\setminus X_{i-1}$ is a $(Z\setminus X_{i-1},r)$-dominator in the graph $G-X_{i-1}$. Since $u\in Z\setminus X_{i-1}$, in $G-X_{i-1}$ there is a path $P$ of length at most $r$ from $u$ to a vertex of $D_i$. Take $w$ to be the first vertex of this path that belongs to $X_i\setminus X_{i-1}$. Then the prefix of $P$ from $u$ to $w$ certifies that $w$ is an $i$-good vertex.
\cqed\end{proof}

Claims~\ref{cl:upper} and~\ref{cl:lower} contradict each other, which finishes the proof.
\end{proof}
In Lemma~\ref{lem:reduce-corce} we will set $\domcoresize(\nabla(\mc G))=K_{\numiter}$, so that Lemma~\ref{lem:extraction-terminates} can be applied.

Therefore, unless the size of~$Z$ is bounded by $K_{\numiter} \cdot k$, the construction terminates within
$\numiter=\sum_{i=0}^r \Cclnei^i\leq (r+1)\Cclnei^r$ iterations with a pair~$(X,S)$. By the
construction of~$X$ and~$S$, we have the following properties:
\begin{itemize}
\item $|X|\leq K_{\numiter}\cdot k$;
\item $X$ is a $(Z,r)$-dominator in $G$ (because $Y_1\subseteq X$);
\item for each $u\in V(G)\setminus X$, we have $|\prg{3r}{u}{X}{G}|\leq \Cclnei$;
\item $|S| > C_0\cdot |X|$;
\item $S \subseteq Z \setminus X$ and $S$ is $2r$-scattered in $G - X$.
\end{itemize}
With sets~$X$ and~$S$ computed we proceed to the second phase, that
is, finding an irrelevant dominatee that can be removed from~$Z$.

% ###                                                               
%  #  #####  #####  ###### #      ###### #    #   ##   #    # ##### 
%  #  #    # #    # #      #      #      #    #  #  #  ##   #   #   
%  #  #    # #    # #####  #      #####  #    # #    # # #  #   #   
%  #  #####  #####  #      #      #      #    # ###### #  # #   #   
%  #  #   #  #   #  #      #      #       #  #  #    # #   ##   #   
% ### #    # #    # ###### ###### ######   ##   #    # #    #   #   

\subsubsection{Finding an irrelevant dominatee}\label{sec:finding}

\newcommand{\Cnum}{C_{\rm{nei}}}

Given~$G$,~$Z$, and the constructed sets~$X$ and~$S$, we denote by $R =
V(G) \setminus X$ the set of vertices outside~$X$.  Using this notation,~$S$
is $2r$-scattered in the graph~$G[R]$.  Recall that for any vertex $u
\in R$, we have $|\pr{3r}{u}{X}|\leq \Cclnei$.

Define the following equivalence relation $\simeq$ on $S$: for $u,v\in S$, let
$$
u\simeq v\quad \Leftrightarrow \quad \pr{i}{u}{X} = \pr{i}{v}{X}\textrm{ for each $1\leq i\leq 3r$.}
$$
Let us denote by $\Cnum$ the constant $c$ given by Lemma~\ref{lem:fewneighbourhoods} for class $\mc G$ and radius $3r$. Hence, the number of different $3r$-projections in $X$ of vertices of $R$ is bounded by $\Cnum\cdot |X|$.

\begin{lemma}\label{lem:num-classes}
Equivalence relation $\simeq$ has at most $\Cnum\cdot (3r)^{\Cclnei} \cdot |X|$ classes.
\end{lemma}
\begin{proof}
Observe that for each $u\in S$,
$$\pr{1}{u}{X}\subseteq \pr{2}{u}{X}\subseteq \ldots \subseteq \pr{3r-1}{u}{X}\subseteq \pr{3r}{u}{X}.$$
By Lemma~\ref{lem:fewneighbourhoods}, the number of choices for $\pr{3r}{u}{X}$ is at most $\Cnum\cdot |X|$. Moreover, since $u\in R$, we have that $|\pr{3r}{u}{X}|\leq \Cclnei$. Hence, to define sets $\pr{i}{u}{X}$ for $1\leq i<3r$ it suffices, for every $w\in \pr{3r}{u}{X}$, to choose the smallest index $j$, $1\leq j\leq 3r$, such that $w\in \pr{j}{u}{X}$. The number of such choices is at most $(3r)^{\Cclnei}$, and hence the claim follows.
\end{proof}

We can finally set the constant $C_0$ that was introduced in the previous section: $C_0=(\Cclnei+1)\cdot \Cnum\cdot (3r)^{\Cclnei}$. Since we have that $|S|>C_0\cdot |X|$, from Lemma~\ref{lem:num-classes} and the pigeonhole principle we infer that there is a class $\kappa$ of relation $\simeq$ with $|\kappa|>\Cclnei+1$. Note that we can find such a class $\kappa$ in polynomial time, by computing the classes of $\simeq$ directly from the definition and examining their sizes. We are ready to prove the final lemma of this section: any vertex of $\kappa$ can be removed from the $r$-domination core $Z$ (recall that $S\subseteq Z$).

\begin{lemma}\label{lem:irrelevant}
Let $z$ be an arbitrary vertex of $\kappa$. Then $Z\setminus \{z\}$ is an $r$-domination core.
\end{lemma}
\begin{proof}
Let $Z'=Z\setminus \{z\}$. Take any minimum-size $(Z',r)$-dominator $D$ in $G$. If $D$ also dominates $z$, then $D$ is a minimum-size $(Z,r)$-dominator as well. Since $Z$ was an $r$-domination core, we infer that $D$ is an $r$-dominating set in $G$, and we are done. Hence, suppose $z$ is not $r$-dominated by $D$. We prove that this case leads to a contradiction, which will conclude the proof.

Every vertex $s\in \kappa\setminus \{z\}$ is $r$-dominated by $D$. For each such $s$, let $v(s)$ be an arbitrarily chosen vertex of $D$ that $r$-dominates $s$, and let $P(s)$ be an arbitrarily chosen path of length at most $r$ that connects $v(s)$ with $s$.
\begin{claim}\label{cl:dom-via-X}
For each $s\in \kappa\setminus \{z\}$, path $P(s)$ does not pass through any vertex of $X$ (in particular $v(s)\notin X$). Consequently, vertices $v(s)$ for $s\in \kappa\setminus \{z\}$ are pairwise different.
\end{claim}
\begin{proof}
Suppose otherwise and let $w$ be the vertex of $V(P(s))\cap X$ that is closest to $s$ on $P(s)$. Then the suffix of $P(s)$ from $w$ to $s$ certifies that $w\in \pr{j}{s}{X}$, for $j$ being the length of this suffix. As $s\simeq z$, we also have that $w\in \pr{j}{z}{X}$, so there is a path $Q$ of length at most $j$ from $w$ to $z$. By concatenating the prefix of $P(s)$ from $v(s)$ to $w$ with $Q$ we obtain a walk of length at most $r$ from $v(s)$ to $z$, a contradiction with the assumption that $z$ is not $r$-dominated by $D$.

For the second part of the claim, suppose $v(s)=v(s')$ for some distinct $s,s'\in \kappa\setminus \{z\}$. Then the concatenation of $P(s)$ and $P(s')$ would be a path of length at most $2r$ connecting $s$ and $s'$ that is entirely contained in $G[R]$. This would be a contradiction with the fact that $S$ is $2r$-scattered in $G[R]$.
\cqed\end{proof}

Let $W=\{v(s)\colon s\in \kappa\setminus \{z\}\}$. From Claim~\ref{cl:dom-via-X} we have that $|W|=|\kappa\setminus \{z\}|\geq \Cclnei+1$. Define $D'=(D\setminus W)\cup \pr{3r}{z}{X}$. Since $|\pr{3r}{z}{X}|\leq \Cclnei$, we have that $|D'|<|D|$.

\begin{claim}\label{cl:exchange}
$D'$ is a $(Z',r)$-dominator.
\end{claim}
\begin{proof}
%%%Marcin: I collapsed the last two paragraphs into one
For the sake of contradiction, suppose there is some $a\in Z'$ that is not $r$-dominated by $D'$. Since $a$ was $r$-dominated by $D$ and $D\setminus D'=W$, there must be a vertex $s\in \kappa\setminus \{z\}$ such that vertex $v(s)$ $r$-dominates $a$. Consequently, in $G$ there is a path $Q_0$ of length at most $r$ that leads from $v(s)$ to $a$.
Furthermore, since $X$ is a $(Z,r)$-dominator,
there is a path $Q_1$ of length at most $r$ that leads from $a$ to 
some $x \in X$. Let $Q$ be the concatenation of $P(s)$, $Q_0$, and $Q_1$;
$Q$ is a walk of length at most $3r$ that connects $s$ and $x \in X$.

Let $x'$ be the first (closest to $s$) vertex on $Q$ that belongs to $X$;
such a vertex exists as $x \in X$ is on $Q$. 
As the length of $Q$ is at most $3r$, we have $x' \in \pr{3r}{s}{X}$.
Since $s \simeq z$, we have $x' \in \pr{3r}{z}{X}$, and, consequently, $x' \in D'$.
However, by Claim~\ref{cl:dom-via-X}, $x'$ does not lie on $P(s)$. Hence $x'$ lies on the part of $Q$ between $v(s)$ and $x$, but each vertex of this part is at distance at most $r$ from $a$ on $Q$. Thus
$a$ is $r$-dominated by $x'$, a contradiction.
\cqed\end{proof}

%Suppose first that $Q_1$ passes from some vertex of $X$; let $w$ be the first such vertex (counting from $s$). By Claim~\ref{cl:dom-via-X}, $w$ must lie on $Q_0$. As $w$ is the first vertex of $V(Q_1)\cap X$ on $Q_1$, the prefix of $Q_1$ from $s$ to $w$ certifies that $w\in \pr{2r}{s}{X}$. But since $s\simeq z$, we have $\pr{2r}{s}{X}=\pr{2r}{z}{X}\subseteq D'$, and hence $w\in D'$. Hence, the suffix of $Q_0$ from $w$ to $a$ is path of length at most $r$ connecting a vertex of $D$ with $a$. This shows that $a$ is $r$-dominated by $D'$.

%Suppose second that $Q_1$ does not pass through any vertex of $X$. Recall that $X$ is a $(Z,r)$-dominator. Hence, as $a\in Z$, there exists some vertex $w\in X$ and a path $Q_2$ of length at most $r$ that starts in $a$, ends in $w$, and whose internal vertices do not belong to $X$. Let $Q_3$ be the concatenation of $Q_1$ and $Q_2$. Then $Q_3$ certifies that $w\in \pr{3r}{s}{X}$. But since $s\simeq z$, we have $\pr{3r}{s}{X}=\pr{3r}{z}{X}\subseteq D'$. Hence $w\in D'$ and $a$ is $r$-dominated by $D'$.

As $|D'|<|D|$, Claim~\ref{cl:exchange} is a contradiction with the assumption that $D$ is a minimum-size $(Z',r)$-dominator. This concludes the proof.
\end{proof}

Lemma~\ref{lem:irrelevant} finishes the proof of Lemma~\ref{lem:reduce-corce}: we set $z$ to be any vertex of $\kappa$. 

% ######                                                    
% #     # ###### #####  #    #  ####  ##### #  ####  #    # 
% #     # #      #    # #    # #    #   #   # #    # ##   # 
% ######  #####  #    # #    # #        #   # #    # # #  # 
% #   #   #      #    # #    # #        #   # #    # #  # # 
% #    #  #      #    # #    # #    #   #   # #    # #   ## 
% #     # ###### #####   ####   ####    #   #  ####  #    #

\subsection{Reducing dominators}\label{sec:dominators}

In the rest of this section we work with arbitrary $r$ towards the proof of Theorems~\ref{thm:mainrbe-dominator} and~\ref{thm:mainrbe-reduced}. At some point we will argue that for $r=1$, the statement of Theorem~\ref{thm:main-be} is immediate. For convenience, we recall the statements of the results we are going to prove.

\restatemainbe*
\restatemainbedominator*
\restatemainbereduced*

Having reduced the number of vertices whose domination is essential,
we arrive at the situation where the vast majority of vertices serve
only the role of dominators, or, when $r>1$, they serve as connections between dominators with dominatees. 
Now, it is relatively easy to reduce
the number of candidate dominators in one step. This immediately gives the sought kernel for $r=1$, i.e., proves Theorem~\ref{thm:main-be}. For $r>2$, the treatment of vertices connecting dominators and dominatees without introducing additional gadgets turns out to be problematic. Therefore, we are unable to give a kernel that is an induced subgraph of the original graph, and we resort to the statements of Theorems~\ref{thm:mainrbe-dominator} and~\ref{thm:mainrbe-reduced}. 

  We proceed to the proof of Theorem~\ref{thm:mainrbe-dominator}. The algorithm works as follows.  First, we apply the algorithm of
  Theorem~\ref{thm:dcore} to compute a small domination core in the
  graph.  In case the algorithm gives a negative answer, we output
  that $\ds_r(G) > k$.  Hence, from here on, we assume that we have
  correctly computed an $r$-domination core $Z_0\subseteq V(G)$ of size at
  most $\domcoresize(\nabla(\mc G)) \cdot k$.

  Compute $Z=\cl{r}{Z_0}$ using Lemma~\ref{lem:closure}; then we have that $|Z|\leq 3r\grad_{r-1}(\mc G)|Z_0|=\Oh(k)$. Observe that in any graph, any superset of an $r$-domination core is also an $r$-domination core; this follows easily from the definition. Consequently, $Z$ is an $r$-domination core in $G$.

Partition $V(G)\setminus Z$ into equivalence classes with respect to the following relation $\simeq$, defined similarly as in Section~\ref{sec:finding}: For $u,v\in V(G)\setminus Z$, set:
$$
u\simeq v\quad \Leftrightarrow \quad \pr{i}{u}{Z} = \pr{i}{v}{Z}\textrm{ for each $1\leq i\leq r$.}
$$
  From Lemma~\ref{lem:closure} we know that for each $u\in V(G)\setminus Z$, it holds that $|\pr{i}{v}{Z}|\leq 9r\grad_{r-1}(\mc G)^2$. Moreover, Lemma~\ref{lem:fewneighbourhoods} implies that the number of possible different projections $\pr{r}{u}{Z}$ for $u\in V(G)\setminus Z$ is at most $c\cdot |Z|$, for some constant $c$ depending on the grads of $\mc G$. Hence, using the same reasoning as in the proof of Lemma~\ref{lem:num-classes} we obtain the following.

\begin{claim}\label{cl:num-classes}
For $C=c\cdot r^{9r\grad_{r-1}(\mc G)^2}$, the equivalence relation $\simeq$ has at most $C \cdot |Z|$ classes.
\end{claim}

Construct set $Y$ as follows: start with $Z$ and, for each equivalence class~$\kappa$ of relation $\simeq$, add an arbitrarily selected member $v_\kappa$ of $\kappa$. Hence we have that $|Y|\leq (C+1)\cdot |Z|$, so in particular $|Y|=\Oh(k)$.

\begin{claim}\label{cl:containedY}
There exists a minimum-size $r$-dominating set in $G$ that is contained in $Y$. 
\end{claim}
\begin{proof}
Let $D$ be a minimum-size $r$-dominating set in $G$, so $|D|=\ds_r(G)=\ds_r(G,Z)$ (because $Z$ is an $r$-domination core). It follows
  that~$D$ is a minimum-size $(Z,r)$-dominator as well. We construct~$D'$ by replacing $\kappa \cap D$ with~$v_\kappa$ for each class~$\kappa$ of $\simeq$
  that has a nonempty intersection with~$D$. Clearly, $|D'| \leq |D| = \ds(G,Z)$ and $D'\subseteq Y$.  Moreover,~$D'$ is
  still a $(Z,r)$-dominator in~$G$. Indeed, the definition of $\simeq$ implies that the representative vertex~$v_\kappa$
  $r$-dominates exactly the same vertices in~$Z$ as any other vertex of $D \cap \kappa$. Therefore, since $|D'|\leq \ds_r(G,Z)$, it must hold that~$D'$ is a
  minimum-size $(Z,r)$-dominator in~$G$ and $|D'|=\ds_r(G,Z)$.  Since~$Z$ is an $r$-domination core, we infer that~$D'$ is an $r$-dominating set in~$G$, and obviously $D'\subseteq Y$.
\cqed\end{proof}

The proof of Theorem~\ref{thm:main-be} now follows from the following simple claim.

\begin{claim}
If $r=1$, then $\ds(G)\leq k$ if and only if $\ds(G[Y])\leq k$.
\end{claim}
\begin{proof}
If $\ds(G)\leq k$, then by Claim~\ref{cl:containedY} there is also a dominating set $D$ of $G$ that has size at most $k$ and is contained in $Y$. Then $D$ is also a dominating set of $G[Y]$, and hence $\ds(G[Y])\leq k$.

If $\ds(G[Y])\leq k$, then there is a set $D'\subseteq Y$ that has size at most $k$ and dominates $Y$ in $G$. As $Z\subseteq Y$, it is also a $(Z,r)$-dominator in $G$. As $Z$ is an $r$-domination core, we infer that $\ds(G)=\ds(G,Z)\leq |D'|\leq k$.
\cqed\end{proof}

For Theorem~\ref{thm:mainrbe-dominator}, we run the algorithm of Lemma~\ref{lem:shpclo} on set $Y$, and let $W=Y'$ be the obtained superset of $Y$. By Lemma~\ref{lem:shpclo} we have that $|W|\leq C'\cdot |Y|$ for some constant $C'$, so in particular $|W|=\Oh(k)$. Then Theorem~\ref{thm:mainrbe-dominator} follows immediately from the following verification.

\begin{claim}
$\ds_r(G)\leq k$ if and only if $\ds_r(G[W],Z)\leq k$.
\end{claim}
\begin{proof}
If $\ds_r(G)\leq k$, then by Claim~\ref{cl:containedY} there is also an $r$-dominating set $D$ of $G$ that has size at most $k$ and is contained in $Y$. By Lemma~\ref{lem:shpclo} (property~\eqref{pr:sh-corr}), whenever some vertex $u\in D\subseteq Y$ $r$-dominates some other vertex $v\in Y$ in $G$, then $u$ also $r$-dominates $v$ in $G[W]$. Since $Z\subseteq Y$, we infer that $D$ is an $(Z,r)$-dominator in $G[W]$, and consequently $\ds_r(G[W],Z)\leq k$.

If $\ds_r(G[W],Z)\leq k$, then there is a set $D'\subseteq W$ that has size at most $k$ and $r$-dominates $Z$ in $G[W]$. Then in particular $D'$ is also a $(Z,r)$-dominator in $G$. As $Z$ is an $r$-domination core, we infer that $\ds_r(G)=\ds_r(G,Z)\leq |D'|\leq k$.
\cqed\end{proof}

Finally, Theorem~\ref{thm:mainrbe-reduced} follows from Theorem~\ref{thm:mainrbe-dominator} by applying the reduction of Lemma~\ref{lem:dominator-reduction} to $G[W]$ and set $Z$.

\section{A kernel for nowhere dense graphs}\label{sec:nowhere-dense}

\newcommand{\excbic}{c}
\newcommand{\excbicone}{c'}

In this section we give an almost linear kernel for {\sc{Dominating Set}} in nowhere dense graph classes. In other words, we prove Theorem~\ref{thm:main-nd}. The proof will follow the high-level strategy that was used in Section~\ref{sec:bnd-exp}, but will be technically more complicated. The main reason is that in the nowhere dense setting we lack the analogue of Lemma~\ref{lem:closure}, which was a crucial tool for simplifying the analysis of the instance once sets $X$ and $S$ are constructed, and for lifting the result to {\sc{$r$-Dominating Set}}. Therefore, the proof in this section will be only for $r=1$, and additional tools specific for nowhere dense graph classes will be necessary. 

The proof contained in this section actually follows closely our initial approach for {\sc{Dominating Set}} on classes of bounded expansion, which can be found in an earlier version of this work~\cite{previous}. While the following presentation will be self-contained, the reader might find it useful to read the description contained in~\cite{previous} for bounded expansion classes before approaching this section.

\subsection{Additional preliminaries for nowhere dense graph classes}\label{sec:prelim-nd}

% #     #                                              ######                              
% ##    #  ####  #    # #    # ###### #####  ######    #     # ###### #    #  ####  ###### 
% # #   # #    # #    # #    # #      #    # #         #     # #      ##   # #      #      
% #  #  # #    # #    # ###### #####  #    # #####     #     # #####  # #  #  ####  #####  
% #   # # #    # # ## # #    # #      #####  #         #     # #      #  # #      # #      
% #    ## #    # ##  ## #    # #      #   #  #         #     # #      #   ## #    # #      
% #     #  ####  #    # #    # ###### #    # ######    ######  ###### #    #  ####  ###### 
                                                           
In this section we introduce auxiliary definitions and facts about
nowhere dense graph classes that will be needed throughout this section. First, we recall the notion of {\em{weak colorings}}.

\paragraph*{Weak colorings.}
For a graph $G$, let $\Pi(G)$ denote the set of all linear orderings of~$V(G)$.  Given a
graph~$G$, an integer~$r$ and an ordering $\sigma \in \Pi(G)$, we say
that a vertex~$u$ is \emph{weakly $r$-accessible} from a vertex~$v$
in~$\sigma$ if $u<_{\sigma} v$ and there is a path~$P$ of length at
most~$r$ with endpoints~$u$ and~$v$ such that every internal
vertex~$w$ on~$P$ has the property that $u <_{\sigma} w$.  We denote
by $B^{G,\sigma}_r(v)$ the set of vertices that are weakly
$r$-accessible from~$v$ in~$\sigma$.  When~$G$ is clear from context,
we drop it from the superscript and write $B^\sigma_r(v)$.

\begin{definition}[Weak $r$-coloring number]
  \label{def:wcol}
  The \emph{weak $r$-coloring} number of a graph~$G$ is defined as
  \[
  \wcol_r(G) = 1 + \min_{\sigma \in \Pi(G)} \max_{v \in V(G)} | B^\sigma_r (v) | .
  \]
\end{definition}

The weak coloring number of a graph is related to its grads.
We shall need the following upper bound, which follows from
\cite[Proposition~4.8 and Theorem~7.11]{Sparsity}:

\begin{lemma}[\cite{Sparsity}]\label{lem:wcol-grad}
  For any graph $G$, it holds that $\wcol_2(G) \leq (8\grad_1(G)^3
  +1)^2$.
\end{lemma}

\newcommand{\dvorakfun}{f_{\textrm{dv}}}  % f_dv(eps) n^eps is the apx ratio
\newcommand{\neifun}{f_{\textrm{nei}}}    % f_nei(eps)*|X|*n^eps is the number of neighborhoods
\newcommand{\chrgfun}{f_{\textrm{chrg}}}  % f_chrg(eps) n^eps is the upper bound on the number of charges
\newcommand{\gradfun}{f_{\nabla}}           % f_grad(eps,r) n^eps is the upper bound on \grad_r(G) 
\newcommand{\wcolfun}{f_{\wcol}}            % f_wcol(eps,r) n^eps is the upper bound on \wcol_r(G) 
\newcommand{\cnumfun}{f_{\numcliques}}            % f_#omega(eps,r) n^1+eps is the upper bound on \#omega(G \nab r)

\paragraph*{Nowhere dense classes.} We first introduce the definition of a nowhere dense graph class; recall that $\omega(G)$ denotes the size of the largest clique in $G$ and $\omega(\mc G)=\sup_{G\in \mc G} \omega(G)$.

\begin{definition}[Nowhere dense]\label{def:nd}
  A graph class $\mc G$ is \emph{nowhere dense} if there exists a
  function $f_\omega \colon \N \to \N$ such that for
  all~$r$ we have that $\omega(\mc G \nab r) \leq f_\omega(r)$.
\end{definition}

Since cliques have non-constant density, we have that every class of bounded expansion is also nowhere dense; however, the converse is not true~\cite{Sparsity}. We shall mostly rely on the following alternative characterization of nowhere dense graph classes, which follows easily from the following results of~\cite{NesetrilM11a}: Theorem 4.1, points (ii) and (x), and Corollary 4.3.

\begin{proposition}[\cite{NesetrilM11a}]\label{prop:nd-parameter}
Let $\mc G$ be a nowhere dense graph class. Then: 
\begin{enumerate}[(1)]
\item\label{grad} There is a function $\gradfun(r,\eps)$ such that $\grad_0(G')\leq \gradfun(r,\eps)\cdot |G'|^\eps$ for every integer $r\geq 0$, $G'\in \mc G\nab r$, and real $\eps>0$. In particular, $\grad_r(G)\leq \gradfun(r,\eps)\cdot |G|^\eps$ for every integer $r\geq 0$, $G\in \mc G$, and real $\eps>0$.
\item\label{wcol} There is a function $\wcolfun(r,\eps)$ such that $\wcol_r(G)\leq \wcolfun(r,\eps)\cdot |G|^\eps$ for every integer $r\geq 0$, $G\in \mc G$, and real $\eps>0$.
\end{enumerate}
\end{proposition}
As shown in \cite{NesetrilM11a}, conditions (\ref{grad}) and (\ref{wcol}) are in fact equivalent to $\mc G$ being nowhere dense, provided that $\mc G$ is closed under taking subgraphs.

We remark that in the other literature on the topic, it is customary to use an alternative variant of this statement: for instance, there exists a constant $N^{\wcol}_{r,\eps}$ such that $\wcol_r(G)\leq |G|^\eps$ for any integer $r$, real $\eps$ and graph $G\in \mc G$ such that $|G|\geq N^{\wcol}_{r,\eps}$; see e.g. \cite[Lemma 5.3]{GroheKS14}. Whereas this formulation can be easily seen to be equivalent to ours, we find it more cumbersome to use in the proofs.

\paragraph*{Clique density.} It turns out that the \emph{clique density}, \ie the number of complete subgraphs
in a graph divided by the size of the graph, is an important measure that
that determines the structure of nowhere dense graphs. Recall that 
$\numcliques(G)$ denotes the total number of cliques in $G$.  

\begin{lemma}[Clique density of nowhere dense graph]\label{lem:nd-cliques}
Let $\mc G$ be a nowhere dense class of graphs. Then there exists a function $\cnumfun(r,\eps)$ such that for any $G\in \mc G$, integer $r\geq 0$ and real $\eps>0$, we have that $\numcliques(G\nab r)\leq \cnumfun(r,\eps)\cdot |G|^{1+\eps}$.
\end{lemma}
\begin{proof}
Take any $H\in G\nab r$; of course, $|H|\leq |G|$. Since $G\in \mc G$, we have that $\omega(H)\leq f_\omega(r)$. By Proposition~\ref{prop:nd-parameter}, point (\ref{wcol}) applied to $r=1$ and $\eps'=\eps/(f_\omega(r)-1)$, there exists an ordering $\sigma\in \Pi(H)$ such that for each $v\in V(H)$ we have that $B^{H,\sigma}_1(v)=\{u\ \colon\ u<_\sigma v\wedge uv\in E(H)\}$ has size at most $\wcolfun(1,\eps')\cdot |H|^{\eps'}\leq \wcolfun(1,\eps')\cdot |G|^{\eps'}$. For each clique $Q\subseteq V(H)$, let $v_Q$ be the last vertex of $Q$ in $\sigma$. Then we have that $Q\subseteq B^{H,\sigma}_1(v_Q)$. Therefore, for each $v\in V(H)$ we have that the number of cliques $Q\subseteq V(H)$ with $v=v_Q$ is at most
$$\sum_{d=0}^{f_\omega(r)-1} |B^{H,\sigma}_1(v)|^d\leq f_\omega(r)\cdot |B^{H,\sigma}_1(v)|^{f_\omega(r)-1}\leq f_\omega(r)\cdot \wcolfun(1,\eps')^{f_\omega(r)-1}\cdot |G|^{\eps}.$$
The claim follows by summing through all the vertices of $H$ and using the fact that $|H|\leq |G|$. 
\end{proof}

We now use the following result from \cite[Lemma~6.6, arxiv version v2]{BndExpKernels13}
that relates the structure of bipartite graphs to the edge- and
clique-density of its respective graph class. %This is a nowhere dense analogue of Proposition~\ref{prop:twin-classes}.

\begin{proposition}\label{prop:bipartite-general}
  Let $G=(X,Y,E)$ be a bipartite graph, and let $\mc G_1$ be the family of $1$-shallow minors of $G$ that have at most $|X|$ vertices. Let further $h = \max_{H \in \mc G_1}
  ({\numcliques(H)}/{|H|})$. Then there are at most
  \begin{enumerate}
    \item $2\grad_0(\mc G_1)\cdot|X|$ vertices in $Y$ with degree 
      larger than $\omega(\mc G_1)$;
    \item $\big(h+2\grad_0(\mc G_1)\big)\cdot|X|$ 
      subsets $A \subseteq X$ such that $A = N(u)$ for some 
      $u \in Y$.
  \end{enumerate}
\end{proposition}

With these tools at hand, we can prove the following important
lemma that serves the role of Lemma~\ref{lem:twin-classes}
in the nowhere dense case.

\begin{lemma}[Twin classes]\label{lem:twin-classes2}
  Let $\mc G$ be nowhere dense graph class. Then there exists a function $\neifun(\cdot)$
  such that for any graph $G\in \mc G$, any nonempty vertex subset $X\subseteq V(G)$ and any $\eps > 0$, the following holds:
  \[
  | \{ A\subseteq X \colon \exists_{v\in V\setminus X}\ A=N_X(v) \}| \leq \neifun(\eps)\cdot |X|^{1+\eps}.
  \]
\end{lemma}
\begin{proof}
  We would like to use the second bound of Proposition~\ref{prop:bipartite-general}. Fix $\eps > 0$,
  a graph $G \in \mc G$ and a nonempty vertex set $X \subseteq G$. Let $G_0$ be the bipartite graph $(X,V(G)\setminus X,E(G)\cap (X\times (V(G)\setminus X)))$.
  To obtain the sought bound, we need bounds on the 
  quantities $h:=\sup_{H \in \mc G_1}(\numcliques(H)/|H|)$
  and $\grad_0(\mc G_1)$, where $\mc G_1$ is defined for $G_0$ as in Proposition~\ref{prop:bipartite-general}.

  Since $G_0$ is a subgraph of $G$, we have that $\mc G_1\subseteq G\nab 1$. Hence, from Lemma~\ref{lem:nd-cliques} we obtain
  \begin{equation}\label{twin1}
    h =\sup_{H \in \mc G_1}\frac{\numcliques(H)}{|H|} 
    \leq \sup_{\substack{H \in G \nab 1 \\ |H| \leq |X|}}
           \frac{\cnumfun(1,\eps)|H|^{1+\eps}}{|H|}
    \leq f_{\numcliques}(1,\eps) \cdot |X|^\eps.
  \end{equation}
  The bound of the grad follows directly from Proposition~\ref{prop:nd-parameter}, point (\ref{grad}): 
  \begin{equation}\label{twin2}
    \grad_0( \mc G_1 ) = \sup_{\substack{H \in G \nab 1 \\ |H| \leq |X|}} \grad_0(H)
    \leq \gradfun(1,\eps) \cdot |X|^\eps.
  \end{equation}
  By plugging (\ref{twin1}) and (\ref{twin2}) in upper bound of Proposition~\ref{prop:bipartite-general} (2), we obtain that
  \[
    | \{ A\subseteq X \colon \exists_{v\in V\setminus X}\ A=N_X(v) \}| \leq (h + 2\grad_0(\mc G_1)) \cdot |X|
      \leq (\cnumfun(1,\eps) + 2\gradfun(1,\eps)) \cdot |X|^{1+\eps}.
  \]
  Hence we can set $\neifun(\eps)=\cnumfun(1,\eps) + 2\gradfun(1,\eps)$.
\end{proof}

\begin{comment}
\begin{figure}[t]
  \begin{center}
   \includegraphics[scale=1]{figs/charging.pdf}
   \caption{\label{fig:charging} An impossible situation in a $4$-centered coloring.}
  \end{center}
\end{figure}
\end{comment}

Using Lemma~\ref{lem:twin-classes2}, we can now prove the following result, which intuitively says that not only the number of neighborhoods in a graph from a nowhere dense class is small, but also these neighborhoods are ``uniformly distributed''. This lemma is a nowhere dense analogue of a result we called the ``charging lemma'' for bounded expansion classes, and which was used in the previous version of this work~\cite{previous}. Due to the introduction of Lemma~\ref{lem:closure}, the charging lemma for bounded expansion classes is no longer needed, and hence we omit it. However, in the nowhere dense case we still need this result.

\begin{lemma}\label{lem:orient2}
  Let $\mc G$ be a nowhere dense graph class. Then there exists a function $\chrgfun(\cdot)$ such that the following holds. For any $\eps>0$ and any bipartite graph $G=(X,Y,E)\in \mc G$ such that every vertex from $Y$ has a nonempty neighborhood in $X$ and no two vertices of $Y$ have the same neighborhood in $X$, there exists a mapping $\phi \colon Y \to X$ with the following properties:
  \begin{itemize}
    \item $u\phi(u)\in E$ for each $u\in Y$;
    \item $|\phi^{-1}(v)| \leq \chrgfun(\eps) \cdot |G|^{\eps} $ for each $v\in X$.
  \end{itemize}
\end{lemma}
\begin{proof}
Without loss of generality assume that $\mc G$ is closed under taking subgraphs, since otherwise we can consider the closure of $\mc G$ under this operation, which is also nowhere dense.

Let us fix $G=(X,Y,E)$ and $\eps>0$. Using Lemma~\ref{lem:wcol-grad} we infer that there exists an ordering $\sigma\in \Pi(G)$ such that for every vertex $v$, we have $|B^\sigma_2(v)|\leq (8\grad_1(G)^3+1)^2$. By applying Proposition~\ref{prop:nd-parameter}, point (\ref{grad}), for $r=1$ and $\eps/12$, we obtain that $|B^\sigma_2(v)|\leq f_0(\eps)\cdot |G|^{\eps/2}$, for some value $f_0(\eps)$ depending on $\gradfun(1,\eps/12)$.

Construct~$\phi\colon Y\to X$ as follows: for every $u \in Y$, set~$\phi(u)$ to that
  vertex of~$N(u)$ that is last in~$\sigma$; note that the validity of
  this definition is asserted by the assumption that $Y$ does not
  contain isolated vertices.  The first condition is trivially
  satisfied by~$\phi$, so we proceed to proving the second one.

Fix a vertex~$v\in X$ and consider all the vertices~$u$ with $\phi(u) = v$. Let 
$$U_v^-=\{u\ \colon\ u\in Y\,\wedge\, \phi(u)=v\,\wedge\, u<_\sigma v\}\textrm{ and }U_v^+=\{u\ \colon\ u\in Y\,\wedge\, \phi(u)=v\,\wedge\, v<_\sigma u\}.$$
Clearly we have that $U_v^-\subseteq B^\sigma_1(v)$ and hence
$$|U_v^-|\leq |B^\sigma_1(v)|\leq |B^\sigma_2(v)|\leq f_0(\eps)\cdot |G|^{\eps/2}.$$ 
Since $v=\phi(u)$ was chosen to be the last vertex of $N(u)$ in $\sigma$, for every vertex $u\in U_v^+$ we have that $N(u)\subseteq B^\sigma_2(v)\cup \{v\}$. Since every pair of vertices in $Y$ have different neighborhoods in $X$, we can apply Lemma~\ref{lem:twin-classes2} to the bipartite graph induced in $G$ between $B^\sigma_2(v)$ and $U_v^+$ (note that this graph belongs to $\mc G$ since $\mc G$ is closed under taking subgraphs) and parameter $1$, and conclude that 
$$|U_v^+|\leq \neifun(1)\cdot |B^\sigma_2(v)|^2\leq \neifun(1)\cdot f_0(\eps)^2\cdot |G|^{\eps}.$$
Concluding,
$$|\phi^{-1}(v)|=|U_v^-|+|U_v^+|\leq f_0(\eps)\cdot |G|^{\eps/2}+\neifun(1)\cdot f_0(\eps)^2\cdot |G|^{\eps}.$$
Hence we can take $\chrgfun(\eps)=f_0(\eps)+\neifun(1)\cdot f_0(\eps)^2$.
\end{proof}

Finally, we state the variant of \Dvorak's algorithm suitable for nowhere dense graphs. The following lemma follows directly from plugging the bound of Proposition~\ref{prop:nd-parameter}, point (\ref{grad}), into Lemma~\ref{lem:better-dvorak} (for $r=1$).

\begin{lemma}\label{lem:better-dvorak-nd}
Let $\mc G$ be a nowhere dense class of graphs. Then there exists a function $\dvorakfun(\cdot)$ and a polynomial-time algorithm that, given a graph $G\in \mc G$, a vertex subset $Z\subseteq V(G)$ and an integer $k$, finds either:
\begin{itemize}
  \item a $Z$-dominator in $G$ whose size is bounded by $\dvorakfun(\eps)\cdot k\cdot |G|^\eps$, for every $\eps>0$, or 
  \item a subset of~$Z$ of size at least~$k+1$ that is $2$-scattered in $G$.
\end{itemize}
\end{lemma}

\subsection{Setting up the proof}

We proceed to the proof of Theorem~\ref{thm:main-nd}. From now on, we assume that $\mc G$ is a fixed nowhere dense graph class. Without loss of generality we assume that $\mc G$ is closed under taking subgraphs, since otherwise we may consider the closure of $\mc G$ under this operation, which is also nowhere dense. We fix all the functions given by Proposition~\ref{prop:nd-parameter} and Lemmas~\ref{lem:twin-classes2},~\ref{lem:orient2},~\ref{lem:better-dvorak-nd} for the class $\mc G$. Observe that the class $\mc G \nab 1$ is also nowhere dense, hence we can apply these results also to this class. We therefore fix also the functions given by Proposition~\ref{prop:nd-parameter} and Lemmas~\ref{lem:twin-classes2},~\ref{lem:orient2},~\ref{lem:better-dvorak-nd} for $\mc G \nab 1$, and we shall denote them by $\gradfun^1(\cdot,\cdot)$, $\neifun^1(\cdot)$, $\chrgfun^1(\cdot)$ etc. Moreover, since $\mc G$ is nowhere dense, there exist constants $\excbic$ and $\excbicone$ such that $K_{\excbic,\excbic}\notin \mathcal{G}\nab 0$ and $K_{\excbicone,\excbicone}\notin \mathcal{G}\nab 1$; in the following we shall use these constants extensively.

We also fix the real value $\eps>0$ for which the algorithm is constructed. Recall that Theorem~\ref{thm:main-nd} asserts the existence of an algorithm for each fixed value of $\eps$, and not an algorithm that gets $\eps$ on the input. Thus, the values of functions given by Proposition~\ref{prop:nd-parameter} and Lemmas~\ref{lem:twin-classes2},~\ref{lem:orient2},~\ref{lem:better-dvorak-nd} for classes $\mc G$ and $\mc G\nab 1$ applied to any fixed $\eps'$ depending on $\eps$ can be hard-coded in the algorithm, and do not need to be computed. If we would like to implement one algorithm that works for $\eps$ given on the input, then we would need to assume that class $\mc G$ is {\em{effectively nowhere dense}}, that is, that function $f_\omega(r)$ in Definition~\ref{def:nd} is computable~\cite{DawarK09}. Then we would be able to derive that all the functions introduced in Section~\ref{sec:prelim-nd} are computable as well.

Let $(G,k)$ be the input instance of {\sc{Dominating Set}} such that $G\in \mathcal{G}$. We denote $n=|G|$. 

\subsection{Reducing dominatees}

Exactly as in Section~\ref{sec:bnd-exp}, we are going to reduce the number of vertices whose domination is essential in the graph to almost linear in terms of $k$. More formally, we are going to find domination core that has size bounded by $g(\eps)\cdot k\cdot n^{\eps}$, for some function $g(\cdot)$ and every $\eps>0$. In this proof we shall use the same definition of a domination core as in Section~\ref{sec:bnd-exp}, but restricted to $r=1$:

\begin{definition}
  Let~$G$ be a graph and~$Z$ be a subset of vertices.  We say that~$Z$
  is a \emph{domination core} in~$G$ if % the following holds:
  every minimum-size $Z$-dominator in~$G$ is also a dominating set
  in~$G$.
\end{definition}

Mirroring Theorem~\ref{thm:dcore}, we prove the following result:

\begin{theorem}\label{thm:dcore-nd}
There exists a function $g(\cdot)$ such that for every $\eps > 0$ there exists a polynomial-time algorithm that, given an instance $(G,k)$ where $G\in \mathcal{G}$, either correctly concludes that $\ds(G) > k$, or finds a domination core $Z \subseteq V(G)$ with $|Z| \leq g(\eps)\cdot k\cdot n^{\eps}$.
\end{theorem}

Again, the proof of Theorem~\ref{thm:dcore-nd} follows trivially from iterative application of the following lemma that enables us to identify a vertex that can be safely removed from the domination core.

\begin{lemma}\label{lem:reduce-corce-nd}
There exists a function $g(\cdot)$ such that for every $\eps > 0$ there exists a polynomial-time algorithm that, given a vertex subset $Z\subseteq V(G)$ with $|Z|>g(\eps)\cdot k \cdot n^\eps$ and a promise that $Z$ is a domination core, either correctly concludes that $\ds(G) > k$, or finds a vertex $z\in Z$ such that $Z\setminus \{z\}$ is still a domination core.
\end{lemma}

From now on we focus on proving Lemma~\ref{lem:reduce-corce-nd}. We fix the constant $\eps>0$ given to the algorithm; without loss of generality we assume that $\eps<1/10$. That is, all the constants introduced in the sequel may depend on $\eps$.

\subsubsection{Iterative extraction of $Z$-dominators}\label{sec:extraction-nd}

We now present the analogue of the subroutine presented in Section~\ref{sec:extraction} for the nowhere dense case. Due to lack of Lemma~\ref{lem:closure}, the implementation will be quite different. In particular, the argument that the procedure finishes after a constant number of iteration is based on a different principle, suited for nowhere dense classes. The fact that we are currently unable to lift this argument to an arbitrary radius $r$ is the main limitation for proving an almost linear kernel for {\sc{$r$-Dominating Set}} on nowhere dense graph classes.

\begin{comment}
As in the bounded expansion case, the goal is to find a pair of disjoint subsets $X$ and $S$ with the following properties: $X$ is bounded linearly in terms of $k\cdot n^{\eps/2}$, whereas $S$ is $2$-scattered in $G-X$ and is at least $C\cdot n^{\eps'}$ times larger than $X$, for some $\eps'>0$ and a constant $C$ chosen as large as we like. If we now generalize the reasoning of Section~\ref{sec:extraction} directly to the nowhere dense case, then every consecutive $Z$-dominator $X_i$ would be $f(\delta)\cdot n^{\delta}$ times larger than the previous one, for any $\delta>0$. As $\grad_0(G)$ is not bounded by a constant anymore, in a direct generalization we would have problems with proving the analogue of Lemma~\ref{lem:extraction-terminates}: the argument that the construction terminates after a constant number of iterations breaks per se. We therefore replace it with a different argument based on discovering a large biclique subgraph in case the procedure runs for too many iterations.
\end{comment}

Let $\delta=\frac{\eps}{4\excbic}>0$ and let us fix some constant $C$, to be decided later. First, we apply the algorithm of Lemma~\ref{lem:better-dvorak-nd} to $G$, $Z$, and parameters $k$ and $\delta$. This algorithm either outputs a subset $S\subseteq Z$ such that $|S|>k$ and $S$ is $2$-scattered in $G$, or a $Z$-dominator $X_1$ such that $|X_1|\leq \dvorakfun(\delta)\cdot k\cdot n^{\delta}$. In case $S$ is found, every vertex of $G$ can dominate at most one vertex of $S$ and thus we can conclude that $\ds(G,Z)>k$. As $\ds(G,Z)=\ds(G)$, we infer that $\ds(G)>k$ and we can terminate the algorithm and provide a negative answer. Hence, from now on we assume that the $Z$-dominator $X_1$ has been successfully constructed.

Now, we inductively construct $Z$-dominators $X_2,X_3,X_4,\ldots$ such that $X_1\subseteq X_2\subseteq X_3\subseteq X_4\subseteq \ldots$. We maintain the invariant that
$$|X_i|\leq C_i\cdot k\cdot n^{(2i-1)\delta},$$
where constants $C_i$ are defined as 
$$C_i := (1 + \dvorakfun(\delta)\cdot C)^{i-1}\cdot \dvorakfun(\delta).$$
Observe that $|X_1|\leq \dvorakfun(\delta)\cdot k\cdot n^{\delta}$, which means that the invariant is satisfied at the first step. We now describe how $X_{i+1}$ is constructed based on $X_i$ for consecutive $i=1,2,3,\ldots$.

\begin{enumerate}
\item\label{ap-dvorak-nd} First, apply the algorithm of
  Lemma~\ref{lem:better-dvorak-nd} to graph $G-X_i$, set
  $Z\setminus X_i$, and parameter $C \cdot |X_i|\cdot n^{\delta}$.

\item\label{big-scattered-nd} Suppose the algorithm has found a set $S\subseteq Z\setminus X_i$
  that is $2$-scattered in $G\setminus X_i$ and has cardinality larger than
  $C \cdot |X_i|\cdot n^{\delta}$. We set $X=X_i$, terminate the construction of sets $X_i$ and proceed to the second phase with the pair $(X,S)$.

\item\label{small-domset-nd} Otherwise, the algorithm has found a $(Z\setminus X_i)$-dominator
  $D_{i+1}$ in $G \setminus X_i$ such that
  \begin{align*}
      |D_{i+1}| &\leq  \dvorakfun(\delta)\cdot C \cdot |X_i|\cdot n^{\delta}\cdot n^{\delta}\\
                &= \dvorakfun(\delta)\cdot C \cdot |X_i|\cdot n^{2\delta}
  \end{align*}
  We set $X_{i+1} = X_i \cup D_{i+1}$ and proceed to the next $i$. Observe that
  \begin{align*}
    |X_{i+1}|=|X_i|+|D_{i+1}| &\leq (1 + \dvorakfun(\delta)\cdot C)
                \cdot |X_i|\cdot n^{2\delta} \\
              &\leq (1 + \dvorakfun(\delta)\cdot C)\cdot C_i\cdot k\cdot n^{(2i-1)\delta}\cdot n^{2\delta} \\
              &= C_{i+1}\cdot k\cdot n^{(2(i+1)-1)\delta}.
  \end{align*}
  Hence, the invariant that $|X_i| \leq C_i\cdot k\cdot n^{(2i-1)\delta}$ is maintained in the next iteration.
\end{enumerate}

We now present the analogue of Lemma~\ref{lem:extraction-terminates}: we prove that the construction terminates by outputting a pair $(X,S)$ after a constant number of iterations. Note, however, that the argument is quite different.

\begin{lemma}\label{lem:extraction-terminates-nd}
Assuming that $|Z|>(\excbic\cdot \neifun(\eps/2)+1)\cdot C_{\excbic}\cdot k\cdot n^{\eps}$, the construction terminates by outputting some pair $(X,S)$ after at most $\excbic-1$ iterations, i.e., before constructing $X_{\excbic}$.
\end{lemma}
\begin{proof}
For the sake of contradiction, suppose that the procedure actually performed $\excbic-1$ iterations, successfully constructing disjoint sets $X_1,D_2,D_3,\ldots,D_\excbic$, where $X_i=X_1\cup D_2\cup D_3\cup \ldots\cup D_i$ for $i=1,2,\ldots,\excbic$. Let $Q:=X_\excbic$; then we know that each of these sets is a $(Z\setminus Q)$-dominator. Observe that 
\begin{align*}
|Z\setminus Q| \geq |Z|-|Q|& > (\excbic\cdot \neifun(\eps/2)+1)\cdot C_{\excbic}\cdot k\cdot n^{\eps} - C_{\excbic}\cdot k\cdot n^{(2\excbic-1)\delta}\\
& \geq (\excbic\cdot \neifun(\eps/2)+1)\cdot C_{\excbic}\cdot k\cdot n^{\eps} - C_{\excbic}\cdot k\cdot n^{\eps/2}\\
& \geq \excbic\cdot \neifun(\eps/2)\cdot C_{\excbic}\cdot k\cdot n^{\eps}\geq \excbic\cdot \neifun(\eps/2)\cdot |Q|\cdot n^{\eps/2}.
\end{align*}
Now, partition vertices of $Z\setminus Q$ into classes with respect to their neighborhoods in $Q$. By Lemma~\ref{lem:twin-classes2}, we infer that the number of these classes is at most $\neifun(\eps/2)\cdot |Q|\cdot n^{\eps/2}$. Since $|Z\setminus Q|>\excbic\cdot \neifun(\eps/2)\cdot |Q|\cdot n^{\eps/2}$, we infer that one of these classes $\kappa$ satisfies $|\kappa|\geq \excbic$. However, each member of $\kappa$ has neighbors in each of the $(Z\setminus Q)$-dominators $X_1,D_2,D_3,\ldots,D_\excbic$, and hence the common $Q$-neighborhood of vertices of $\kappa$ is of size at least $\excbic$. Thus we see that the induced subgraph $G[\kappa\cup N_Q(\kappa)]$ contains a $K_{\excbic,\excbic}$ as a subgraph, a contradiction.
\end{proof}

Hence, provided that the cardinality of $Z$ satisfies the lower bound stated in Lemma~\ref{lem:extraction-terminates-nd}, the construction will terminate after at most $\excbic-1$ iterations, thus constructing sets $X$ and $S$ with the following properties:
\begin{itemize}
\item $|X|\leq C_{\excbic-1}\cdot k\cdot n^{\eps/2}$;
\item $X$ is a $Z$-dominator in $G$;
\item $|S|>C\cdot |X|\cdot n^{\delta}$;
\item $S\subseteq Z\setminus X$ and $S$ is $2$-scattered in $G-X$.
\end{itemize}
With sets $X$ and $S$ we proceed to the second phase, that is, finding an irrelevant dominatee. Note that the main difference w.r.t. the proof from Section~\ref{sec:bnd-exp} for the bounded expansion case is that we do not have a bound on the sizes of constant-radius projections of vertices of $V(G)\setminus X$ onto $X$. This will make the forthcoming analysis much more challenging.

\subsubsection{Finding an irrelevant dominatee}

Let us denote $R:=V(G)\setminus X$. Then $S\subseteq Z\cap R$ and $S$ is $2$-scattered in $G[R]$. Therefore, sets $N[s]\cap R$ are pairwise disjoint for all $s\in S$. Recall that for a vertex $u\in R$, the $X$-neighborhood of $u$ is defined as $N_X(u)=N(u)\cap X$. For $W\subseteq R$, we define $N_X(W)=\bigcup_{u\in W} N_X(u)$.

\begin{figure}[htb]
  \centering
  \includegraphics[width=.8\textwidth]{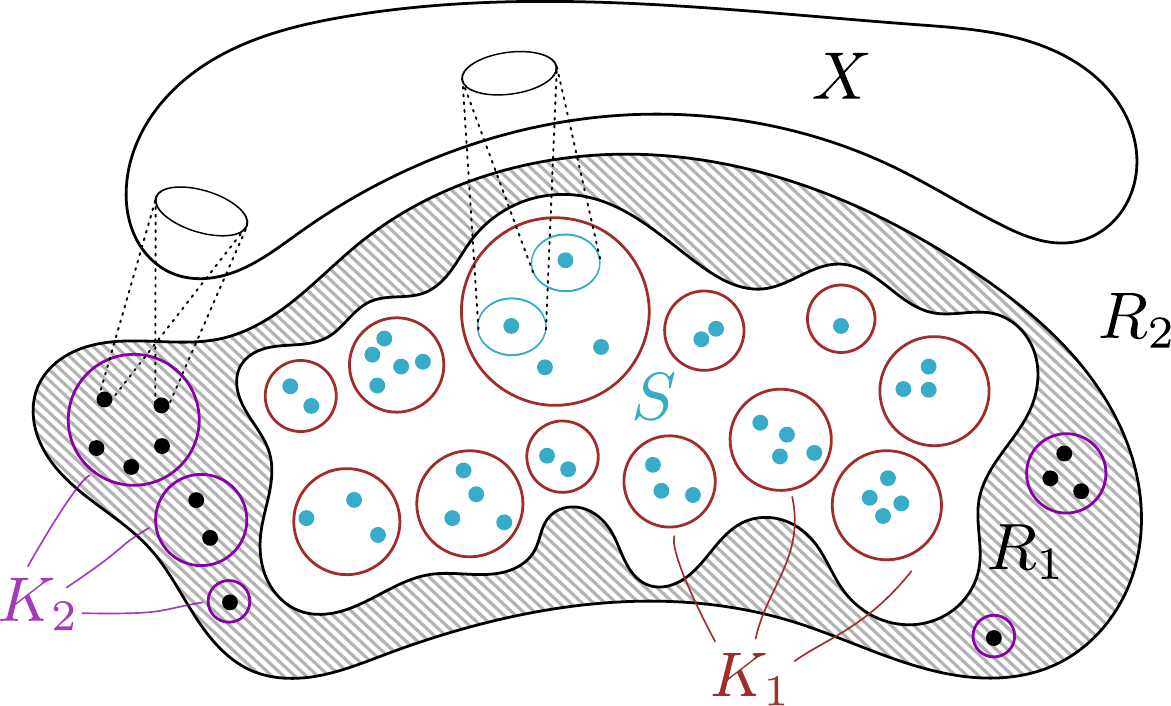}
  \caption{\label{fig:overview} Overview over important vertex sets.}
\end{figure}

We construct an auxiliary graph $G' \in G \nab 1$ as follows: for every
vertex $s \in S$, we contract every vertex of the set $N(s) \setminus
X$ into~$s$.  Since the vertices of~$S$ are $2$-scattered in~$G-X$,
the sets $N(s) \setminus X$ are pairwise disjoint for different $s\in S$
and this operation creates a $1$-shallow minor of~$G$.  The vertex
of~$G'$ onto which the set~$N(s) \setminus X$ is contracted to is renamed 
as~$s$.  We denote by $N'(\cdot)$ and $N'[\cdot]$, respectively,
open and closed neighborhoods of vertices in~$G'$. The
$X$-neighborhoods in~$G'$ are denoted $N'_X(u) = N'(u)\cap X$, for $u
\in V(G')\setminus X$, and $N'_X(W)=\bigcup_{u\in W} N'_X(u)$ for $W\subseteq V(G')\setminus X$.

First, we show that only few vertices of $S$ can have large $X$-neighborhoods in $G'$.

\begin{lemma}\label{lem:S-small-N-nd}
The number of vertices $s \in S$ for which $|N'_X(s)| \geq \excbicone$ holds is at most $\excbicone\cdot \neifun^1(\delta)\cdot |X|\cdot n^{\delta}$.
\end{lemma}
\begin{proof}
Let $S'=\{s\ \colon\ s\in S\wedge |N'_X(s)| \geq \excbicone\}$, and for the sake of contradiction suppose $|S'|>\excbicone\cdot \neifun^1(\delta)\cdot |X|\cdot n^{\delta}$. Consider the graph $G'[S'\cup X]$ and partition the vertices of $S'$ with respect to their $X$-neighborhoods in this graph. As $G'[S'\cup X]\in G\nab 1\subseteq \mc G\nab 1$, by Lemma~\ref{lem:twin-classes2} we infer that the number of these classes is at most $\neifun^1(\delta)\cdot |X|\cdot n^{\delta}$. Hence, one of the classes, say $\kappa$, has cardinality at least $\excbicone$. Since each member of $\kappa\subseteq S'$ has at least $\excbicone$ neighbors in $X$ in graph $G'$, and this $X$-neighborhood is common among the vertices of $\kappa$, we infer that $|N_X'(\kappa)|\geq \excbicone$ and $G'[\kappa\cup N_X'(\kappa)]$ contains a biclique $K_{\excbicone,\excbicone}$ as a subgraph. This is a contradiction with $G'\in \mc G\nab 1$.
\end{proof}

We remove from $S$ all the vertices that have $X$-neighborhoods in $G'$ larger of size at least $\excbicone$. In this manner, Lemma~\ref{lem:S-small-N-nd} ensures us that the size of $S$ shrinks by at most $\excbicone\cdot \neifun^1(\delta)\cdot |X|\cdot n^{\delta}$. Hence, if we set $C:=C_0+\excbicone\cdot \neifun^1(\delta)$ for some $C_0$ to be determined later, then after performing this step we still have that the resulting set has size more than $C_0\cdot |X|\cdot n^{\delta}$. By somewhat abusing the notation, we denote the resulting set also as $S$, and we reconstruct the graph $G'$ according to the new definition of $S$. In this manner, from now on we assume that $|S|>C_0\cdot |X|\cdot n^{\delta}$ and that $|N_X'(s)|<\excbicone$ for each $s\in S$.

Let $R_1 = R \cap N[S]$ be those vertices of~$R$ that can possibly
dominate a vertex in~$S$, and let $R_2 = R \setminus R_1$ be all the
other vertices in~$R$.  We now partition the vertices of $G' - X$ into
classes according to their neighborhoods in~$X$.  Note that by the
construction of~$G'$, we have that $V(G'-X)=S\cup R_2$.  We define the
equivalence relation $\simeq_X$ over $S \cup R_2$ as follows:
\[
u \simeq_X v \Leftrightarrow N'_X(u) = N'_X(v) .
\]
In the following, we consider the quotients (sets of classes of
% sorry for the \!\! -- it's a hack to remove relation spacing
abstraction) $K_1 = S{/\!\!\simeq_X}$ and $K_2 = R_2{/\!\!\simeq_X}$.  We
will also use $K = K_1 \cup K_2$.  Note that since vertices of~$R_2$
are untouched during the construction of~$G'$, we have that~$K_2$ is
simply the partitioning of vertices of~$R_2$ with respect to their
$X$-neighborhoods in~$G$.  Each $\kappa \in K$ will simply be called a
\emph{class}.  For a class $\kappa\in K$, by $N'_X(\kappa)$ we denote
the common $X$-neighborhood of vertices of~$\kappa$ in~$G'$.

Observe that each class $\kappa\in K_1$ consists of vertices from $S\subseteq Z$,
which, since $X$ is a $Z$-dominator, have to have neighbors in $X$ in
graph $G$.  Hence, $N'_X(\kappa)$ is nonempty for each $\kappa\in
K_1$.  However, in $K_2$ there may be a class $\kappa_\emptyset$ whose
vertices do not have neighbors in $X$; i.e.,
$N'_X(\kappa_\emptyset)=\emptyset$.  Note that the vertices of this
class, provided it exists, cannot be contained in $Z$.

For a class $\kappa \in K_1$ we define $U_\kappa = N[\kappa]\cap R$.  That
is,~$U_\kappa$ comprises all vertices of $R\subseteq V(G)$ that have
been contracted onto the vertices of~$\kappa$ during the construction
of~$G'$.  Since~$S$ is $2$-scattered in~$G[R]$, the sets~$U_\kappa$
for $\kappa\in K_1$ are pairwise disjoint.  Moreover, $(U_\kappa)_{\kappa\in K_1}$
forms a partition of~$R_1$.

Intuitively, our goal now is to identify a large class $\kappa \in K_1$ that
cannot be dominated by a small set of vertices in~$R$.  We then argue
that such a class contains a vertex that is irrelevant: it can be
removed from~$Z$ without breaking the invariant that~$Z$ is a
domination core.

First, we define an auxiliary graph that captures the interaction
between the classes in $K$.

\begin{definition}
  The \emph{class graph}~$H$ is a graph with vertex set~$K$ that
  contains an edge between $\kappa, \kappa' \in K$ if and only if
  there exists $u \in \kappa$ and $u' \in \kappa'$ such that $uu' \in
  E(G')$.
\end{definition}

\begin{figure}[htb]
  \centering
  \includegraphics[width=.5\textwidth]{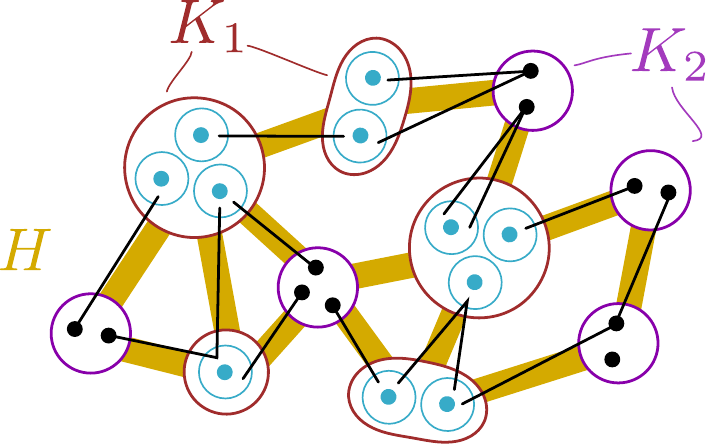}
  \caption{The class graph $H$ with vertex set $K_1 \cup K_2 = K$.}
  \label{fig:classgraph}
\end{figure}

The crucial observation, which we are going to prove next, is that the class graph actually cannot be too
large and complicated: it has an almost linear number of vertices and edges, measured in $|X|$. In the following, we denote $\gamma=\delta/2$.

\begin{lemma}[Size of the class graph]\label{lem:classgraph-V-nd} 
The following holds:
\begin{itemize}
\item $|K_1|\leq \neifun^1(\gamma)\cdot |X|\cdot n^{\gamma}$, and
\item $|K_2|\leq \neifun(\gamma)\cdot |X|\cdot n^{\gamma}$.
\end{itemize}
Consequently, $|V(H)|=|K_1|+|K_2| \leq 2\neifun^1(\gamma)\cdot |X|\cdot n^{\gamma}$.
\end{lemma}
\begin{proof}
The upper bound on $|K_2|$ (second item) follows directly from Lemma~\ref{lem:twin-classes2} applied to the graph $G-R_1$, set $X$, and parameter $\gamma$ (we use that $|X|\leq n$). In order to obtain an upper bound on $|K_1|$, we apply Lemma~\ref{lem:twin-classes2} to the graph $G'[S\cup X]\in \mc G\nab 1$. We thus infer that the number of possible $X$-neighborhoods among the vertices in $S$, and hence the number of classes in $K_1$, is at most $\neifun^1(\gamma)\cdot |X|\cdot n^{\gamma}$.
\end{proof}

\begin{lemma}[Grad of the class graph]\label{lem:classgraph-grad-nd}
There exists a function $h(\cdot)$ such that for every $r\geq 0$ it holds that $\grad_r(H) \leq h(r)\cdot n^{\gamma}$.
\end{lemma}
\begin{proof}
Let us fix $r$ and let $\beta = \frac{\gamma}{3(3r+2)^2}$. In the following we assume that $\kappa_\emptyset$ exists; otherwise the argument is even simpler as we do not need to consider this class separately.

We construct sets $L_1$ and $L_2$ by picking an arbitrary vertex from each class of $K_1$ and each class of $K_2\setminus \{\kappa_\emptyset\}$, respectively. Consider bipartite graph $G_1'=(X,L_1,E(G') \cap (X \times L_1))$ and $G_2'=(X,L_2,E(G') \cap (X \times L_2))$, i.e., the bipartite graph induced in $G'$ between $X$ and $L_1$ or $L_2$, respectively. By the definitions of classes of $K_1$ and $K_2$, and the fact that $X$ is a $Z$-dominator, we infer that both these graphs satisfy the assumptions of Lemma~\ref{lem:orient2}.

Hence, Lemma~\ref{lem:orient2} ensures us that there exist assignments $\phi_1\colon L_1\to X$ and $\phi_2\colon L_2\to X$ such that $|\phi_t^{-1}(u)|\leq \chrgfun^1(\beta)\cdot n^\beta$ for each $u\in X$ and $t=1,2$, and moreover $v\phi_t(v)\in E(G')$ for each $v\in L_t$. Let us combine these assignments into $\phi\colon L_1\cup L_2\to X$ such that $|\phi^{-1}(u)|\leq \tau$ for each $u\in X$, where $\tau= 2\chrgfun^1(\beta)\cdot n^\beta$. By somehow abusing the notation, we regard $\phi$ also as an assignment with domain $K_1\cup K_2\setminus \{\kappa_\emptyset\}$ in a natural way.

We now consider the lexicographic product $G'' = G' \lexprod
  K_{\tau}$.  Let us construct a $1$-shallow minor $H'\in G''\nab 1$
  as follows: for every class $\kappa \in K \setminus
  \{\kappa_\emptyset\}$, contract all the copies of all the vertices
  of~$\kappa$ onto one of the copies of $\phi(\kappa) \in X$, so that
  every class $\kappa \in K \setminus \{\kappa_\emptyset\}$ is
  contracted onto a different vertex.  Since every vertex of~$X$ is
  chosen at most~$\tau$ times by~$\phi$, such a contraction is possible.
  Let $\overline{\phi}\colon K \setminus \{\kappa_\emptyset\}\to
  V(G'')$ be an injection that assigns classes of $K \setminus
  \{\kappa_\emptyset\}$ to the copies of vertices of~$X$ they are
  contracted onto.  Then it is easy to see that~$\overline{\phi}$
  defines a subgraph embedding of $H-\{\kappa_\emptyset\}$ into~$H'$.
  Consequently, $H-\{\kappa_\emptyset\}$ is a $1$-shallow minor
  of~$G''$. Hence, we can upper bound the grads of $H-\{\kappa_\emptyset\}$ using Proposition~\ref{prop:comp-of-grads} and Lemma~\ref{lem:lexprod}:
\begin{align*}
\grad_r(H-\{\kappa_\emptyset\}) &\leq \grad_{3r+1}(G'') \leq 4(8\tau(3r+1+\tau)\grad_{3r+1}(G')+4\tau)^{(3r+2)^2}\\
           &\leq 4(16\chrgfun^1(\beta)(3r+1+2\chrgfun^1(\beta))n^{2\beta}\cdot \grad_{9r+4}(G)+8\chrgfun^1(\beta)\cdot n^\beta)^{(3r+2)^2}\\
           &\leq 4(24\chrgfun^1(\beta)(3r+1+2\chrgfun^1(\beta))n^{2\beta}\cdot \gradfun(\beta,9r+4)\cdot n^{\beta})^{(3r+2)^2}\\
           &= 4(24 \chrgfun^1(\beta)(3r+1+2\chrgfun^1(\beta))\cdot \gradfun(\beta,9r+4))^{(3r+2)^2}\cdot n^{\gamma}.
\end{align*}
Graph $H$ can be obtained from $H-\{\kappa_\emptyset\}$ by adding a universal vertex and then possibly removing some edges. Hence, by Lemma~\ref{lem:grad-universal-vertex} we infer that 
$$\grad_r(H)\leq \grad_r(H-\{\kappa_\emptyset\})+1\leq 5(24 \chrgfun^1(\beta)(3r+1+2\chrgfun^1(\beta))\cdot \gradfun(\beta,9r+4))^{(3r+2)^2}\cdot n^{\gamma}.$$
This concludes the proof.
\end{proof}

\newcommand{\CE}{C_E}

\begin{corollary}\label{cor:classgraph-nd} 
There exists a constant $\CE$ such that $|E(H)|\leq \CE\cdot |X|\cdot n^{\delta}$.
\end{corollary}
\begin{proof}
  Since $|E(H)| \leq \grad_0(H) \cdot |V(H)|$, we apply the upper bounds proven
  in Lemma~\ref{lem:classgraph-V-nd} and in Lemma~\ref{lem:classgraph-grad-nd} and obtain
  \begin{align*}
      |E(H)| &\leq h(0)\cdot n^{\gamma}\cdot 2\neifun^1(\gamma)\cdot |X|\cdot n^{\gamma} \\
             &= 2h(0)\neifun^1(\gamma)\cdot |X|\cdot n^{\delta}.
  \end{align*}
  Hence, we can take $\CE:=2h(0)\neifun^1(\gamma)$.
\end{proof}

Now is the moment when we can finally set the constant $C_0$ that governs how much larger is $S$ compared to $X$; more precisely, we assumed that $|S|>C_0\cdot |X|\cdot n^{\delta}$. We namely set 
$$C_0=2^{\excbicone}\cdot \left( (\excbicone+1)\cdot \neifun^1(\gamma)+2\excbicone\CE\right).$$
The following lemma is the crux of our approach in this section. Using the bound on the sparsity of $H$, we identify a subclass whose size is large compared to its possible interaction in $H$.

\begin{lemma}[Large subclass]\label{lem:large-subclass-nd}
  There exists a class $\kappa \in K_1$ and a subset $\lambda \subseteq \kappa$
  with the properties that every member~$s \in \lambda$ has the
  same neighborhood $N_X(s)$ in $G$ and
  $$| \lambda | > \excbicone \cdot(\deg_H(\kappa) + 1)  + 1.$$
\end{lemma}
\begin{proof}
Let us define a potential function for classes $\kappa\in K_1$ as follows:
  $$
    \Phi(\kappa) = |\kappa| - 2^{\excbicone} (\excbicone \cdot(\deg_H(\kappa) + 1)  + 1).
  $$
Summing up this potential through all the classes of $K_1$ we obtain the following:
  \begin{align*}
    \sum_{\kappa \in K_1} \Phi(\kappa) &=
      \sum_{\kappa \in K_1} |\kappa| - 
      2^{\excbicone}\cdot \sum_{\kappa \in K_1} (\excbicone \cdot(\deg_H(\kappa) + 1)  + 1) \\
      &= |S| - 
      2^{\excbicone}\cdot \left( \excbicone \sum_{\kappa \in K_1}\deg_H(\kappa)+ (\excbicone+1)|K_1| \right).
  \end{align*}
We now use the fact that $\sum_{\kappa \in K_1}\deg_H(\kappa)\leq \sum_{\kappa \in V(H)}\deg_H(\kappa)=2|E(H)|$ and the bounds of Lemma~\ref{lem:classgraph-V-nd} and Corollary~\ref{cor:classgraph-nd}:
  \begin{align*}
      \sum_{\kappa \in K_1} \Phi(\kappa) &\geq |S| - 2^{\excbicone} \left( 2\excbicone\cdot \CE \cdot |X|\cdot n^\delta + (\excbicone+1)\cdot \neifun^1(\gamma)\cdot |X|\cdot n^\gamma \right)\\
                                         &\geq |S| - 2^{\excbicone} \left( 2\excbicone\cdot \CE \cdot |X|\cdot n^\delta + (\excbicone+1)\cdot \neifun^1(\gamma)\cdot |X|\cdot n^\delta \right)\\
                                         &= |S| - C_0\cdot |X|\cdot n^\delta>0.      
  \end{align*}
Hence, we infer that there exists a class $\kappa\in K_1$ such that $\Phi(\kappa)>0$. Equivalently,
$$|\kappa| > 2^{\excbicone} (\excbicone \cdot(\deg_H(\kappa) + 1)  + 1).$$
Let us partition vertices of $\kappa$ into subclasses with respect to their neighborhoods in $X$ in graph $G$ (recall that they have the same neighborhoods in $X$ in graph $G'$ by the definition of $\simeq_X$, but the neighborhoods in $X$ in graph $G$ may differ). Recall that we have that $|N'_X(\kappa')|<\excbicone$ for each $\kappa'\in K_1$, and $N_X(s)\subseteq N'_X(\kappa)$ for each $s\in \kappa$, so the number of these subclasses is actually less than $2^{\excbicone}$. Hence, there exists a subclass $\lambda\subseteq \kappa$ of vertices with the same $X$-neighborhood in $G$ such that $|\lambda|\geq |\kappa|/2^{\excbicone}>\excbicone \cdot(\deg_H(\kappa) + 1)  + 1$.
\end{proof}

We now prove the bottom line: every vertex of $\lambda$ is an irrelevant dominatee.

\begin{lemma}\label{lem:replace-nd}
  Let $z$ be an arbitrary vertex of $\lambda$.  Then
  $Z\setminus \{z\}$ is still a domination core.
\end{lemma}
\begin{proof}
Let $Z'=Z\setminus \{z\}$ and let $N_X(\lambda)\subseteq N'_X(\kappa)$ be equal to $N_X(s)$ for any
$s\in \lambda$. Take any minimum-size $Z'$-dominator $D$; we need to prove that $D$ is a dominating set of $G$. In the following we work in the graph $G$ all the time.

Suppose first that $D\cap N_X(\lambda)\neq \emptyset$. Then in particular $z$ is also dominated by $D$, hence $D$ is also a $Z$-dominator. As $Z\supseteq Z'$, $D$ must be a minimum-size $Z$-dominator, and hence also a dominating set in $G$ since $Z$ was a domination core.

Suppose then that $D\cap N_X(\lambda) = \emptyset$. We are going to arrive at a contradiction with the assumption that $D$ is of minimum possible size. Since $\lambda\setminus\{z\}\subseteq S\setminus\{z\}\subseteq Z'$, vertices of $\lambda\setminus\{z\}$ need in particular to be dominated by $D$. Since $S$ is $2$-scattered in $G-X$, so is $\lambda\setminus\{z\}$ as well. Hence any vertex of $D$ can dominate only at most one vertex of $\lambda\setminus\{z\}$, as none of them can be dominated from $X$ by the assumption that $D\cap N_X(\lambda) = \emptyset$. Also, the vertices of $D$ that dominate vertices of $\lambda\setminus\{z\}$ need to be contained in $U_\kappa$; recall that $U_\kappa:=\bigcup_{s\in \kappa} N[s]\cap R$ is the set of all vertices contracted onto vertices of $\kappa$ during the construction of $G'$. Hence, we conclude that $|D\cap U_\kappa|\geq |\lambda\setminus\{z\}|> \excbicone \cdot(\deg_H(\kappa) + 1)$.

Construct now a set $D'$ from $D$ by the following steps:
\begin{enumerate}[(a)]
\item\label{st:rem} remove all the vertices of $D\cap U_\kappa$,
\item\label{st:add1} add all the vertices of $N'_X(\kappa_1)$ for every $\kappa_1\in N_H[\kappa]\cap K_1$, and 
\item\label{st:add2} add an arbitrary vertex of $N_X(\kappa_2)$ for each $\kappa_2\in N_H[\kappa]\cap K_2$, provided that $N_X(\kappa_2)$ is non-empty. 
\end{enumerate}
In step~\eqref{st:rem} we have removed more than $\excbicone \cdot(\deg_H(\kappa) + 1)$ vertices from $D$, whereas in steps~\eqref{st:add1} and~\eqref{st:add2} we have added in total at most $\excbicone\cdot (\deg_H(\kappa) + 1)$ vertices: at most $\excbicone$ vertices per each $\kappa_1\in N_H[\kappa]\cap K_1$, and at most one vertex per each $\kappa_2\in N_H[\kappa]\cap K_2$. Hence $|D'|<|D|$, and to arrive at a contradiction it remains to prove that $D'$ is a $Z'$-dominator. %From here on the proof is exactly the same as in Lemma~\ref{lem:replace}, but we recall it briefly for completeness.

Take any $u\in Z'$ which became not dominated when $D\cap U_\kappa$ was removed during the construction of $D'$; we prove that $u$ is dominated by the vertices added to $D'$ in steps~\eqref{st:add1} and~\eqref{st:add2}. Since $u$ was dominated by a vertex from $D\cap U_\kappa$, we have four cases: $u$ can belong (a) to $N'_X(\kappa)$, or (b) to $U_\kappa$, or (c) to $U_{\kappa_1}$ for some $\kappa_1 \in N_H(\kappa) \cap K_1$, or (d) to some $\kappa_2 \in N_H(\kappa) \cap K_2$. Moreover, since $u\in Z'$ and $X$ is a $Z$-dominator, we infer that $u$ has at least one neighbor in $X$. In case (a) we have explicitly included $N'_X(\kappa)$ to $D'$, so even $u\in D'$. In cases (b) and (c) we have added the sets $N'_X(\kappa_1)$ to $D'$ for each $\kappa_1\in N_H[\kappa]\cap K_1$, so any neighbor of $u$ in $X$ belongs to $D'$ and thus dominates $u$. In case (d), we have that $N_X(u)=N_X(\kappa_2)$ and this set is non-empty, since $u$ indeed has a neighbor in $X$. Hence, we added one vertex of $N_X(u)$ to set $D'$ and this vertex thus dominates $u$.
\end{proof}

We now conclude the proof of Lemma~\ref{lem:reduce-corce-nd}, which also concludes the proof of Theorem~\ref{thm:dcore-nd}. Adopting the notation of Section~\ref{sec:extraction-nd}, we take $g(\eps)=(\excbic\cdot \neifun(\eps/2)+1)\cdot C_{\excbic}$ (note that $C_{\excbic}$ also depends on $\eps$), so that Lemma~\ref{lem:extraction-terminates-nd} is applicable whenever $|Z|>g(\eps)\cdot k\cdot n^{\eps}$. Hence, we can safely apply the algorithm of Section~\ref{sec:extraction-nd}, which clearly works in polynomial time as it boils down to a constant number of applications of the algorithm of Lemma~\ref{lem:better-dvorak-nd}, and obtain a pair $(X,S)$ that can be used in the second phase. Construction of the class graph $H$ can be clearly done in polynomial time. Also, in polynomial time we can recognize the class $\kappa$ and subclass $\lambda\subseteq \kappa$ that satisfy the statement of Lemma~\ref{lem:large-subclass-nd}: this requires iterating through all the classes $\kappa\in K_1$, and then examining the partition of the vertices of the found class $\kappa$ with respect to the neighborhoods in $X$. Finally, Lemma~\ref{lem:replace-nd} ensures that any vertex of $\lambda$ can be output by the algorithm as an irrelevant dominatee.

\subsection{Reducing dominators}\label{sec:dominators-nd}

Having presented how to compute a small dominating core in the nowhere dense case, we can proceed to the proof of Theorem~\ref{thm:main-nd}. Before this, we prove one more lemma from which the main result for the nowhere dense case will follow very easily. Its proof is essentially the same as the proof of Theorem~\ref{thm:main-be}.

\begin{lemma}\label{lem:nd-main-neps}
  Let $\mathcal{G}$ be a nowhere dense graph class and let $\eps > 0$
  be a real number. There exists a constant $C_\eps$ and a polynomial-time algorithm that,
  given an $n$-vertex graph $G \in \mathcal{G}$ and an integer $k$, either
  correctly concludes that $\ds(G) > k$ or finds a subset of vertices
  $Y \subseteq V(G)$ of size at most $C_\eps\cdot k\cdot n^{\eps}$ with the property that
  $\ds(G) \leq k$ if and only if $\ds(G[Y]) \leq k$.
\end{lemma}
\begin{proof}
The algorithm works as follows. First, using the algorithm of Theorem~\ref{thm:dcore-nd} for parameter $\eps/2$ we compute a domination core $Z\subseteq V(G)$ such that $|Z|\leq g(\eps/2)\cdot k\cdot n^{\eps/2}$. If the algorithm of Theorem~\ref{thm:dcore-nd} concluded that $\ds(G)>k$, then we can also terminate and provide this outcome. Hence, from now on we assume that the domination core $Z$ has been successfully computed.

Let $R:=V(G)\setminus Z$ and partition the vertices of $R$ into classes with respect to their neighborhoods in $Z$. From Lemma~\ref{lem:twin-classes2} we infer that the number of these classes is at most $\neifun(\eps/2)\cdot |Z|\cdot n^{\eps/2}\leq \neifun(\eps/2)g(\eps/2)\cdot k\cdot n^{\eps}$. Construct set $Y$ by taking $Z$ and, for every nonempty class $\kappa$ of the considered partition, adding an arbitrarily picked vertex $v_\kappa\in \kappa$. Note that in this manner we have that:
$$|Y|\leq |Z|+\neifun(\eps/2)g(\eps/2)\cdot k\cdot n^{\eps}\leq (\neifun(\eps/2)+1)g(\eps/2)\cdot k\cdot n^{\eps},$$
which means that we can set $C_\eps:=(\neifun(\eps/2)+1)g(\eps/2)$. We are left with verifying that $\ds(G)\leq k$ if and only if $\ds(G[Y]) \leq k$.

Suppose first that $\ds(G)\leq k$, and let $D$ be a minimum-size dominating set in $G$ so that $|D|=\ds(G)\leq k$. As $D$ is a dominating set in $G$, it is in particular a $Z$-dominator, and it is a minimum-size $Z$-dominator in $G$ since $Z$ is a domination core and $\ds(G)=\ds(G,Z)$. Construct set $D'$ from $D$ by replacing the set $\kappa\cap D$ with $\{v_\kappa\}$ for each class $\kappa$ of the partition with $|\kappa\cap D|\geq 1$. Clearly, $|D'|\leq |D|$. Moreover, observe that set $D'$ is also a $Z$-dominator in $G$, since every vertex $v_\kappa$ dominates exactly the same set of vertices in $Z$ as other vertices of $\kappa$. As $D$ was a minimum-size $Z$-dominator, we infer that in fact $|D'|=|D|=\ds(G,Z)$ and $D'$ is also a minimum-size $Z$-dominator. Since $Z$ is a domination core, we infer that $D'$ is a dominating set in $G$. Finally, as $D'\subseteq Y$, we infer that $D'$ is also a dominating set in $G[Y]$ and hence $\ds(G[Y])\leq |D'|\leq k$.

Suppose now that $\ds(G[Y])\leq k$, and let $D'$ be a dominating set in $G[Y]$ such that $|D'|\leq k$. Set $D'$ in particular dominates the whole set $Z\subseteq Y$, which means that $D'$ is also a $Z$-dominator in $G$. Hence $\ds(G,Z)\leq |D'|\leq k$. As $Z$ is a domination core, we have that $\ds(G)=\ds(G,Z)$ and we conclude that $\ds(G)\leq k$.
\end{proof}

\restatemainnd*
\begin{proof}
We apply the algorithm of Lemma~\ref{lem:nd-main-neps} iteratively to obtain sets $V(G)=Y_0\supseteq Y_1\supseteq Y_2\supseteq Y_3\supseteq\ldots$: In the $i$-th iteration we apply the algorithm to $G[Y_{i-1}]$ in order to compute $Y_i\subseteq Y_{i-1}$. We proceed in this manner up to the point when the algorithm returns $Y_i=Y_{i-1}$, in which case we simply output $Y:=Y_i$. Clearly, $Y$ computed in this manner satisfies the requirement that $\ds(G)\leq k$ if and only if $\ds(G[Y]) \leq k$, so it remains to establish the upper bound on the size of $Y$.

Since the algorithm of Lemma~\ref{lem:nd-main-neps} returned $Y_i=Y_{i-1}$, it follows that
$$|Y|=|Y_i|\leq C_\eps\cdot k\cdot |Y_{i-1}|^{\eps}=C_\eps\cdot k\cdot |Y|^{\eps}.$$
Here, $C_\eps$ is the constant from the statement of Lemma~\ref{lem:nd-main-neps}. Consequently,
\begin{align*}
|Y|& \leq (C_\eps\cdot k)^{\frac{1}{1-\eps}}\leq C_\eps^{1+2\eps}\cdot k^{1+2\eps}.
\end{align*}
By rescaling $\eps$ by factor $2$ we obtain the result.
\end{proof}

%%% Local Variables:
%%% TeX-command-default: "Make"
%%% mode: latex
%%% TeX-master: "main"
%%% End: 

%  #####  ######   #####     #     #                                                  
% #     # #     # #     #    #     #   ##   #####  #####  #    # ######  ####   ####  
% #       #     # #          #     #  #  #  #    # #    # ##   # #      #      #      
% #       #     #  #####     ####### #    # #    # #    # # #  # #####   ####   ####  
% #       #     #       #    #     # ###### #####  #    # #  # # #           #      # 
% #     # #     # #     #    #     # #    # #   #  #    # #   ## #      #    # #    # 
%  #####  ######   #####     #     # #    # #    # #####  #    # ######  ####   #### 

\section{Hardness of {\sc{Connected Dominating Set}}}\label{sec:condom}

In this section we prove Theorem~\ref{thm:condom-lb}; let us recall its statement.

\restatecondomlb*

The proof of Theorem~\ref{thm:condom-lb} is a refinement of the 
proof of Cygan et al.~\cite{bimber} that \probCDS{} does not admit
a polynomial kernel in graphs of bounded degeneracy.
The main idea of~\cite{bimber} is to use \probMotif{} as a pivot problem.

\defparproblem{\probMotif{}}{A graph $G$, an integer $k$, and a surjective function $c:V(G) \to [k]$.}{$k$}{Does there exist a set $X \subseteq V(G)$ of size exactly $k$ such that $G[X]$ is connected and $c|_X$ is bijective?}

We call the function $c$ a \emph{coloring} and each value $i \in [k]$ is a \emph{color}.
In this wording, in the \probMotif{} problem we seek for a set of vertices, one of every color, that induces a connected subgraph 
of $G$.

Fellows et al.~\cite{mike:motif} were first to study the parameterized complexity of \probMotif{} and, among other results,
they prove that the problem is hard already in a very restrictive setting.

\begin{theorem}[\cite{mike:motif,bimber}]\label{thm:mike-motif}
The \probMotif{} problem, restricted to graphs $G$ being trees of maximum degree $3$,
is NP-complete and does not admit a polynomial compression when parameterized by $k$ unless \compassunless.
\end{theorem}

Here, a \emph{polynomial compression} is a generalization of the notion of a polynomial kernel,
where we relax the requirement that the output needs to be an instance of a original problem.
Formally, a polynomial compression from a parameterized language $P$ into a (classic) language $L$ is an algorithm that, given
an instance $(x,k)$, works in time polynomial in $|x|+k$ and outputs a string $y$ with the following properties: 
(i) $(x,k) \in P$ if and only if $y \in L$, and (ii) $|y|$ is bounded polynomially in $k$.

The main observation of~\cite{bimber} is that \probMotif{} easily reduces to \probCDS{}.
Let $I=(G,k,c)$ be a \probMotif{} instance. Consider a graph
$\CDSG{I}$ constructed as follows: we first take $\CDSG{I} = G$ and then, for every color $i \in [k]$, we add two vertices
$w_i$ and $w_i^\circ$, connected by an edge, and make $w_i$ adjacent to $c^{-1}(i)$, that is, to all vertices of $G$ of color $i$.
It is easy to observe the following.
\begin{lemma}[\cite{bimber}]
$I$ is a yes-instance to \probMotif{} if and only if $\CDSG{I}$ admits a connected dominating set of size at most $2k$.
\end{lemma}
\begin{proof}
Let $W = \{w_i : 1 \leq i \leq k\}$. Observe that $W$ is a dominating set in $\CDSG{I}$.
If $k=1$, then $W$ is also connected and the claim is trivial, so assume $k\geq 2$.

In one direction, observe that if $X$ is a solution to \probMotif{} instance, then $X \cup W$ is a connected dominating set
in $G$ of size $2k$: $W$ dominates $V(G)$, while $G[X]$ is connected and every $w_i \in W$ has a (unique) neighbor 
in $X \cap c^{-1}(i)$.

In the other direction, let $Y$ be a connected dominating set of size at most $2k$ in $\CDSG{I}$.
Observe that, due to pendant vertices $w_i^\circ$, the set $Y$ needs to contain $W$.
Since $k \geq 2$, to make $W \subseteq Y$ connected, for every $1 \leq i \leq k$ the set $Y$ needs to contain a vertex $y_i \in c^{-1}(i)$.
Since $|Y| \leq 2k$, we have already enumerated all vertices of $Y$: $Y = \{w_i: 1 \leq i \leq k\} \cup \{y_i : 1 \leq i \leq k\}$.
Thus, every $w_i$ is of degree one in $\CDSG{I}[Y]$ and, consequently, $\CDSG{I}[\{y_i: 1 \leq i \leq k\}]$ is connected.
Hence, $\{y_i: 1 \leq i \leq k\}$ is a solution to \probMotif{} on $I$.
\end{proof}

In is easy to see that the reduction from $I=(G,k,c)$ to $\CDSG{I}$ described above translates not only NP-hardness,
but also kernelization lower bound: any polynomial compression
for \probCDS{}, pipelined with the aforementioned reduction, would give a polynomial compression for \probMotif{}.

As observed in~\cite{bimber}, if $G$ is a tree, then $\CDSG{I}$ is $2$-degenerate. However, $\CDSG{I}$ may not be of bounded expansion,
due to arbitrary connections in the graph introduced by the edges incident to vertices $w_i$.
Our main goal for the rest of this section is to tweak the reduction described above to make $\CDSG{I}$ of bounded expansion.

To control the expansion of $\CDSG{I}$ --- and prove Theorem~\ref{thm:condom-lb} --- we need to control how 
the colors of $I$ can neighbor each other. More formally, given an instance $I=(G,k,c)$ of \probMotif{},
let us define the \emph{color graph} $\ColG{I}$ to be a graph with vertex set $V(\ColG{I})=[k]$ and $ij \in E(\ColG{I})$ if and only if
there exists an edge $xy \in E(G)$ with $c(x) = i$ and $c(y) = j$.
The next lemma shows that if we can control the maximum degree of $\ColG{I}$, then $\CDSG{I}$ is of bounded expansion.

\begin{lemma}\label{lem:maxdeg-to-exp}
Let $(G,k,c)$ be a \probMotif{} instance.
Assume that the maximum degree of $G$ is at most $\Delta_G$, and the maximum degree of $\ColG{I}$ is at most $\Delta_H$.
Then, for every $r \geq 1$, every $r$-shallow topological minor of $\CDSG{I}$ is
$\max(\Delta_G+1, (\Delta_H+1)^{2r})$-degenerate. 
\end{lemma}
\begin{proof}
Fix $r \geq 1$. Let $H$ be an $r$-shallow topological minor of $\CDSG{I}$.
To prove the lemma, it suffices to show that $H$ contains a vertex of degree at most $\max(\Delta_G+1, (\Delta_H+1)^{2r}+1)$;
the same reasoning can be performed for every induced subgraph of $H$.
Let us fix one model of $H$ in $\CDSG{I}$, and consider one vertex $x \in V(H)$ mapped to a root vertex $v \in V(\CDSG{I})$.

If $v \in V(G)$, then the degree of $v$
in $\CDSG{I}$ is at most $\Delta_G+1$, and the same bound holds for the degree of $x$ in $H$.
If $v = w_i^\circ$ for some $1 \leq i \leq k$, then the degree of $v$ in $\CDSG{I}$ is $1$, and the degree of $x$ in $H$ is at most $1$.
Thus, it remains to consider the case where every vertex $x \in V(H)$ is mapped to some vertex $w_i$, $1 \leq i \leq k$.

Consider then a vertex $w_i$. For an integer $d \geq 1$,
we say that a color $j$ is \emph{reachable within distance $d$ from $w_i$} if there exists
a vertex $v \in V(\CDSG{I})$ within distance $d$ from $w_i$ such that $c(v) = j$.
Let $L_d$ be the set of colors reachable from $w_i$ within distance $d$. Observe that
the bound on the degree of $\ColG{I}$ implies the following:
\begin{claim}\label{cl:maxdeg-to-exp}
For every $d \geq 1$ it holds that $|L_d| \leq (\Delta_H+1)^{d-1}$.
\end{claim}
\begin{proof}
We prove by induction on $d$. For $d=1$, observe that $L_1 = \{i\}$.

Consider now $j \in L_{d+1} \setminus L_d$. Since $j \notin L_d$ and every vertex $w_\iota$ has only neighbors
in $c^{-1}(\iota)$ (apart from the pendant $w_\iota^\circ$), there exists
a color $j' \in L_d$ and an edge $xy \in E(G)$ such that $c(x) = j$ and $c(y) = j'$.
Consequently, $jj' \in E(\ColG{I})$. Since the maximum degree of $\ColG{I}$ is bounded by $\Delta_H$,
  we have $|L_{d+1} \setminus L_d| \leq \Delta_H |L_d|$ and the claim follows.
\cqed\end{proof}
By Claim~\ref{cl:maxdeg-to-exp}, for a fixed vertex $w_i$, at most $(\Delta_H+1)^{2r}$ other vertices
$w_j$ are within distance at most $2r+1$ in $\CDSG{I}$ from $w_i$. Consequently,
no $r$-shallow topological minor with roots in vertices $w_i$ can have a vertex of degree more than
$(\Delta_H+1)^{2r}$. This concludes the proof of the lemma.
\end{proof}

By Lemma~\ref{lem:maxdeg-to-exp}, to prove Theorem~\ref{thm:condom-lb} it suffices to show
that the lower bounds of Theorem~\ref{thm:mike-motif}
still hold if we restrict the maximum degree of $\ColG{I}$. Luckily,
this turns out to be quite an easy task (see also Figure~\ref{fig:condom-lb-red} for an illustration of the gadget used).

\begin{figure}[htb]
  \centering
  \includegraphics[width=.95\textwidth]{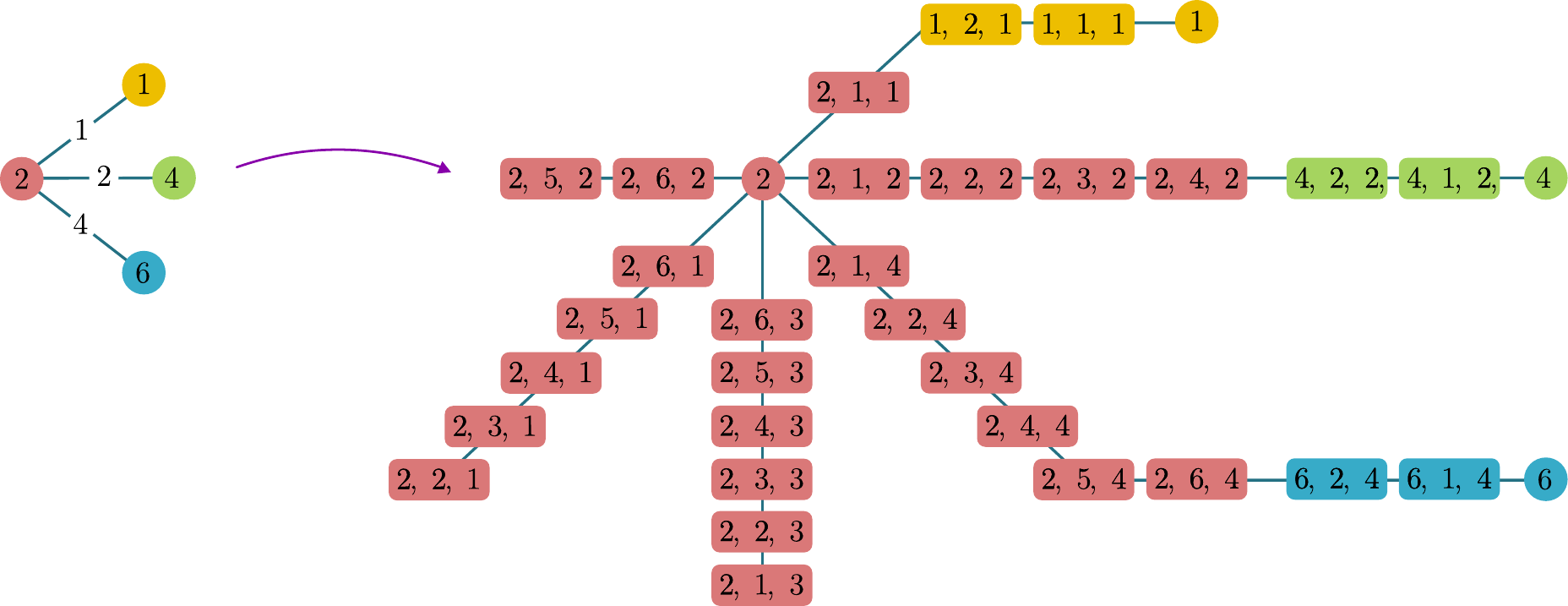}
  \caption{\label{fig:condom-lb-red} Part of the graph corresponding to a vertex $u$ of color $c(u)=2$, with neighbors of colors $1$, $4$ and $6$, and assuming $k=6$ and $\Delta_G=3$.
  The numbers on edges correspond to their colors in the coloring $f$.}
\end{figure}

\begin{lemma}\label{lem:maxdeg-colg}
There exists a polynomial algorithm that, given a \probMotif{} instance $(G,k,c)$ where
the maximum degree of $G$ is bounded by $\Delta_G$, outputs an equivalent
\probMotif{} instance $I'=(G', k', c')$ where
$k' = k + (\Delta_G+1) k^2$, the maximum degree of $G'$ is bounded by $2\Delta_G+2$,
and the maximum degree of $\ColG{I'}$ is bounded by $\max(2\Delta_G+2,3)$.
\end{lemma}
\begin{proof}
For clarity of presentation, we identify the new set of colors, $[k']$, with $[k] \cup ([k] \times [k] \times [\Delta_G+1])$.
By Vizing's theorem, the edges of $G$ can be colored with $\Delta_G+1$ colors such that
no two incident edges have the same color. Moreover, such a coloring can be found in polynomial
time~\cite{find-vizing}. Let $f:E(G) \to [\Delta_G+1]$ be any such coloring.

For integers $i,a,b \in [k]$ and $\alpha \in [\Delta_G+1]$ with $a \leq b$
we define an $(i,\alpha;a,b)$-path to be a path on $b-a+1$ vertices denoted $x_{i,j,\alpha}$ for $a \leq j \leq b$
and with colors $c'(x_{i,j,\alpha}) = (i,j,\alpha)$.

We construct the instance $I'$ as follows. We start with $V(G') = V(G)$ and $c' = c$.
Then, for every edge $uv$ we make the following construction. Assume $c(u) = i$, $c(v) = j$, and $f(uv) = \alpha$.
We first take an $(i,\alpha;1,j)$-path $P_u^\alpha$ and an $(j,\alpha;1,i)$-path $P_v^\alpha$, and connect
them as follows: we make $x_{i,1,\alpha}$ on $P_u^\alpha$ adjacent to $u$, $x_{j,1,\alpha}$ on $P_v^\alpha$ adjacent
to $v$, and $x_{i,j,\alpha}$ on $P_u^\alpha$ adjacent to $x_{j,i,\alpha}$ on $P_v^\alpha$.
In this way we have added a path $P_u^\alpha \cup P_v^\alpha$ between $u$ and $v$ of length $j+i+1$.
Second, if $j < k$ we take a $(i,\alpha;j+1,k)$-path $Q_u^\alpha$ and make $x_{i,k,\alpha}$ on this path
adjacent to $u$. Similarly, if $i < k$ we take a $(j,\alpha;i+1,k)$-path $Q_v^\alpha$ and make $x_{j,k,\alpha}$
on this path adjacent to $v$.
If $j=k$ or $i=k$, then the corresponding path $Q_u^\alpha$ or $Q_v^\alpha$ is defined to be an empty path
for the sake of further notation.

Furthermore, if for some $u \in V(G)$ and $\alpha \in [\Delta_G+1]$ there does not exist an edge incident to $u$
colored (by $f$) with color $\alpha$, then we create a $(c(u),\alpha;1,k)$-path $Q_u^\alpha$
and make $x_{i,k,\alpha}$ on this path adjacent to $u$.

This concludes the description of the instance $I'=(G',k',c')$.
In the next three claims we prove the desired properties of $I'$.

\begin{claim}
The instances $I$ and $I'$ are equivalent.
\end{claim}
\begin{proof}
For $u \in V(G)$, let $W_u$ be the set of vertices of $G'$ associated with $u$, that is, the vertex $u$
as well as all vertices on all paths $P_u^\alpha$ and $Q_u^\alpha$, $\alpha \in [\Delta_G+1]$.
Observe that, by construction, the set $W_u$ contains exactly one vertex of every color
of $\{c(u)\} \cup (\{c(u)\} \times [k] \times [\Delta_G+1])$, and no vertices of other colors.
Furthermore, $G[W_u]$ is connected.
Consequently, if $X \subseteq V(G)$ is a solution to the \probMotif{} instance $I$,
then $X' := \bigcup_{u \in X} W_u$ is a solution to $I'$: for every edge $uv \in E(G[X])$, the corresponding
path $P_u^{f(uv)} \cup P_v^{f(uv)}$ is completely contained in $G'[X']$.

In the other direction, let $X' \subseteq V(G')$ be a solution to $I'$. We claim that $X := X' \cap V(G)$
is a solution to $I$. If $k=1$, then the claim is trivial, so assume $k \geq 2$.
Clearly, $X$ contains exactly one vertex of every color of $[k]$. Consider the following graph
$G_X$: $V(G_X) = X$ and $uv \in E(G_X)$ if and only if there exists a path in $G'[X']$ between $u$ and $v$
with no internal vertex in $X$. Clearly, the connectivity of $G'[X']$ implies that $G_X$ is connected as well.
Furthermore, observe that every vertex of $V(G') \setminus V(G)$ in $G'$ is of degree at most $2$.
Consequently, for every $uv \in E(G_X)$, the corresponding path in $G'[X']$ has to be equal to 
$P_u^{f(uv)} \cup P_v^{f(uv)}$; in particular, $uv \in E(G)$. We infer that $G_X$ is a subgraph of $G[X]$
and, hence, $G[X]$ is connected. This finishes the proof of the claim.
\cqed\end{proof}

\begin{claim}
The maximum degree of $G'$ is at most $2\Delta_G+2$.
\end{claim}
\begin{proof}
Every vertex of $V(G') \setminus V(G)$ is of degree at most two in $G'$.
Every vertex $v \in V(G)$ is adjacent in $G'$ to at most one vertex of every color of $\{c(v)\} \times \{1,k\} \times [\Delta_G+1]$,
and thus is of degree at most $2(\Delta_G+1)$.
\cqed\end{proof}
\begin{claim}
The maximum degree of $\ColG{I'}$ is at most $\max(3,2\Delta_G+2)$.
\end{claim}
\begin{proof}
As already observed,
every vertex $v \in V(G)$ is adjacent to at most one vertex of every color of $\{c(v)\} \times \{1,k\} \times [\Delta_G+1]$.
Thus, the degree of the color $i \in [k]$ in $\ColG{I'}$ is at most $2(\Delta_G+1)$.
Furthermore, observe that a vertex of color $(i,j,\alpha) \in [k] \times [k] \times [\Delta_G+1]$ can be adjacent only 
to vertices of colors: $(j,i,\alpha)$, $(i,j+1,\alpha)$ if $j<k$, $(i,j-1,\alpha)$ if $j>1$, and $i$ if $j \in \{1,k\}$.
Thus, the degree of the color $(i,j,\alpha)$ in $\ColG{I'}$ is at most $3$.
\cqed\end{proof}
The above three claims conclude the proof of Lemma~\ref{lem:maxdeg-colg}.
\end{proof}

Lemma~\ref{lem:maxdeg-colg} translates the lower bounds of Theorem~\ref{thm:mike-motif}
to the case of bounded degree of $\ColG{I}$, by setting $\Delta_G=3$.
\begin{corollary}\label{cor:maxdeg-colg}
The \probMotif{} problem, restricted to instances $I=(G,k,c)$ where
the maximum degree of $G$ and the maximum degree of $\ColG{I}$ is at most $8$,
is NP-complete and does not admit a polynomial compression when parameterized by $k$ unless \compassunless.
\end{corollary}

Let us conclude with a wrap up of the proof of Theorem~\ref{thm:condom-lb}.
Let $\mathcal{G}$ be the class of graphs where, for every $r \geq 1$,
every $r$-shallow topological minor is $9^{2r}$-degenerate.
Assume we have a polynomial compression algorithm $\mathcal{A}$ for \probCDS{} restricted to $\mathcal{G}$.
Let $I=(G,k,c)$ be a \probMotif{} instance where
the maximum degree of $G$ and the maximum degree of $\ColG{I}$ is at most $8$.
By Lemma~\ref{lem:maxdeg-to-exp}, $\CDSG{I} \in \mathcal{G}$. 
Thus, by applying $\mathcal{A}$ to $\CDSG{I}$ for every such instance $I$, we obtain a polynomial compression 
for \probMotif{} for instances with the maximum degree of $G$ and $\ColG{I}$ bounded by $8$.
Theorem~\ref{thm:condom-lb} follows then from Corollary~\ref{cor:maxdeg-colg}.

\section{Domination on somewhere dense graph classes}\label{sec:w-hardness}

\wsomewheredense*
\begin{proof}
Let $\mc H_p$ be the class of $p$-subdivisions of all the simple graphs, that is, the class comprising all the graphs that can be obtained from any simple graph by replacing every edge by a path of length $p$. We need the following claim, which \Dvorak et al.~\cite{DvorakKT13} attribute to \Nesetril and Ossona de Mendez~\cite{NesetrilM11a}. Unfortunately, in~\cite{NesetrilM11a} we could not find the proof of this exact statement, so for the sake of completeness we prove it ourselves.

\begin{claim}\label{cl:subdivisions}
For every somewhere dense graph class $\mc G$ that is closed under taking subgraphs, there exists an integer $r_0$ such that $\mc H_{r_0}\subseteq \mc G$.
\end{claim}
\begin{proof}
Since $\mc G$ is somewhere dense, by \cite[Theorem 4.1 (iii)]{NesetrilM11a} we have there exists a constant $r_1$ such that $\mc G$ contains every complete graph as a topological minor of depth $r_1$. Since $\mc G$ is closed under taking subgraphs, this means that for every $n\in \N$ there exists a graph $H_n\in \mc G$ that can be obtained from a clique $K_n$ by replacing every edge by a path of length at most $r_2:=2r_1+1$.

For every $n$, let $N(n)$ be the Ramsey number such that a complete graph on $N(n)$ vertices with edges colored with $r_2$ colors always contains a monochromatic complete subgraph on $n$ vertices. Examine the complete graph $K_{N(n)}$ and assign to every edge of $K_{N(n)}$ a color from $\{1,2,\ldots,r_2\}$ depending on the length of the corresponding path in $H_{N(n)}$. By the definition of $N(n)$ we infer that there exists a color $r(n)\in \{1,2,\ldots,r_2\}$ such that there is a monochromatic complete subgraph on $n$ vertices with every edge colored with $r(n)$. This means that $H_{N(n)}$ contains a subgraph that is isomorphic to clique $K_n$ with every edge replaced by a path of length $r(n)$. Thus, $H_{N(n)}$ contains as subgraphs also all the $r(n)$-subdivisions of all the graphs on at most $n$ vertices. We conclude by taking $r_0$ to be any number that appears infinitely many times in the sequence $(r(i))_{i\in \N}$.
\cqed\end{proof}

\newcommand{\Ff}{\mathcal{F}}
\newcommand{\Gg}{\mathcal{G}}

\begin{figure}
  \begin{center}
  \includegraphics[scale=.8]{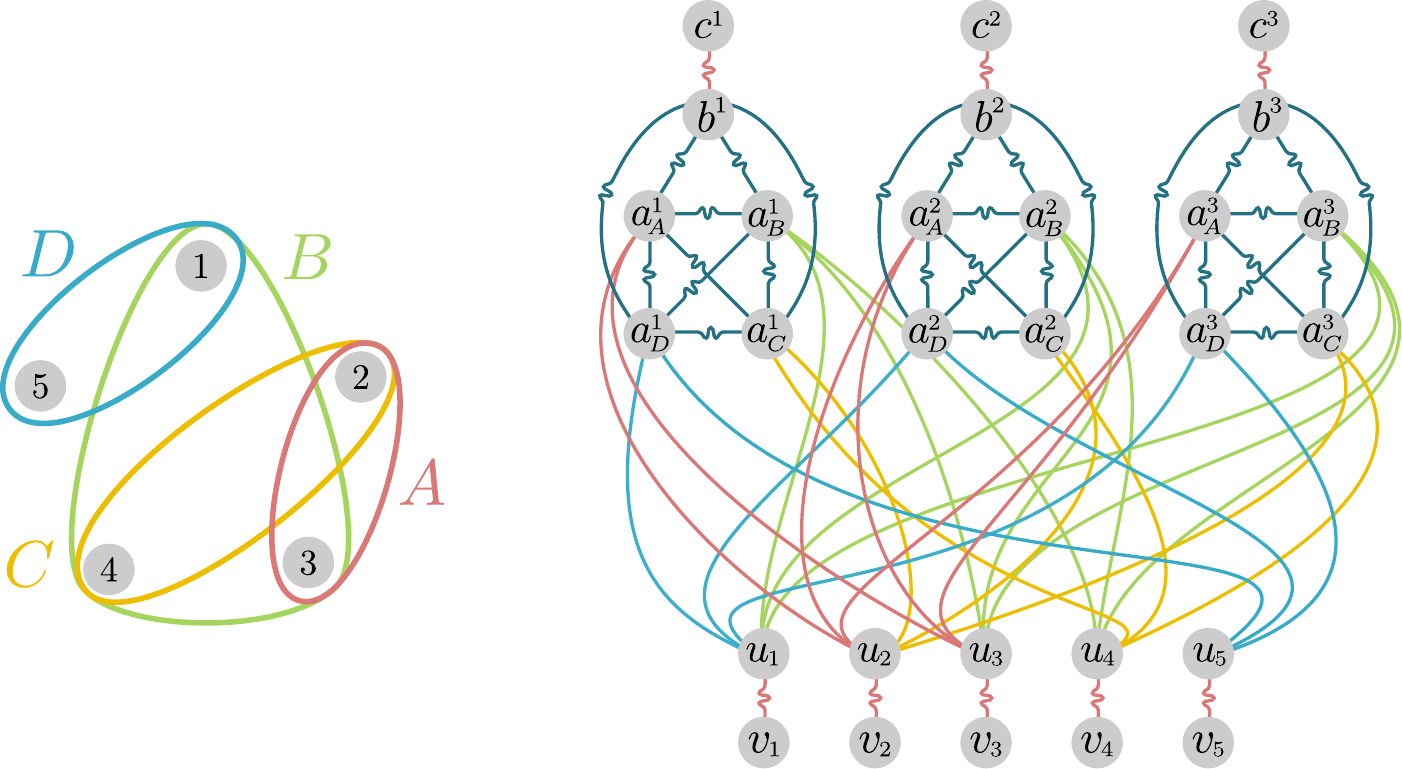}
  \caption{\label{fig:somedense} Example of the reduction for $U = [5], \Ff = \{A,B,C,D\}$ and $k=3$.
  The edges on the right denote paths of length $2r_0$, except those connecting $b_i, c_i$ and
  $u_e,v_e$ whose length is $r_0$.}
  \end{center}
\end{figure}

By Claim~\ref{cl:subdivisions}, the proof of Theorem~\ref{thm:w2-hardness} reduces to proving that for any integer $r_0\geq 0$ there exists an integer $r$ such that {\sc{$r$-Dominating Set}} is $\W{2}$-hard on the class $\mc H_{r_0}$. We prove this fact for $r=3r_0$ by a reduction from the {\sc{Set Cover}} problem parameterized by the requested solution size, which is known to be $\W{2}$-hard~\cite{DowneyF99,FlumGrohebook}. Recall that the instance of the {\sc{Set Cover}} problem consists of $(U,\Ff,k)$, where $U$ is a finite universe, $\Ff\subseteq 2^U$ is a family of subsets of the universe, and $k$ is an integer. The question is whether there exists a subfamily $\Gg\subseteq \Ff$ of size $k$ such that every element of $U$ is covered by $\Gg$, i.e., $\bigcup \Gg=U$.

Given the instance $(U,\Ff,k)$, we construct a graph $G$ as follows; see Figure~\ref{fig:somedense} for an illustration. First, for every $i\in [k]$ do the following:
\begin{itemize}
\item For each $X\in \Ff$, create a vertex $a_X^i$; let $A^i=\{a_X^i\ \colon\  X\in \Ff\}$. For every pair of distinct sets $X,X'\in \Ff$, connect $a_X^i$ and $a_{X'}^i$ with a path of length $2r_0$, thus making the set $A^i$ into a $2r_0$-subdivided clique.
\item Add a vertex $b^i$ and connect it to every vertex of $A^i$ using a path of length $2r_0$.
\item Add a pendant path of length $r_0$ with one endpoint at $b^i$. Let the second endpoint of this path be denoted by $c^i$.
\end{itemize}
Next, for every $e\in U$ do the following:
\begin{itemize}
\item Create a vertex $u_e$ and connect it to every vertex $a_X^i$ such that $i\in [k]$, $X\in \Ff$ and $e\in X$ using a path of length $2r_0$.
\item Add a pendant path of length $r_0$ with one endpoint at $u_e$. Let the second endpoint of this path be denoted by $v_e$.
\end{itemize}
This concludes the construction. It is easy to see that $G\in \mc H_{r_0}$, since $G$ consists of the named vertices connected by paths of length $r_0$ or $2r_0$. It remains to show that instance $(G,k)$ of {\sc{$3r_0$-Dominating Set}} is equivalent to the input instance $(U,\Ff,k)$ of {\sc{Set Cover}}.

\begin{claim}\label{cl:lefttoright}
If instance $(U,\Ff,k)$ of {\sc{Set Cover}} has a solution, then so does instance $(G,k)$ of {\sc{$3r_0$-Dominating Set}}.
\end{claim}
\begin{proof}
Let $\Gg=\{X_1,X_2,\ldots,X_k\}$ be an arbitrary enumeration of a solution $\Gg$ to $(U,\Ff,k)$. Let $D=\{a_{X_i}^i\ \colon\ i\in [k]\}$. We claim that set $D$ $3r_0$-dominates the graph $G$. Observe that, by the construction, every vertex of $G$ is at distance at most $r_0$ from some vertex of $R:=\{b^i\ \colon\ i\in [k]\}\cup \{a_X^i\ \colon\ i\in [k],\, X\in \Ff\}\cup \{u_e\ \colon\ e\in U\}$. Therefore, it suffices to prove that every vertex of $R$ is at distance at most $2r_0$ from a vertex belonging to $D$.

Firstly, every vertex $b^i$ for $i\in [k]$ is at distance $2r_0$ from $a_{X_i}^i$. Secondly, the same holds also for every vertex $a_{X'}^i$ for every $X'\in \Ff$, $X'\neq X_i$. Finally, each vertex $u_e$ is at distance $2r_0$ from vertex $a_{X_i}^i$ for any $X_i$ such that $e\in X_i$; since $U=\bigcup \Gg$, such an index $i$ always exists. By considering all the cases, we conclude that $D$ is indeed a $3r_0$-dominating set in $G$.
\cqed\end{proof}

\begin{claim}\label{cl:righttoleft}
If instance $(G,k)$ of {\sc{$3r_0$-Dominating Set}} has a solution, then so does instance $(U,\Ff,k)$ of {\sc{Set Cover}}.
\end{claim}
\begin{proof}
Let $D$ be a solution to $(G,k)$. For every $i\in [k]$, let $C^i$ be the set of vertices at distance at most $3r_0$ from $c^i$; observe that $C^i$ comprises $c^i$, $b^i$, all the vertices of $A^i$, and all vertices lying on the paths connecting $b^i$ with vertices of $\{c^i\}\cup A^i$. As $c^i$ is $3r_0$-dominated by $D$, every set $C^i$ has a nonempty intersection with $D$. As sets $C^i$ are pairwise disjoint and $|D|\leq k$, we infer that $|D|=k$, $D\subseteq \bigcup_{i\in [k]} C_i$ and every set $C^i$ contains exactly one vertex of $D$. Define $\Gg=\{X_1,X_2,\ldots,X_k\}$ as follows: if $C^i\cap D\subseteq A^i$ then let $X_i$ be such that $C^i\cap D=\{a_{X_i}^i\}$, and otherwise set $X_i$ to be an arbitrary set from $\Ff$. We claim that $\Gg$ constructed in this manner is a solution to $(U,\Ff,k)$.

Take any $e\in U$ and consider the vertex $v_e$. This vertex has to be dominated by $D$, however the only vertices of $\bigcup_{i\in [k]} C_i$ that are at distance at most $3r_0$ from $v_e$ are vertices of the form $a_X^i$ for $i\in [k]$ and $X\in \Ff$ such that $e\in X$. We infer that at least one of these vertices must belong to $D$, so there exists an index $i$ with the following property: set $X_i$ is chosen so that $C^i\cap D=\{a_{X_i}^i\}$ and moreover $e\in X_i$. Since $e$ was chosen arbitrarily, we conclude that $U\subseteq \bigcup \Gg$.
\cqed\end{proof}

Claims~\ref{cl:lefttoright} and~\ref{cl:righttoleft} verify the correctness of the reduction, and thus the proof is concluded.
\end{proof}

\section{Conclusions and Further Research}\label{sec:conclusion}

We have shown that, for each $r\geq 1$, {\sc{$r$-Dominating Set}} admits a linear kernel on any graph class of bounded expansion. Before this work, the most general family of graph classes where such a kernelization result was known were apex-minor-free graphs~\cite{FominLST10}, whereas in the case of the classic {\sc{Dominating Set}}, linear kernels were shown also for general $H$-minor-free~\cite{FominLST12} and $H$-topological-minor-free classes~\cite{FominLST13}. Moreover, for $r=1$, i.e. the {\sc{Dominating Set}} problem, we can also give a kernel on any nowhere dense class of graphs, at the cost of increasing the size bound to almost linear, i.e., $\Oh(k^{1+\eps})$ for any $\eps>0$. These results vastly and broadly extend the current frontier of kernelization results for domination problem on sparse graph classes. On Figure~\ref{fig:kernelmapnce} we depicted the currently explored landscape of parameterized complexity of {\sc{Dominating Set}}, {\sc{$r$-Dominating Set}}, and {\sc{Connected Dominating Set}} on sparse graphs.

We would like to point out several features of our algorithms that at first glance may be not apparent from their description.

Firstly, while we describe our kernelization algorithms for a fixed graph class $\mc G$, it is the case that both the running time and the bound on the size of the kernel do not depend on all the grads of $\mc G$, but only on $r$ and $\grad_p(\mc G)$ for some constant $p$ depending on $r$ (and, of course, on $\grad_{p'}(\mc G)$ for $p'<p$). All the arguments use only $\grad_{3r-1}(\mc G)$ apart from Lemma~\ref{lem:fewneighbourhoods}, where the constant $c$ can be traced to depend on $\grad_{p}(G)$ for some $p=2^{\Oh(r)}$; cf. Theorem~\ref{thm:centered-colorings} and the proof of Lemma~\ref{lem:layered-nei}. In the case of $r=1$, we can replace the usage of Lemma~\ref{lem:fewneighbourhoods}, by the second part of Lemma~\ref{lem:twin-classes}, where the dependence is only on $\grad_1(\mc G)$. Hence, the kernel for {\sc{Dominating Set}} on graph classes of bounded expansion uses only finiteness of $\grad_0(\mc G)$, $\grad_1(\mc G)$, and $\grad_2(\mc G)$. Consequently, the algorithm can be applied on any graph class where these grads are finite constants, for example on the class of subgraphs of cliques with every edge subdivided $5$ times; note that this class is actually somewhere dense.

Secondly, we would like to point out that the algorithm in fact does not necessarily need to have an a priori knowledge of the values of $\grad_{p'}(\mathcal{G})$ for $0\leq p'\leq p$. In fact, the algorithm can be run with hypothetical upper bounds on the grads (say, on $\grad_p(\mc G)$), and it will either succeed in finding a correct kernel, or it will find a proof that the actual value of the grads is larger than assumed. Indeed, the crucial exchange argument in the proof of Lemma~\ref{lem:irrelevant} only relies on a comparison of the number of vertices in $\kappa$ with the size of its $3r$-projection on $X$.  Hence, whenever this comparison reveals that any member of~$\kappa$ is an irrelevant dominatee, this conclusion is drawn independently of the actual value of the grads, and hence is always correct. As a result, the algorithm can be run with larger and larger hypothetical bounds up to the point when a kernel is constructed. Therefore, after easy modifications the algorithm can be applied to basically {\em{any}} graph in hope of finding a reasonable kernel, and our analysis only shows guarantees on the output size in terms of the graph's densest $p$-shallow minor (or $2$-shallow minor, in case of {\sc{Dominating Set}}).

Thirdly, whereas the constant in the kernel size may seem impractical, we would like to point out that it provides a major improvement over the previous works. The kernels for apex-minor-free, $H$-minor-free, and $H$-topological-minor-free graphs of~\cite{FominLST10,FominLST12,FominLST13} are based on arguments originating in bidimensionality theory, graph minors, and finite-state properties of {\sc{Dominating Set}} and~{\sc{$r$-Dominating Set}}. Therefore, the dependence of the constant in the kernel size on the size of~$H$ is very difficult to trace. Even very crude estimations show that it is multiple-exponential, however still elementary. A careful analysis of our algorithm reveals that the kernel given by Theorem~\ref{thm:main-be} for {\sc{Dominating Set}} has size $2^{\textrm{poly}(\grad_2(\mc G))}\cdot k$, whereas for $\mathcal{G}$ being the class of $H$-minor-free graphs we have that $\grad_2(\mathcal{G})=\Oh(|V(H)|\cdot \sqrt{\log |V(H)|})$~\cite[Lemma 4.1]{Sparsity}. Thus, the constants obtained using our general technique are not only explicit, but also much lower than the ones obtained earlier.

\onlyfullversion{
\begin{figure}[h]
  \begin{center}
  \hspace*{.5cm}\includegraphics[scale=.9]{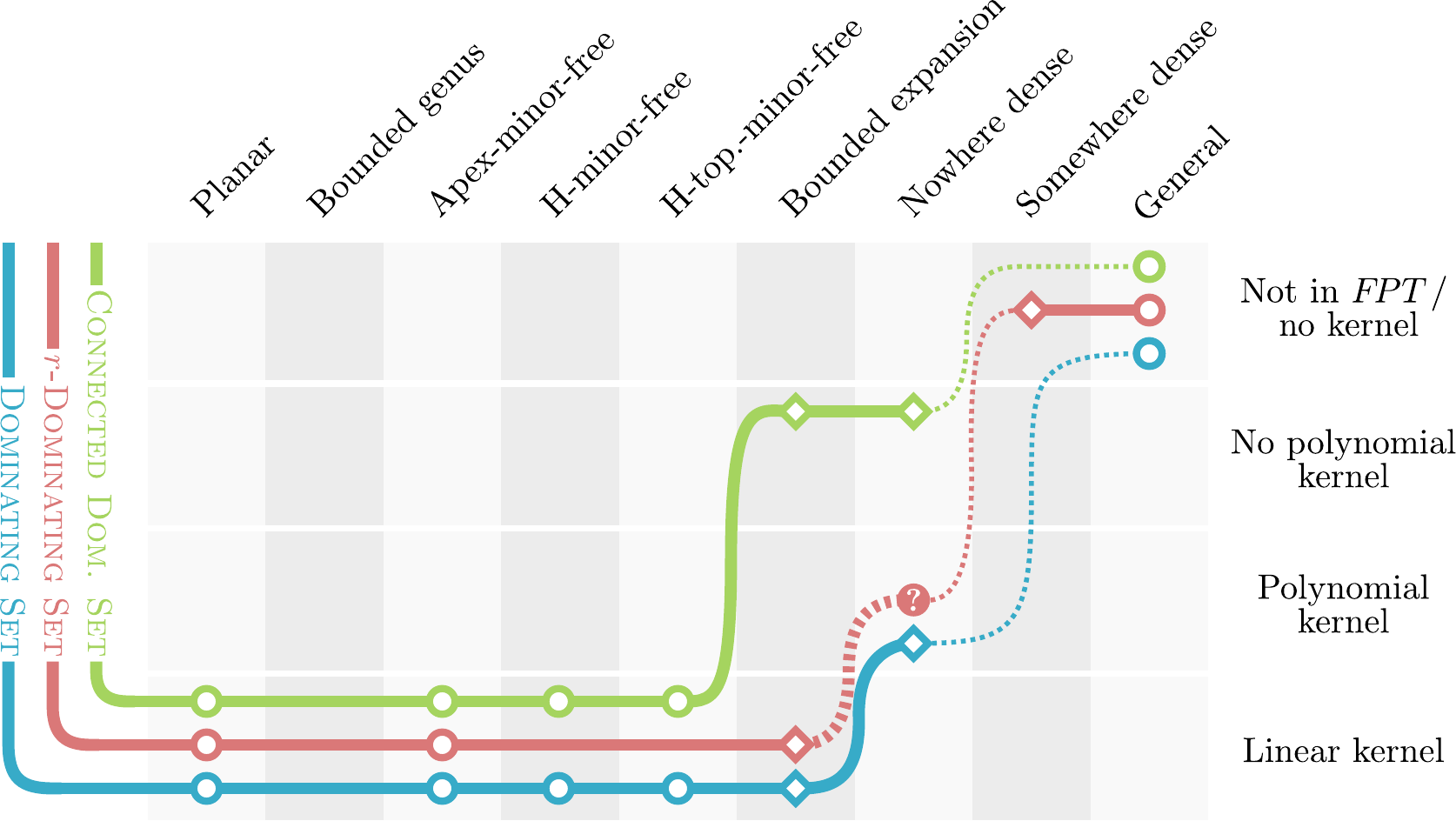}
  \end{center}
  \caption{\label{fig:kernelmapnce} An overview over results contained in past publications (circles) and this
  paper (diamonds) for \textsc{Dominating Set}, \textsc{$r$-Dominating Set} and \textsc{Connected Dominating Set}.
  The dashed line and question mark denote conjectured statements. The dotted lines represent the (unclear) transition of the complexity between nowhere dense graph classes and general graphs through larger and larger classes of somewhere dense graphs.
  }
\end{figure}
}

We conclude by stating some open questions and prospects for future work.

In order to make the algorithm more practical it is necessary to implement it in time linear in the size of the graph. In the current presentation we have not estimated the exact running times of the kernelization procedures; however, they are at least quadratic due to removing vertices from the domination core one by one. We expect that with more technical insight, the irrelevant dominatees can be removed in larger portions, which would lead to linear running time. However, we wanted to keep the current presentation as simple as possible, and hence we deferred optimizing the running time to future work.

From the theoretical point of view, the most important question left is understanding the kernelization complexity of {\sc{$r$-Dominating Set}} on nowhere dense graph classes. So far we know that this problem admits a linear kernel on any class of bounded expansion, for each $r$, whereas on any somewhere dense class closed under taking subgraphs, for some $r$ it is $\W{2}$-hard. Our approach for bounded expansion graph classes fails to generalize to nowhere dense classes mostly because of technical reasons. We believe that, in fact, for any nowhere dense class $\mc G$ and any positive integer $r$, {\sc{$r$-Dominating Set}} has an almost linear kernel on $\mc G$. Together with the lower bound of Theorem~\ref{thm:w2-hardness}, this would confirm the following dichotomy conjecture that we pose:

\begin{conjecture}
Let $\mc G$ be a graph class closed under taking subgraphs. Then:
\begin{itemize}
\item If $\mc G$ is nowhere dense, then for every integer $r\geq 1$ and real $\eps>0$, {\sc{$r$-Dominating Set}} admits an $\Oh(k^{1+\eps})$ kernel on $\mc G$.
\item If $\mc G$ is somewhere dense, then {\sc{$r$-Dominating Set}} is $\W{2}$-hard on $\mc G$ for some integer $r\geq 1$.
\end{itemize}
\end{conjecture}

% Redefine conference citation, if needed...
\def\redefineme{
    \def\LNCS{LNCS}%
    \def\ICALP##1{Proc. of ##1 ICALP}%
    \def\FOCS##1{Proc. of ##1 FOCS}%
    \def\COCOON##1{Proc. of ##1 COCOON}%
    \def\SODA##1{Proc. of ##1 SODA}%
    \def\SWAT##1{Proc. of ##1 SWAT}%
    \def\IWPEC##1{Proc. of ##1 IWPEC}%
    \def\IWOCA##1{Proc. of ##1 IWOCA}%
    \def\ISAAC##1{Proc. of ##1 ISAAC}%
    \def\STACS##1{Proc. of ##1 STACS}%
    \def\IWOCA##1{Proc. of ##1 IWOCA}%
    \def\ESA##1{Proc. of ##1 ESA}%
    \def\WG##1{Proc. of ##1 WG}%
    \def\LIPIcs##1{LIPIcs}%
    \def\LIPIcs{LIPIcs}%
    \def\LICS##1{Proc. of ##1 LICS}%
}

\bibliographystyle{abbrv}
\bibliography{./biblio}

\pagebreak

%\section*{Appendix}\label{sec:appendix}
\appendix

\section{Omitted simple proofs from Section~\ref{sec:preliminaries}}\label{sec:omitted}

\pendant*
\begin{proof}
Suppose $\grad_r(G')\geq 1$, and let $M'$ be an $r$-shallow minor of $G'$ with $\den(M')=\grad_r(G')$. If $\{v'\}$ is not one of the branch sets of the model of $M'$ in $G'$, then by removing $v'$ from the branch set of the model it belongs to (if it belongs to any) we obtain an $r$-shallow minor model of $M'$ in $G$, which implies that $\grad_r(G)\geq \den(M')=\grad_r(G')$. Otherwise, the vertex of which $\{v'\}$ is a branch set has degree $1$ in $M'$. Since $\den(M')\geq 1$, it follows that by removing this vertex we obtain a graph $M$ with $\den(M)\geq \den(M')$. At the same time $M$ is an $r$-shallow minor of $G$, which implies that $\grad_r(G)\geq \den(M')=\grad_r(G')$.
\end{proof}

\universalvertex*
\begin{proof}
  % THE PROOF COMMENTED OUT BELOW WAS REMOVED BY PAAL&MARKUS SINCE WE
  % HAD PROBLEMS UNDERSTANDING IT.
  
  Let $M'$ be an $r$-shallow
  minor of $G'$ with $\den(M') = \grad_r(G')$.  If the minor model
  % \todo{define minor model?} MP: yeah, for readers from a nursery
  of $M'$ in $G'$ does not contain the
  universal vertex, we have that $\grad_r(G') = \grad_r(G)$.  So
  suppose it contains the universal vertex.  Then, by using the same
  minor model in~$G$ but removing the branch set that contains the universal vertex, we obtain an $r$-shallow minor~$M$ of $G$ which lacks one vertex and
  at most $|M'|-1$ edges with respect to~$M'$.  Hence, we have the following:
\begin{align*}
\grad_r(G)\geq& \den(M) \geq \frac{||M'|| - |M'| + 1}{|M'| - 1} =\\
&= \frac{||M'||}{|M'|-1}-1 \geq \den(M')-1 = \grad_r(G')-1.
\end{align*}
\end{proof}

The topological grad $\topgrad_r(G)$ of a graph $G$ is defined
similarly to the grad, but we restrict ourselves to topological
$r$-shallow minors, i.e., we may only contract vertex disjoint paths
as follows: A shallow topological minor of a given graph $G$ at depth
$r$, for some half-integral $r$, is a graph $H$ obtained from $G$ by
taking a subgraph and then contracting internally vertex disjoint
paths of length at most $2r + 1$ to edges.  We denote the set of
$r$-shallow topological minors of $G$ by $G \topnab r$.  Then the
definition of a topological grad follows:
\begin{definition}[Topological grad (top-grad)]
  Let $G$ be a graph and $r$ a half-integral.  Then we define the
  topological grad as
  \[
  \topgrad_r(G) = \max_{H \in G \topnab r} \den(H) .
  \]
\end{definition}

It is known that topological grads are comparable to normal ones; for the following inequalities see Corollary 4.1 of~\cite{Sparsity}:
\begin{equation}\label{eq:gradtopgrad}
    \topgrad_r(G) \leq \grad_r(G) \leq 4(4 \topgrad_r(G))^{(r+1)^2} 
\end{equation}

Using the notion of topological grad as a pivot parameter, we can now
prove Lemma~\ref{lem:lexprod}.

\lexprodprop*
\begin{proof}
The following inequality has been proven in \cite[Proposition 4.6]{Sparsity}:
$$\topgrad_r(G \lexprod K_c) \leq \max\{2r(c-1)+1,c^2\} \cdot \topgrad_r(G) + c - 1.$$
We remark that even though Proposition 4.6 of \cite{Sparsity} assumes that $c\geq 2$, the claim holds also for $c=1$. We now observe that
$$\max\{2r(c-1)+1,c^2\}\leq 2r(c-1)+1+c^2\leq 2c(r+c),$$
and hence
\begin{equation}\label{eq:topgrad-lexprod}
\topgrad_r(G \lexprod K_c) \leq 2c(r+c)\cdot \topgrad_r(G)+c
\end{equation}
By combining (\ref{eq:gradtopgrad}) and (\ref{eq:topgrad-lexprod}) we obtain
\begin{align*}
\grad_r(G \lexprod K_c) &\leq 4(4 \topgrad_r(G\lexprod K_c))^{(r+1)^2}\\
&\leq 4(8c(r+c)\cdot \topgrad_r(G)+4c)^{(r+1)^2}\\
&\leq 4(8c(r+c)\cdot \grad_r(G)+4c)^{(r+1)^2},
\end{align*}
as claimed.
\end{proof}

\dvorakfull*
\begin{proof}
The core argument of \Dvorak~\cite{Dvorak13} lies in the following statement; here $\alpha_m(G)$ denotes the maximum size of an $m$-scattered set in $G$.\footnote{The definitions of weak colorings, used in this proof, are provided in Section~\ref{sec:prelim-nd}.}
\begin{theorem}[Theorem 4 of \cite{Dvorak13}]\label{thm:dvorak-precise}
If $1\leq m\leq 2k+1$ and $G$ satisfies $\wcol_m(G)\leq c$, then $\ds(G)\leq c^2\alpha_m(G)$. Furthermore, if an ordering $\sigma$ of $V(G)$ such that $|B^\sigma_m(v)|<c$ for every $v\in V(G)$
is given, then a $k$-dominating set $D$ and an $m$-scattered set $A$ such that $|D|\leq c^2|A|$ can be found in $\Oh(c^2\cdot \max(k,m)\cdot |V(G)|)$ time.
\end{theorem}
If we set $k=r$ and $m=2r$, then the proof of Theorem~\ref{thm:dvorak} boils down to finding an ordering of $V(G)$ with a near-optimal $2r$-weak coloring number. As \Dvorak observes, this can be done using the notion of $m$-admissibility, which is a similar measure of orderings of $V(G)$ as weak colorings. In particular (see Lemma 5 in~\cite{Dvorak13} and the discussion after it), an ordering of $V(G)$ of $2r$-admissibility $c$ has weak coloring number at most $(c(c-1)^{2r-1}+1)^{2r}$. Also, as \Dvorak~\cite{Dvorak13}, argues $m$-admissibility admits a simple polynomial-time $m$-approximation algorithm. By applying this algorithm we can thus obtain an ordering of $V(G)$ with weak coloring number at most $(2rc(2rc-1)^{2r-1}+1)^{2r}$, where $c=\textrm{adm}_{2r}(G)$ is the optimum $2r$-admissibility of $G$.

We are left with bounding the $2r$-admissibility of a graph in terms of its grads. For this, we use a trivial bound $\textrm{adm}_{2r}(G)<\col_{2r}(G)$ (see Exercise 4.5 in~\cite{Sparsity}) and the bound
$$\col_{2r}(G)\leq 1+Q_r(\grad_r(G)),$$
for some polynomial $Q_r$, following from Theorem 7.11 in~\cite{Sparsity}. Thus $\textrm{adm}_{2r}(G)\leq P_r(\grad_r(G))$, and hence the approximation algorithm for $2r$-admissibility outputs an ordering with weak coloring number at most $R_r(\grad_r(G))$, for some polynomial $R_r$.
By applying the algorithm of Theorem~\ref{thm:dvorak-precise} we can either find a $2r$-scattered set $A$ of size at least $k+1$, or a dominating set $D$ of size at most $P_r(\grad_r(G))\cdot k$ for $P_r(x)=(R_r(x))^2$, as claimed.
\end{proof}

\begin{comment}
\begin{lemma}
  For any graph~$G$, it holds that $\wcol_r \leq \left( (2\grad_{(r-1)/2}(G))^{(2r)^{2r}} + 1) \right)^r$.
\end{lemma}
\begin{proof}
  Let $k \geq \grad_{(r-1)/2}(G)$. By Theorem~7.11 in~\cite{Sparsity}, we have the bound
  $\col_r \leq q_r + 1$ where $q_r$ is recursively defined via
  \begin{align*}
      q_1 &= 2k \\
      q_{i+1} &= 2k \cdot q^{2i^2}_i ,~\text{for}~i \geq 1.
  \end{align*}
  We prove by induction that $q_i \leq (2k)^{(2i)^{2i}}$. As this holds easily for $q_1$, we
  assume the statement holds for $q_i$ and show that
  \[
    q_{i+1} = 2k \cdot q^{2i^2}_i \leq 2k \cdot (2k)^{ (2(i-1))^{2(i-1)} \cdot 2i^2}
  \]
  Since
  \[
      (2(i-1))^{2(i-1)} \cdot 2i^2 + 1 
      \leq 2^{2(i-1) + 2} \cdot (i-1)^{2(i-1)} 
      \leq (2i)^{2i}
  \]
  we have that $q_{i+1} \leq (2k)^{(2(i+1))^{2(i+1)}}$. And therefore that
  \[
    \col_r(G) \leq \left( 2\grad_{(r-1)/2}(G) \right)^{(2r)^{2r}}   + 1
  \]
  Using the fact that $\wcol_r(G) \leq \col_r(G)^r$, the claim follows.
\end{proof}
\end{comment}

%%% Local Variables:
%%% TeX-command-default: "Make"
%%% mode: latex
%%% TeX-master: "main"
%%% End: 

\section{Proof of Lemma~\ref{lem:fewneighbourhoods}}\label{sec:fewneighbourhoods}

\newcommand{\Fam}{\mathcal{F}}
\newcommand{\coloring}{\phi}
\newcommand{\sig}{\sigma}
\newcommand{\type}{\tau}
\newcommand{\Sig}{\Sigma}
\newcommand{\nfun}{\rho}

\subsection{$p$-centered colorings}

We first recall the notion of a {\em{$p$-centered coloring}} that will be the crucial tool throughout the proof.

\begin{definition}
For a graph $G$ and positive integer $p$, a {\em{$p$-centered coloring}} of $G$ is a coloring of $V(G)$ with colors $[q]$, for some positive integer $q$, such that for every connected subgraph $H$ of $G$ that contains at most $p$ colors, there is a color that has exactly one vertex in $H$. By $\widetilde{\chi}_p(G)$ we denote the minimal number of colors $q$ needed for a $p$-centered coloring of $G$.
\end{definition}

It is known that graph classes of bounded expansion admit $p$-centered colorings with few colors.

\begin{theorem}[\cite{Sparsity}]\label{thm:centered-colorings}
Let $p$ be a positive integer. There exists a polynomial $R_p$, such that for every graph $G$ the following holds:
$$\widetilde{\chi}_p(G)\leq R_p(\grad_{2^{p-2}+1}(G)).$$
\end{theorem}

This statement follows from Theorem 7.8 of~\cite{Sparsity} and its proof. Theorem 7.8 of~\cite{Sparsity} states only a bound on the number of colors needed for a low tree-depth coloring, which is a slightly weaker notion. However, from the proof it directly follows that the constructed coloring is $p$-centered (see Lemma 7.8 of~\cite{Sparsity}). Also, the result about boundedness of $\widetilde{\chi}_p(G)$ for graphs belonging to classes of bounded expansion, though without the explicit bound stated, follows directly from Theorem~7.7 or Theorem 7.8 of~\cite{Sparsity}, and Theorem 2.5 of~\cite{NesetrilM08b}.

\subsection{Zigzag-free families of functions}

We first introduce some auxiliary combinatorial observations about families of functions that naturally appear in our proofs.

\begin{definition}
A family $\Fam$ of functions $[\ell]\to [n]$ is called {\em{zigzag-free}} if there are no three functions $f_1,f_2,g\in \Fam$ and elements $i,j\in [\ell]$ such that $f_1(i)=g(i)\neq f_2(i)$ and $f_1(j)\neq g(j)=f_2(j)$.
\end{definition}

\begin{lemma}\label{lem:zigzag-free}
Let $\Fam$ be a zigzag-free family of functions from $[\ell]$ to $[n]$. Then $|\Fam|\leq \ell n$.
\end{lemma}
\begin{proof}
We prove the following claim: for every function $g\in \Fam$ there exists a pair $(i,a)\in [\ell]\times [n]$ such that $g(i)=a$ and $g'(i)\neq a$ for every $g'\in \Fam$, $g'\neq g$. Since the pairs $(i,a)$ have to be pairwise different for all $g\in \Fam$, and there are at most $\ell n$ such pairs, the claim will follow.

Suppose for the sake of contradiction that there exists $g\in \Fam$ such that for every $i\in [\ell]$ there is a function $g'\in \Fam$, $g'\neq g$, with $g(i)=g'(i)$. Let $f_2$ be a function from $\Fam$ that agrees with $g$ on the maximum number of indices, among functions different from $g$. Since $f_2\neq g$, there is an index $i$ such that $g(i)\neq f_2(i)$. By our assumption, however, there is a function $f_1\in \Fam$, $f_1\neq g$, such that $g(i)=f_1(i)$. Observe that $f_1$ has to agree with $g$ on every index $j$ where $f_2$ agrees with $g$, since otherwise the triple $(f_1,f_2,g)$ would contradict the fact that $\Fam$ is zigzag-free. Therefore, $f_1$ is different from $g$ and agrees with $g$ on strictly more positions than $f_2$ does, which is a contradiction with the choice of $f_2$.
\end{proof}

\subsection{Layered graphs and signatures}

Throughout this section we will be working with a {\em{layered graph}}, defined as follows. A layered graph $G$ is a graph with the vertex set partitioned into $r+1$ layers $V_0,V_1,\ldots,V_r$, such that every edge of the graph connects a pair of vertices from two consecutive layers in this ordering. Note that, in this manner, every vertex of $V_r$ is at distance at least $r$ from every vertex of $V_0$. For $u\in V_0$, let $R(u)=N^r[u]\cap V_r$, i.e., $R(u)$ is the set of vertices of the last layer that can be reached by a path of length $r$ from $u$. Note that then such a path traverses vertices belonging to consecutive layers $V_0,V_1,\ldots,V_r$, in this order. Every $V_0$-$V_r$ path with this property shall be called a {\em{straight path}}.

From now on, suppose we have some $2(r+1)$-centered coloring $\coloring\colon V(G)\to [q]$ of $G$, that is, a coloring of vertices of $\coloring$ using color set $[q]$ such that for every connected subgraph $H$ of $G$ that contains at most $2(r+1)$ colors, there is a color $i\in [q]$ such that exactly one vertex of $H$ is colored with $i$. If $G$ belongs to a class $\mathcal{G}$ of bounded expansion, then by Theorem~\ref{thm:centered-colorings} we know that there exists a constant $q$ depending on $r$ and the expansion of $\mathcal{G}$ for which such a coloring of $G$ exists.

For a straight path $P$ connecting some $u\in V_0$ and $v\in V_r$, we define the {\em{signature}} of $P$, denoted $\sig(P)$, as the sequence of colors of consecutive vertices appearing on $P$. By $\Sig = [q]^{r+1}$ we denote the set of all possible signatures of straight paths. For $u\in V_0$ and $\alpha\in \Sig$, by $R^\alpha(u)$ we denote the set of all vertices of $V_r$ that are reachable from $u$ by a straight path with signature $\alpha$. We say that $\alpha$ is {\em{realizable at $u$}} if $R^\alpha(u)\neq \emptyset$. The {\em{type}} of vertex $u\in V_0$, denoted $\type(u)$, is defined as the set of signatures realizable at $u$. Note that there are $2^{|\Sig|}=2^{q^{r+1}}$ possible types of vertices in $V_0$.

We now describe the main result of this section, i.e., the upper bound on the number of possible $r$-neighborhoods of vertices of $V_0$ within $V_r$.

\begin{lemma}\label{lem:layered-nei}
$| \{A\subseteq V_r\, \colon\, \exists_{u\in V_0}\, A=R(u)\} | \leq 2^{q^{r+1}}\cdot q^{r+1}\cdot |V_r|.$
\end{lemma}
\begin{proof}
We focus on any subset of signatures $\Pi\subseteq \Sig$, and prove that the number of sets $R(u)$ that are induced by vertices of $u$ with type $\Pi$ is at most $q^{r+1}\cdot |V_r|$. Then the claim follows from summing through all the possible types. Let then $W=\type^{-1}(\Pi)$ be the set of all vertices of $V_0$ that have type $\Pi$.

\begin{claim}\label{cl:equal-or-disjoint}
For any $u,v\in W$ and any $\alpha\in \Pi$, sets $R^\alpha(u)$ and $R^\alpha(v)$ are either equal or disjoint.
\end{claim}
\begin{proof}
For the sake of contradiction, without loss of generality suppose that there are vertices $a,b\in V_r$ such that $a\in R^\alpha(u)\setminus R^\alpha(v)$ and $b\in R^\alpha(u)\cap R^\alpha(v)$. Let $P_1,P_2,P_3$ be straight paths with signature $\alpha$ that respectively connect $u$ with $a$, $u$ with $b$, and $v$ with $b$. 
\begin{center}
  \includegraphics[scale=.75]{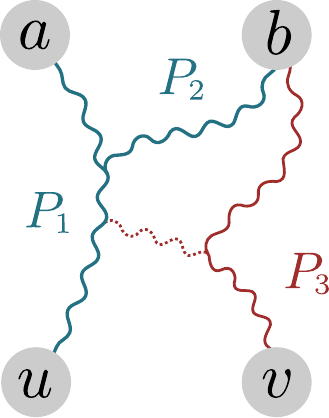}
\end{center}
Observe that $P_1$ and $P_3$ do not intersect: Indeed, if this was the case, then one could create a straight path of signature $\alpha$ connecting $v$ with $a$ by first traversing a prefix of $P_3$ up to the vertex of intersection, and then continuing using a suffix of $P_1$ up to $a$. This would contradict the assumption that $a\notin R^\alpha(v)$. Let $H$ be the subgraph obtained by taking the union of paths $P_1$, $P_2$, and $P_3$. Then $H$ is connected, vertices of $H$ are colored using only at most $r+1$ colors (the ones appearing in $\alpha$), and each such color is used at least twice: once on $P_1$ and once on $P_3$. This is a contradiction with $\coloring$ being a $2(r+1)$-centered coloring.
\cqed\end{proof}

From Claim~\ref{cl:equal-or-disjoint} we infer that for every $\alpha\in \Pi$ we can find a partition $(A^\alpha_1,A^\alpha_2,\ldots,A^\alpha_{n_\alpha})$ of a subset of $V_r$ into nonempty subsets such that for every $u\in W$, we have that $R^\alpha(u)=A^\alpha_j$ for some $j\in [n_\alpha]$. In particular, we have that $n_\alpha\leq n$, where $n:=|V_r|$. Let $\Pi=\{\alpha_1,\alpha_2,\ldots,\alpha_{|\Pi|}\}$. For a vertex $u\in W$, let us define a function $\nfun_u\colon [|\Pi|]\to [n]$ as follows: for $i\in [|\Pi|]$, $\nfun_u(i)$ is such an index $j$ that $R^{\alpha_i}(u)=A^{\alpha_i}_j$. Let also $\Fam$ be the set of all functions $\nfun_u$ for $u\in W$. 

\begin{claim}\label{cl:is-zigzag-free}
$\Fam$ is zigzag-free.
\end{claim}
\begin{proof}
For the sake of contradiction, suppose there are vertices $u_1,u_2,v\in V_0$ and signatures $\alpha,\beta\in \Pi$, such that $R^{\alpha}(u_1)=R^{\alpha}(v)\neq R^{\alpha}(u_2)$ and $R^{\beta}(u_1)\neq R^{\beta}(v)=R^{\beta}(u_2)$. Let us take arbitrary vertices $a\in R^{\alpha}(v)$, $a'\in R^{\alpha}(u_2)$, $b\in R^{\beta}(v)$, $b'\in R^{\beta}(u_1)$. We define six straight paths as follows:

\vspace*{.8em} 
\begin{minipage}[c]{.6\textwidth}
\begin{itemize}
\item $P_{u_1}^\alpha$ is of signature $\alpha$ and connects $u_1$ with $a$;
\item $P_{u_1}^\beta$ is of signature $\beta$ and connects $u_1$ with $b'$;
\item $P_{u_2}^\alpha$ is of signature $\alpha$ and connects $u_2$ with $a'$;
\item $P_{u_2}^\beta$ is of signature $\beta$ and connects $u_2$ with $b$;
\item $P_{v}^\alpha$ is of signature $\alpha$ and connects $v$ with $a$;
\item $P_{v}^\beta$ is of signature $\beta$ and connects $v$ with $b$.
\end{itemize}
\end{minipage}
\begin{minipage}[c]{.3\textwidth}
 \centering\includegraphics[scale=.75]{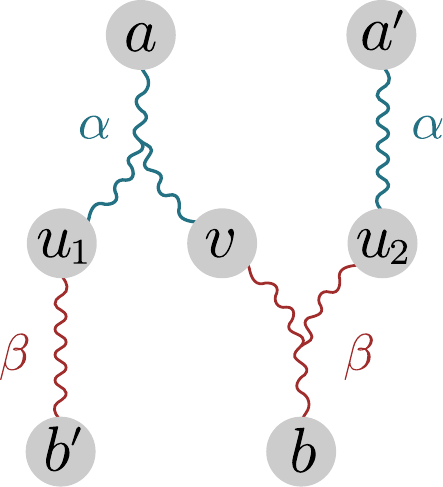}
\end{minipage}
\vspace*{1em}

\noindent
Let $H$ be the subgraph of $G$ obtained by taking the union of these paths; note that $H$ is connected. Now observe that paths $P_{u_1}^\beta$ and $P_{v}^\beta$ do not intersect. Indeed, if they intersected, then similarly as in the proof of Claim~\ref{cl:equal-or-disjoint} we would be able to find a straight path of signature $\beta$ connecting $v$ with $b'$ by concatenating a prefix of $P_{v}^\beta$ and a suffix of $P_{u_1}^\beta$. This, in turn, would imply that $b'\in R^\beta(v)$, which by Claim~\ref{cl:equal-or-disjoint} would mean that $R^{\beta}(u_1)=R^\beta(v)$, contradicting our assumptions. Similarly we prove that paths $P_{u_2}^\alpha$ and $P_{v}^\alpha$ do not intersect.

We conclude that $H$ is a connected subgraph that contains vertices of at most $2(r+1)$ colors (the ones used in signatures $\alpha$ and $\beta$), and each of these colors appears in $H$ at least twice: colors from $\alpha$ appear both on $P_{u_2}^\alpha$ and on $P_{v}^\alpha$, whereas colors from $\beta$ appear both on $P_{u_1}^\beta$ and on $P_{v}^\beta$. This is a contradiction with the assumption that $\coloring$ is a $2(r+1)$-centered coloring.
\cqed\end{proof}

Note that for two vertices $u,u'\in W$ with $\nfun_u=\nfun_{u'}$, we have that $R(u)=R(u')$. Therefore, the lemma immediately follows from combining Lemma~\ref{lem:zigzag-free} with Claim~\ref{cl:is-zigzag-free}.
\end{proof}

\subsection{Proof of Lemma~\ref{lem:fewneighbourhoods}}

Based on graph $G$ and a subset of vertices $X\subseteq V(G)$, we construct an auxiliary layered graph $G'$ with layers $V_0,V_1,\ldots,V_r$ as follows:
\begin{itemize}
\item $V_r$ is a copy of $X$ and $V_i$ is a copy of $V(G)\setminus X$, for each $i=0,1,\ldots,r-1$. The copy of a vertex $u\in V(G)$ in layer $i$ will be denoted by $u^i$.
\item For every $i=0,1,\ldots,r-1$, create edges between $V_i$ and $V_{i+1}$ as follows: for $u^i\in V_i$ and $v^{i+1}\in V_{i+1}$, add the edge $u^iv^{i+1}$ if and only if $u=v$ or $uv\in E(G)$.
\end{itemize}
It is now easy to see that for every $u\in V(G)\setminus X$, we have that $R(u^0)$ is exactly the set of copies of vertices of $\pr{r}{u}{X}$. Observe that $G'$ is a subgraph of the graph $G\lexprod K_{r+1}$. Hence, if $G\in \mathcal{G}$ for some class of bounded expansion $\mathcal{G}$, then, by Lemma~\ref{lem:lexprod}, $G'\in \mathcal{G}'$ for some class $\mathcal{G}'$ that also has bounded expansion, and where $\grad_i(\mathcal{G}')$ is bounded in terms of $r$ and $\grad_i(\mathcal{G})$, for every nonnegative integer $i$. Therefore, we can find a $2(r+1)$-centered coloring $\coloring$ of $G'$ that uses $q$ colors, where $q$ is a constant depending on $r$ and the grads of $\mc G$. Now Lemma~\ref{lem:fewneighbourhoods} follows from a direct application of Lemma~\ref{lem:layered-nei}.

\end{document}